\documentclass[onefignum,onetabnum]{siamonline190516} 

\usepackage{amsfonts}
\usepackage{amsmath}
\usepackage{graphicx}
\usepackage{makecell} 
\usepackage{epstopdf}
\usepackage{algorithmic}
\ifpdf
  \DeclareGraphicsExtensions{.eps,.pdf,.png,.jpg}
\else
  \DeclareGraphicsExtensions{.eps}
\fi

\usepackage{cite} 

\usepackage{enumitem}
\setlist[enumerate]{leftmargin=.5in}
\setlist[itemize]{leftmargin=.5in}

\newcommand{\red}[1]{\textcolor{black}{#1}} 

\newsiamremark{remark}{Remark}
\newsiamremark{hypothesis}{Hypothesis}
\crefname{hypothesis}{Hypothesis}{Hypotheses}
\newsiamthm{claim}{Claim}

\headers{Superposition of hexagonal lattices}{G. Iooss and A. M. Rucklidge}

\title{Patterns and quasipatterns from the superposition of two hexagonal lattices\thanks{Submitted to the editors DATE.
\funding{This work was supported by the EPSRC
(\hbox{EP/P015611/1}, AMR) and by the Leverhulme Trust (\hbox{RF-2018-449/9}, AMR)}}}

\author{G\'erard Iooss\thanks{I.U.F., Universit\'e C\^ote d'Azur, CNRS,
Laboratoire J. A. Dieudonn\'{e}, Parc Valrose F-06108, Nice Cedex 2, France
  (\email{gerard.iooss@unice.fr}).}
\and Alastair M. Rucklidge\thanks{School of Mathematics,
 University of Leeds, Leeds LS2 9JT, UK
  (\email{A.M.Rucklidge@leeds.ac.uk}).}}

\usepackage{amsopn}

\usepackage{tikz}   
\usetikzlibrary{arrows}  

\pgfrealjobname{article} 

\newcommand{\bx}{{\mathbf{x}}}
\newcommand{\bk}{{\mathbf{k}}}
\newcommand{\bs}{{\mathbf{s}}}
\newcommand{\bmm}{{\mathbf{m}}}
\newcommand{\bz}{{\mathbf{z}}}
\newcommand{\bp}{{\mathbf{p}}}
\newcommand{\bpp}{{\mathbf{p'}}}
\newcommand{\bdelta}{{\mathbf{\delta}}}
                            
\newcommand{\Ep}{{\mathcal{E}_{p}}}
\newcommand{\Eqp}{{\mathcal{E}_{qp}}}
\newcommand{\EO}{{\mathcal{E}_{0}}}
\newcommand{\Ei}{{\mathcal{E}_{1}}}
\newcommand{\Eii}{{\mathcal{E}_{2}}}

\ifpdf
\hypersetup{
  pdftitle={Patterns and quasipatterns from the superposition of two hexagonal lattices},
  pdfauthor={G. Iooss and A. M. Rucklidge}
}
\fi

\begin{document}

\maketitle

\begin{abstract}
  \red{When two-dimensional pattern-forming problems are posed on a periodic
  domain, classical techniques (Lyapunov--Schmidt, equivariant bifurcation
  theory) give considerable information about what periodic patterns are formed
  in the transition where the featureless state loses stability. When the
  problem is posed on the whole plane, these periodic patterns are still
  present. Recent work on the Swift--Hohenberg equation (an archetypal
  pattern-forming partial differential equation) has proved the existence of
  quasipatterns, which are not spatially periodic and yet still have long-range
  order. Quasipatterns may have 8-fold, 10-fold, 12-fold and higher rotational
  symmetry, which preclude periodicity. There are also} quasipatterns with
  6-fold rotational symmetry made up from the superposition of two equal-amplitude
  hexagonal patterns rotated by \red{almost any} angle~$\alpha$
  with respect to each other. Here, we \red{revisit} the Swift--Hohenberg
  equation (with quadratic as well as cubic nonlinearities) and prove existence
  of several new quasipatterns. \red{The most surprising are \emph{hexa-rolls}:
  periodic and quasiperiodic patterns made from the superposition of}
  hexagons and rolls (stripes) oriented in almost any direction 
  \red{with respect to each other} and with any
  relative translation; \red{these bifurcate directly from the featureless 
  solution. In addition, we find} quasipatterns made from the superposition of
  hexagons with unequal amplitude (provided the coefficient of the quadratic
  nonlinearity is small). We consider the periodic case as well, and extend the
  class of known solutions, including the superposition of hexagons and rolls.
  \red{While we have focused on the Swift--Hohenberg equation, our work
 contributes to the general question of what periodic or quasiperiodic patterns
 should be found generically in pattern-forming problems on the plane.}
\end{abstract}

\begin{keywords}
  Quasipatterns, superlattice patterns, Swift--Hohenberg equation.
\end{keywords}

\begin{AMS}
  35B36, 37L10, 52C23
\end{AMS}


\section{Introduction}\label{sec:introduction}
Regular patterns are ubiquitous in nature, and carefully controlled laboratory
experiments are capable of producing patterns, in the form of rolls (stripes), squares or
hexagons, with an astonishingly high degree of symmetry. One particular example is the
Faraday wave experiment, in which a layer of viscous fluid is subjected to sinusoidal
vertical vibrations. Without the forcing, the surface of the fluid is flat and
featureless, but as the strength of the forcing increases \red{beyond a 
critical value}, the flat surface loses
stability to two-dimensional patterns of standing waves, which in simple cases take the
form of roll, square or hexagonal patterns~\cite{Arbell2002}. But, with more elaborate
forcing, more complex patterns can be found. \Cref{fig:fw_examples} shows examples of
(a,b)~superlattice patterns and (c,d)~quasipatterns~\cite{Kudrolli1998,Arbell2002}. The
images in (a,c) show the pattern of standing waves on the surface of the fluid, while
(b,d) show the Fourier power spectra. In both cases, the patterns are dominated by twelve
waves, indicated by twelve small circles in \Cref{fig:fw_examples}(b) and by twelve blobs
lying on a circle in \Cref{fig:fw_examples}(d). The distance from the origin to the
twelve peaks gives the wavenumber that dominates the pattern. In the superlattice
example, the twelve peaks are unevenly spaced, but the basic structure is still
hexagonal, and it is spatially periodic with a periodicity equal to $\sqrt{7}$~times the
wavelength of the instability~\cite{Kudrolli1998}. In the quasipattern example, spatial
periodicity has been lost. Instead, the quasipattern has (on average) twelve-fold
rotation symmetry, as seen in the repeating motif of twelve pentagons arranged in a
circle and in the twelve evenly spaced peaks in the Fourier power spectrum in
\cref{fig:fw_examples}(d). The lack of spatial periodicity is apparent in
\cref{fig:fw_examples}(c), while the point nature of the power spectrum in
\cref{fig:fw_examples}(d) indicates that the pattern has long-range order. These two
features, the lack of periodicity (implicit in this case from twelve-fold rotational
symmetry) and the presence of long-range order, are characteristics of quasicrystals in
metallic alloys~\cite{Shechtman1984a} and soft matter~\cite{Hayashida2007}, and in
quasipatterns in fluid dynamics~\cite{Edwards1994}, reaction--diffusion
systems~\cite{Castelino2020} and optical systems~\cite{Aumann2002}.

The discovery of twelve-fold quasipatterns in the Faraday wave
experiment~\cite{Edwards1994} inspired a sequence of papers investigating this
phenomenon~\cite{Muller1994, Zhang1996, Lifshitz1997, Porter2004, Skeldon2007,
Rucklidge2009, Rucklidge2012, Skeldon2015}. One of the main outcomes of this
body of work is an understanding of the mechanism for stabilizing quasipatterns
in Faraday waves. Twelve-fold quasicrystals have also been found in block
copolymer and dendrimer systems~\cite{Zeng2004, Hayashida2007}, in turn
inspiring a considerable volume of work~\cite{Archer2013, Achim2014,
Barkan2014, Jiang2015, Subramanian2016}. It turns out that the same
stabilization mechanism operates in the Faraday wave and the polymer
crystallization systems~\cite{Lifshitz2007a,Ratliff2019}. In both cases, and
indeed in other systems~\cite{Castelino2020, Gokce2020}, a common feature is
that a second unstable or weakly damped length scale plays a key role in
stabilizing the pattern. See~\cite{Savitz2018} for a recent review.

\begin{figure}[tbp]
  \centering
  \hbox to \hsize{%
      \hbox to 0.24\textwidth{\hfil(a)\hfil}\hfil
      \hbox to 0.24\textwidth{\hfil(b)\hfil}\hfil
      \hbox to 0.24\textwidth{\hfil(c)\hfil}\hfil
      \hbox to 0.24\textwidth{\hfil(d)\hfil}}
  \hbox to \hsize{%
      \includegraphics[width=0.24\textwidth]{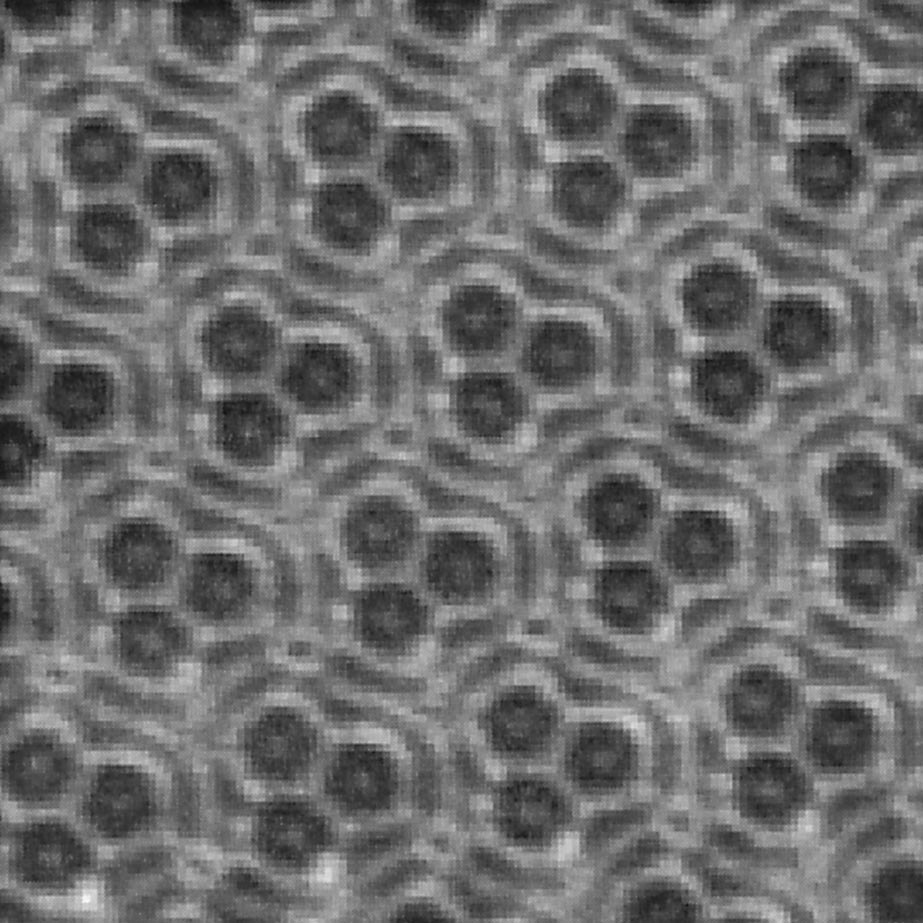}\hfil
      \includegraphics[width=0.24\textwidth]{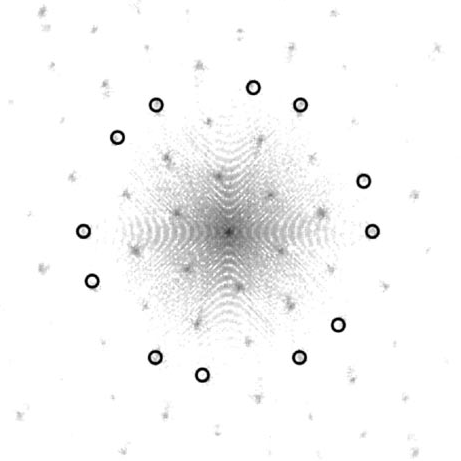}\hfil
      \includegraphics[width=0.24\textwidth]{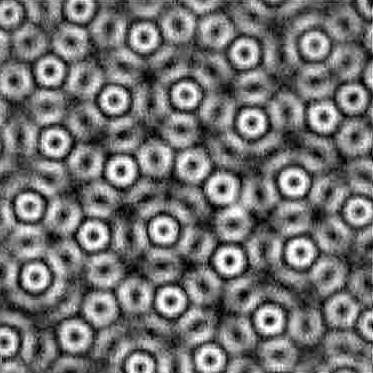}\hfil
      \includegraphics[width=0.24\textwidth]{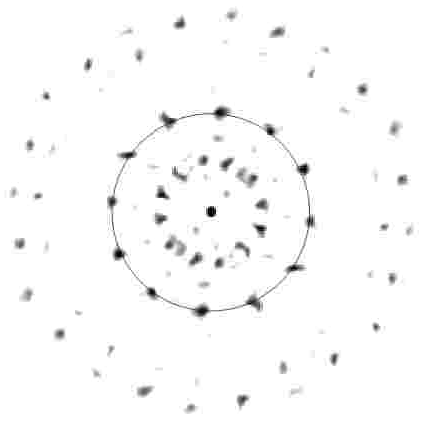}}
  \caption{Examples of (a,b)~superlattice patterns
 (reproduced with permission from~\cite{Kudrolli1998})
 and (c,d)~quasipatterns
 (reproduced with permission from~\cite{Arbell2002}).
 (a,c) show images representing the surface height of the fluid in Faraday wave
experiments, with thin layers of viscous liquids subjected to large-amplitude
multi-frequency forcing;
 (b,d) are Fourier power spectra of the images in (a,c), and indicate the
twelve peaks that dominate the patterns in each case.}
  \label{fig:fw_examples}
\end{figure}

However, as well the question of how superlattice patterns and quasipatterns
are stabilized, there is the question of their existence as solutions of
pattern-forming partial differential equations (PDEs) posed on the plane,
without lateral boundaries~\cite{Iooss2010,Argentina2012,Braaksma2017,Braaksma2019}. Superlattice patterns, which have spatial
periodicity (as in \cref{fig:fw_examples}a) can be analysed in finite domains with
periodic boundary conditions. 
In this case, and near the bifurcation point, spatially periodic patterns have
Fourier expansions with wave vectors that live on a lattice, and the
infinite-dimensional PDE can be reduced rigorously to a finite-dimensional set
of equations for the amplitudes of the primary modes~\cite{Carr1981,
Vanderbauwhede1992}. In the finite dimensional setting, amplitude equations
can be written down, bifurcating equilibrium points found and their stability
analysed~\cite{Dionne1997}. Equivariant bifurcation theory~\cite{Golubitsky1988} is a powerful tool that uses symmetry techniques to prove existence of certain classes of symmetric periodic patterns without recourse to amplitude equations.

But quasipatterns pose a particular challenge for proving existence, in that
the formal power series that describes small amplitude solutions may
diverge~\cite{Rucklidge2003,Iooss2010} owing to the appearance of small
divisors. Nonetheless, existence of quasipatterns with $Q$-fold rotation symmetry ($Q=8$, $10$, $12$, \dots) as
solutions of the steady Swift--Hohenberg equation (see below) has been proved
using methods based on the Nash--Moser theorem~\cite{Braaksma2017}. The same
approach has been applied to other pattern-forming PDEs, such as those for
steady B\'enard--Rayleigh convection~\cite{Braaksma2019}. Throughout, the
existence proofs show that as the amplitude of the quasipattern solution goes
to zero, the solution from the truncated formal expansion approaches a
quasipattern solution of the PDE in a union of disjoint parameter intervals,
going to full measure as the amplitude goes to zero.

\begin{figure}[tbp]
  \centering
  \hbox to \hsize{\hfil
\mbox{\beginpgfgraphicnamed{article_fig_wavevectors_a}%
\begin{tikzpicture}[scale=1.50,>=stealth]
   \draw[->] (-1.3,0) -- (1.3,0) node[right] {$k_x$};
   \draw[->] (0,-1.35) -- (0,1.2) node[above] {$k_y$};
   \draw[very thick] (0,0) circle (1.0);
   \draw[very thick,->] (0,0) --  ( -0.981981, -0.188982);
   \draw[very thick,->] (0,0) --  ( -0.981981,  0.188982);
   \draw[very thick,->] (0,0) --  ( -0.654654, -0.755929) node[below=5pt,  right=-12pt] {$\bk_3$};
   \draw[very thick,->] (0,0) --  ( -0.654654,  0.755929) node[above=8pt,  right=-12pt] {$\bk_5$};
   \draw[very thick,->] (0,0) --  ( -0.327327, -0.944911) node[below=5pt,  right=-12pt] {$\bk_6$};
   \draw[very thick,->] (0,0) --  ( -0.327327,  0.944911) node[above=8pt,  right=-11pt] {$\bk_2$};
   \draw[very thick,->] (0,0) --  (  0.327327, -0.944911);
   \draw[very thick,->] (0,0) --  (  0.327327,  0.944911);
   \draw[very thick,->] (0,0) --  (  0.654654, -0.755929);
   \draw[very thick,->] (0,0) --  (  0.654654,  0.755929);
   \draw[very thick,->] (0,0) --  (  0.981981, -0.188982) node[above=-1pt, right=-1pt] {$\bk_1$};
   \draw[very thick,->] (0,0) --  (  0.981981,  0.188982) node[above= 0pt, right=-1pt] {$\bk_4$};
   \draw (0.294594,   -0.0566946) arc (-10.893376:10.893376:0.3);
   \draw (0.300000,0.000000) -- (1.30000,0.700000) node[above=0pt, right=-2pt] {$\alpha$};
   \fill[gray] (   -0.981981,     -0.944911)  circle (0.03);
   \fill[gray] (   -0.981981,     -0.566947)  circle (0.03);
   \fill[gray] (   -0.981981,     -0.188982)  circle (0.03);
   \fill[gray] (   -0.981981,      0.188982)  circle (0.03);
   \fill[gray] (   -0.981981,      0.566947)  circle (0.03);
   \fill[gray] (   -0.981981,      0.944911)  circle (0.03);
   \fill[gray] (   -0.654654,      -1.13389)  circle (0.03);
   \fill[gray] (   -0.654654,     -0.755929)  circle (0.03);
   \fill[gray] (   -0.654654,     -0.377964)  circle (0.03);
   \fill[gray] (   -0.654654,       0.00000)  circle (0.03);
   \fill[gray] (   -0.654654,      0.377964)  circle (0.03);
   \fill[gray] (   -0.654654,      0.755929)  circle (0.03);
   \fill[gray] (   -0.654654,       1.13389)  circle (0.03);
   \fill[gray] (   -0.327327,     -0.944911)  circle (0.03);
   \fill[gray] (   -0.327327,     -0.566947)  circle (0.03);
   \fill[black] (  -0.327327,     -0.188982)  circle (0.04);
   \fill[black] (  -0.327327,      0.188982)  circle (0.04);
   \fill[gray] (   -0.327327,      0.566947)  circle (0.03);
   \fill[gray] (   -0.327327,      0.944911)  circle (0.03);
   \fill[gray] (     0.00000,      -1.13389)  circle (0.03);
   \fill[gray] (     0.00000,     -0.755929)  circle (0.03);
   \fill[black] (    0.00000,     -0.377964)  circle (0.04);
   \fill[gray] (     0.00000,       0.00000)  circle (0.03);
   \fill[black] (    0.00000,      0.377964)  circle (0.04);
   \fill[gray] (     0.00000,      0.755929)  circle (0.03);
   \fill[gray] (     0.00000,       1.13389)  circle (0.03);
   \fill[gray] (    0.327327,     -0.944911)  circle (0.03);
   \fill[gray] (    0.327327,     -0.566947)  circle (0.03);
   \fill[black] (   0.327327,     -0.188982)  circle (0.04);
   \fill[black] (   0.327327,      0.188982)  circle (0.04);
   \fill[gray] (    0.327327,      0.566947)  circle (0.03);
   \fill[gray] (    0.327327,      0.944911)  circle (0.03);
   \fill[gray] (    0.654654,      -1.13389)  circle (0.03);
   \fill[gray] (    0.654654,     -0.755929)  circle (0.03);
   \fill[gray] (    0.654654,     -0.377964)  circle (0.03);
   \fill[gray] (    0.654654,       0.00000)  circle (0.03);
   \fill[gray] (    0.654654,      0.377964)  circle (0.03);
   \fill[gray] (    0.654654,      0.755929)  circle (0.03);
   \fill[gray] (    0.654654,       1.13389)  circle (0.03);
   \fill[gray] (    0.981981,     -0.944911)  circle (0.03);
   \fill[gray] (    0.981981,     -0.566947)  circle (0.03);
   \fill[gray] (    0.981981,     -0.188982)  circle (0.03);
   \fill[gray] (    0.981981,      0.188982)  circle (0.03);
   \fill[gray] (    0.981981,      0.566947)  circle (0.03);
   \fill[gray] (    0.981981,      0.944911)  circle (0.03);
   \draw[very thick,->] (0,0) --  (0.00000, -0.377964) node[below=0pt] {$\bs_2$};
   \draw[very thick,->] (0,0) --  (0.327327, 0.188982) node[above=4pt, right=-3pt] {$\bs_1$};
 \draw (-1.15, 1.15) node[above] {\strut(a)};
\end{tikzpicture}\endpgfgraphicnamed}\hfil%
\mbox{\beginpgfgraphicnamed{article_fig_wavevectors_b}%
\begin{tikzpicture}[scale=1.50,>=stealth]
   \draw[->] (-1.3,0) -- (1.3,0) node[right] {$k_x$};
   \draw[->] (0,-1.35) -- (0,1.2) node[above] {$k_y$};
   \draw[very thick] (0,0) circle (1.0);
 \draw[very thick,->] (0,0) -- ( -0.9659, -0.2588);                                        
 \draw[very thick,->] (0,0) -- ( -0.9659,  0.2588);                                        
 \draw[very thick,->] (0,0) -- ( -0.7071, -0.7071) node[below=5pt,  right=-12pt] {$\bk_3$};
 \draw[very thick,->] (0,0) -- ( -0.7071,  0.7071) node[above=8pt,  right=-12pt] {$\bk_5$};
 \draw[very thick,->] (0,0) -- ( -0.2588, -0.9659) node[below=5pt,  right=-12pt] {$\bk_6$};
 \draw[very thick,->] (0,0) -- ( -0.2588,  0.9659) node[above=8pt,  right=-11pt] {$\bk_2$};
 \draw[very thick,->] (0,0) -- (  0.2588, -0.9659);                                        
 \draw[very thick,->] (0,0) -- (  0.2588,  0.9659);                                        
 \draw[very thick,->] (0,0) -- (  0.7071, -0.7071);                                        
 \draw[very thick,->] (0,0) -- (  0.7071,  0.7071);                                        
 \draw[very thick,->] (0,0) -- (  0.9659, -0.2588) node[above=-1pt, right=-1pt] {$\bk_1$}; 
 \draw[very thick,->] (0,0) -- (  0.9659,  0.2588) node[above= 0pt, right=-1pt] {$\bk_4$}; 
   \draw (0.289779, -0.0776423) arc (-15:15:0.3);
   \draw (0.30000,0.000000) -- (1.30000,0.700000) node[above=0pt, right=-2pt] {$\alpha$};
 \draw (-1.15, 1.15) node[above] {\strut(b)};
\end{tikzpicture}\endpgfgraphicnamed}\hfil%
\mbox{\beginpgfgraphicnamed{article_fig_wavevectors_c}%
\begin{tikzpicture}[scale=1.50,>=stealth]
   \draw[->] (-1.3,0) -- (1.3,0) node[right] {$k_x$};
   \draw[->] (0,-1.35) -- (0,1.2) node[above] {$k_y$};
   \draw[very thick] (0,0) circle (1.0);
 \draw[very thick,->] (0,0) -- ( -0.97503534, -0.22204975);                                        
 \draw[very thick,->] (0,0) -- ( -0.97503534,  0.22204975);                                        
 \draw[very thick,->] (0,0) -- ( -0.67981839, -0.73338050) node[below=5pt,  right=-12pt] {$\bk_3$};
 \draw[very thick,->] (0,0) -- ( -0.67981839,  0.73338050) node[above=8pt,  right=-12pt] {$\bk_5$};
 \draw[very thick,->] (0,0) -- ( -0.29521695, -0.95543025) node[below=5pt,  right=-12pt] {$\bk_6$};
 \draw[very thick,->] (0,0) -- ( -0.29521695,  0.95543025) node[above=8pt,  right=-12pt] {$\bk_2$};
 \draw[very thick,->] (0,0) -- (  0.29521695, -0.95543025);                                        
 \draw[very thick,->] (0,0) -- (  0.29521695,  0.95543025);                                        
 \draw[very thick,->] (0,0) -- (  0.67981839, -0.73338050);                                        
 \draw[very thick,->] (0,0) -- (  0.67981839,  0.73338050);                                        
 \draw[very thick,->] (0,0) -- (  0.97503534, -0.22204975) node[above=-1pt, right=-1pt] {$\bk_1$}; 
 \draw[very thick,->] (0,0) -- (  0.97503534,  0.22204975) node[above= 0pt, right=-1pt] {$\bk_4$}; 
   \draw (0.292511,   -0.0666149) arc (-12.83:12.83:0.3);
   \draw (0.30000,0.000000) -- (1.30000,0.700000) node[above=0pt, right=-2pt] {$\alpha$};
 \draw (-1.15, 1.15) node[above] {\strut(c)};
\end{tikzpicture}\endpgfgraphicnamed}%
\hfil}

\caption{(a)~Two sets of six equally spaced wave vectors ($\bk_1$, $\bk_2$,
$\bk_3$ and their opposites, and $\bk_4$, $\bk_5$, $\bk_6$ and their opposites)
rotated an angle~$\alpha$ with respect to each other so as to produce spatially
periodic patterns: $\alpha\approx21.79^\circ$, with $\cos\alpha=\frac{13}{14}$
and $\sqrt{3}\sin\alpha=\frac{9}{14}$. The gray dots indicate that the twelve
vectors lie on an underlying hexagonal lattice, generated by the vectors~$\bs_1$ and 
$\bs_2$. Compare with \cref{fig:fw_examples}(b).
 (b)~12-fold quasipatterns are generated by twelve equally spaced vectors:
$\alpha=\frac{\pi}{6}=30^\circ$, with $\cos\alpha=\frac{1}{2}\sqrt{3}$. Compare
with \cref{fig:fw_examples}(d).
 (c)~6-fold quasiperiodic case: $\alpha\approx25.66^\circ$, with
$\cos\alpha=\frac{1}{4}\sqrt{13}$ and
$\sqrt{3}\sin\alpha=\frac{3}{4}$. Quasipatterns generated by equal
combinations of the twelve waves have six-fold rotation symmetry but lack spatial
periodicity.}

\label{fig:wavevectors}

\end{figure}
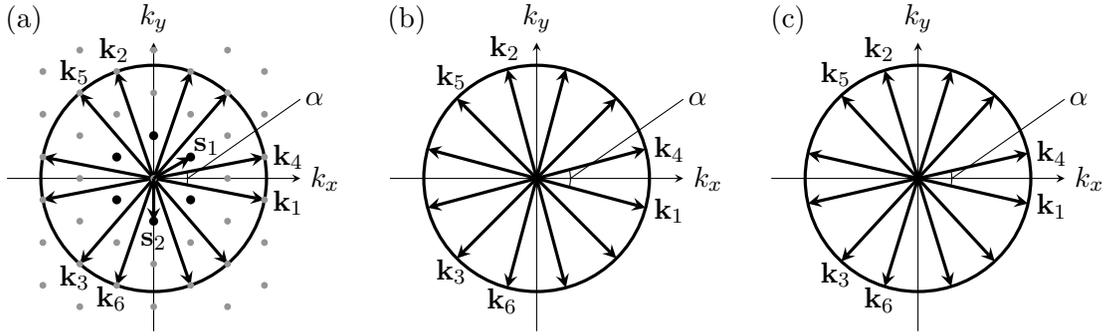

Most previous work on quasipatterns has concentrated on Fourier spectra that
exhibit ``prohibited'' symmetries: eight-, ten-, twelve-fold and higher
rotation symmetries, as in \cref{fig:fw_examples}(c), or icosahedral symmetry
in three dimensions~\cite{Subramanian2016}. There is, however, a class of quasipatterns
with six-fold rotation symmetry, related to the superlattice patterns already discussed.
These patterns can be described in terms of the superposition of twelve waves with twelve
wavevectors, grouped into two sets of six as in \cref{fig:wavevectors}, with the six
vectors within each set spaced evenly around the circle, and with the two sets rotated by
an angle~$\alpha$ with respect to each other, with $0<\alpha<\frac{\pi}{3}$. 
In the quasiperiodic case, we can choose $\alpha$ to be the smallest angle
between the vectors, so $0<\alpha\leq\frac{\pi}{6}$.

The discovery, in the Faraday wave experiment and elsewhere, of these elaborate
superlattice patterns and quasipatterns, with and without spatial periodicity, motivated
investigations into the bifurcation structure of pattern formation problems posed both in
periodic domains and on the whole plane, without lateral boundaries. We focus on an
example of such a problem, the \red{steady} Swift--Hohenberg equation, which is:
 \begin{equation}
 (1+\Delta)^{2}u-\mu u+\chi u^{2}+u^{3}=0,  \label{eq:SHeq}
 \end{equation}
where $u(\bx)$ is a real function of $\bx=(x,y)\in\mathbb{R}^{2}$, $\Delta$~is
the Laplace operator, $\mu$~is a real bifurcation parameter and $\chi$~is a
real parameter. The time-dependent version of this PDE was proposed originally
as a model of small-amplitude fluctuations near the onset of
convection~\cite{Swift1977}, but is now considered an archetypal model of
pattern formation~\cite{Hoyle2006}.

The trivial state $u=0$ is always a solution of~\cref{eq:SHeq}, and as~$\mu$ increases
through zero, many branches of small-amplitude solutions of~\cref{eq:SHeq} are created.
These include periodic patterns such as rolls, squares, hexagons and superlattice
patterns, quasipatterns with the prohibited rotation symmetries of eight-, ten-,
twelve-fold and higher (proved in~\cite{Braaksma2017} with $\chi=0$), as well as (again with $\chi=0$) two
families of six-fold quasipatterns with equal sums of the twelve Fourier modes illustrated
in \cref{fig:wavevectors}(c)~\cite{Iooss2019,Fauve2019}. In this paper, we extend the
analysis in~\cite{Iooss2019} by allowing $\chi\neq0$ and including quasipatterns with
unequal combinations of the twelve Fourier modes, discovering several new classes of solutions.

We approach this problem by deriving nonlinear amplitude equations for the twelve Fourier modes on the unit circle.
One important requirement on the twelve selected modes illustrated in \cref{fig:wavevectors} is therefore that nonlinear combinations of these modes should generate no further modes with wavevectors on the unit circle. If they did, additional amplitude equations would have to be included, a problem we leave for another day. We call the \red{(full measure, as proved in~\cite{Iooss2019} in Lemma~5)} set of~$\alpha$'s that satisfy this condition~$\EO$, defined more precisely in~\cite{Iooss2019} and \red{in \cref{def:E0} below}. Throughout, we use the names of the sets of values of~$\alpha$ from~\cite{Iooss2019}.

There are three possible situations as $\alpha$~is varied: the (zero measure)
periodic case, the (full measure) quasiperiodic case where the results
of~\cite{Iooss2019} can be used, and other quasiperiodic values of~$\alpha$
(zero measure). 
\red{See the definitions below and in \cref{app:definitions_of_all_sets} for more detail.}

\begin{enumerate}

\item
The \red{lattice} is {\emph{periodic}}, and $\alpha\in\Ep$, as in
\cref{fig:wavevectors}(a) \red{(see \cref{def:Ep})}. For these
angles, restricted to $0<\alpha<\frac{\pi}{3}$, both $\cos\alpha$ and $\sqrt{3}\sin\alpha$ must be rational, and the
wave vectors generate a lattice (see \cref{def:Ep} and \cref{lem:Ep_ab} below). This is the case examined
by~\cite{Dionne1997}, and $\alpha\approx21.79^\circ$
($\cos\alpha=\frac{13}{14}$ and $\sqrt{3}\sin\alpha=\frac{9}{14}$) is an
example. For reasons explained below, for some values of~$\alpha\in\Ep$, is it
more convenient to consider $\frac{\pi}{3}-\alpha$ instead, relabelling the
vectors. This set is dense but of measure zero. Not all values of $\alpha\in\Ep$ are also in~$\EO$.

\item
The angle~$\alpha$ is not in~$\Ep$ but it
{\emph{satisfies all three of the requirements}} for the existence proofs
in~\cite{Iooss2019}. The first requirement is that $\alpha\in\EO$ \red{(see \cref{def:E0} below)}:
no integer combination
of the twelve vectors already chosen should lie on the unit circle apart from
the twelve. The second and third requirements are that the numbers $\cos\alpha$ and
$\sqrt{3}\sin\alpha$ should satisfy two ``good'' Diophantine properties. We
define $\Ei$ and $\Eii$ to be the set of such angles, restricted to
$0<\alpha\leq\frac{\pi}{6}$ \red{(see definitions in \cref{app:definitions_of_all_sets}). Then, the set~$\Eii$, which itself requires $\EO$ and~$\Ei$, is the set of angles that satisfy all three requirements}. All rational multiples of~$\pi$ (restricted to $0<\alpha\leq\frac{\pi}{6}$) are in~$\Eii$,  for example, $\alpha=\frac{\pi}{6}=30^\circ$ as in \cref{fig:wavevectors}(b). The angle $\alpha\approx25.66^\circ$ is another example, ($\cos\alpha=\frac{1}{4}\sqrt{13}$ and
$\sqrt{3}\sin\alpha=\frac{3}{4}$, see \cref{fig:wavevectors}(c) and \cref{app:two_examples}). This set is of full measure.

\item
The angle~$\alpha$, still restricted to $0<\alpha\leq\frac{\pi}{6}$, is not in $\Ep$ or~$\Eii$, and although patterns made from these modes may be quasiperiodic, the existence proofs based on the
approach of~\cite{Iooss2019} do not work, at least not without further extension. The angle $\alpha\approx26.44^\circ$
($\cos\alpha=\frac{1}{12}(5+\sqrt{33})$ and
$\sqrt{3}\sin\alpha=\frac{1}{12}(15-\sqrt{33})$) is an example (see \cref{app:two_examples}) since it is not in~$\EO$. This set is dense but of
measure zero.

\end{enumerate}

For~$\alpha\in\Ep\cap\EO$, the resulting superlattice patterns are
spatially periodic, and their bifurcation structure is determined at finite order when
the small amplitude pattern is expressed as a formal power series~\cite{Dionne1997}. The
wavevectors for these spatially periodic superlattice patterns lie on a finer hexagonal
lattice (as in \cref{fig:wavevectors}a).

We define $\Eqp$ to be the complement of~$\Ep$ restricted to $0<\alpha\leq\frac{\pi}{6}$. For $\alpha\in\Eqp$, linear combinations of waves are typically quasiperiodic, but only for $\alpha\in\Eii\subset\Eqp$ can the techniques of~\cite{Iooss2019} be used to prove existence of quasipatterns with these modes as nonlinear solutions of the PDE~\cref{eq:SHeq}.
For the special case $\alpha=\frac{\pi}{6}\in\Eii$, as
in \cref{fig:wavevectors}(b), the quasipattern has twelve-fold rotation symmetry, but
more generally, as in \cref{fig:wavevectors}(c), there can be six-fold rotation symmetry,
more usually associated with hexagons. 
The proof in~\cite{Iooss2019} makes use of the properties of~$\Eii$; at this time, no existence result is known about $\alpha\notin\Eii\cup\Ep$.

The periodic case has been analysed by~\cite{Dionne1997,Silber1998}. They write
the small-amplitude pattern $u(\bx)$ as the sum of six complex amplitudes $z_1$,
\dots, $z_6$ times the six waves $e^{i\bk_1\cdot\bx}$, \dots,
$e^{i\bk_6\cdot\bx}$:
 \begin{equation}
 u(\bx) = \sum_{j=1}^{6} z_j e^{i\bk_j\cdot\bx} + c.c. + \text{high-order terms},
 \label{eq:leading_order}
 \end{equation}
where $c.c.$ refers to the complex conjugate, and the six wavevectors $\bk_1$,
\dots, $\bk_6$ are as illustrated in \cref{fig:wavevectors}(a). They then
derive, using symmetry considerations, the amplitude equations:
 \begin{equation}
 \begin{split}
 0 &= z_1 f_1(u_1,\dots,u_6,q_1,q_4,{\bar q}_4) +
   {\bar z}_2 {\bar z}_3 f_2 (u_1,\dots,u_6,{\bar q}_1,q_4,{\bar q}_4) + {} \\
   & \qquad {} + \text{high-order resonant terms},
 \end{split}
 \label{eq:amplitude}
 \end{equation}
where $u_1=|z_1|^2$, \dots, $u_6=|z_6|^2$, $q_1=z_1z_2z_3$, and $q_4=z_4z_5z_6$. Here,
$f_1$ and $f_2$ are smooth functions of their nine arguments. Five additional equations
can be deduced from permutation symmetry. The high-order resonant terms, present only in
the periodic case, are at least fifth order polynomial functions of the six amplitudes
and their complex conjugates, and depend on the choice of~$\alpha\in\Ep$. Even without
the amplitude equations~\cref{eq:amplitude}, equivariant bifurcation theory can be
used~\cite{Golubitsky1988,Dionne1997} to deduce the existence of various hexagonal and
triangular superlattice patterns, and, within the amplitude equations, the stability of
these patterns can be computed. 

The approach we take does not use equivariant bifurcation theory. Instead, we derive amplitude equations of the form~\cref{eq:amplitude} in the quasiperiodic and periodic cases. In the quasiperiodic case, the equation is a formal power series, but in both cases, the cubic truncation of the first component of amplitude equations is of the form
 \begin{equation}
 0 = \mu z_1 - \alpha_0{\bar z}_2 {\bar z}_3 
    - z_1(\alpha_1u_1 + \alpha_2u_2 + \alpha_2u_3 + \alpha_4u_4 + \alpha_5u_5 + \alpha_6u_6),
 \label{eq:amplitudecubic}
 \end{equation}
where $\alpha_0$, \dots, $\alpha_6$ are coefficients that can be computed from the PDE~\cref{eq:SHeq}. We find small amplitude solutions of the cubic truncation~\cref{eq:amplitudecubic} then verify that these correspond to small amplitude solutions of the untruncated amplitude equations~\cref{eq:amplitude}. 
One remarkable
result is that the formal expansion in powers of the amplitude (and parameter $\chi$ in the cases when $\chi$ is close to 0) of the bifurcating
patterns is given at leading order by the same formulae in both the quasiperiodic and the periodic cases. 
From solutions of the amplitude equations, 
the mathematical
proof of existence of the periodic patterns is given by the classical
Lyapunov--Schmidt method, while for quasipatterns the proof follows the same
lines as in~\cite{Iooss2019}. The truncated expansion
of the formal power series provides the first approximation to the quasipattern solution,
which is a starting point for the Newton iteration process, using the Nash--Moser method
for dealing with the small divisor problem~\cite{Iooss2019} 
\red{(for more details see~\S\labelcref{sec:quasipat_case_high_orders})}.

\begin{table}
\begin{center}
\begin{tabular}{|l|llll|l|}
\hline
Name & \makecell[l]{Section\\Figure} 
     & \makecell[l]{Periodic\\or QP}
     & \makecell[l]{Example\\amplitudes}
     & $\chi$
     & \makecell[l]{Earlier\\results}\rule{0pt}{4.0ex}\\[7pt]    
\hline
\makecell[l]{QP-super-\\hexagons}
     & \makecell[l]{\S\labelcref{sec:quasipat_case_high_orders}\\Fig.~\labelcref{fig:example_twohex_patterns}}
     & QP 
     & $z_1=\dots=z_6\in\mathbb{R}$
     & Any
     & \cite{Iooss2019}\rule{0pt}{4.0ex}\\[12pt]    
\makecell[l]{Unequal QP-super-\\hexagons} 
     & \makecell[l]{\S\labelcref{sec:quasipat_case_high_orders}\\Fig.~\labelcref{fig:example_twohex_patterns}}
     & QP 
     & \makecell[l]{$z_1=z_2=z_3\neq{}$\\$z_4=z_5=z_6\in\mathbb{R}$}
     & $|\chi|\ll1$
     & New\\[12pt]
\makecell[l]{QP-anti-hexagons,\\QP-triangles etc.} 
     & \makecell[l]{\S\labelcref{sec:quasipat_case_high_orders}\\Fig.~\labelcref{fig:example_chizero_patterns}}
     & QP 
     & Various: see \cref{eq:antihexagonsetc}
     & $\chi=0$
     & New\\[12pt]
\makecell[l]{Super-\\hexagons}  
     & \makecell[l]{\S\labelcref{sec:periodic_case_high_orders}\\Fig.~\labelcref{fig:example_periodic_patterns}}
     & Periodic
     & $z_1=\dots=z_6\in\mathbb{R}$
     & Any
     & \cite{Dionne1997}\\[12pt]
\makecell[l]{Triangular\\superlattice}  
     & \makecell[l]{\S\labelcref{sec:periodic_case_high_orders}\\Fig.~\labelcref{fig:example_periodic_patterns}}
     & Periodic
     & \makecell[l]{Equal amplitudes\\$\mbox{Phases}\approx\frac{\pi}{3},\frac{2\pi}{3}$}
     & Any
     & \cite{Silber1998}\\[12pt]
\makecell[l]{Hexa-rolls\\(rolls dominant)}  
     & \makecell[l]{\S\labelcref{sec:hexarolls_rolls_dominant}\\Fig.~\labelcref{fig:example_hexaroll_patterns}}
     & \makecell[l]{QP and\\periodic}
     & \makecell[l]{\makecell[l]{$z_1\approx z_2\approx z_3\ll z_4$,\\$z_5=z_6=0$}}
     & \makecell[l]{$\chi$~neither\\too small nor\\too large}
     & New\\[12pt]
\makecell[l]{Hexa-rolls\\(balanced)}  
     & \makecell[l]{\S\labelcref{sec:hexarolls_balanced}\\Fig.~\labelcref{fig:example_hexaroll_patterns}}
     & \makecell[l]{QP and\\periodic}
     & \makecell[l]{\makecell[l]{$z_1\approx z_2\approx z_3\sim z_4$,\\$z_5=z_6=0$}}
     & $|\chi|\ll1$
     & New\\[8pt]
\hline
\end{tabular}
\end{center}
\vspace{0.6ex}
\caption{\noindent\red{Summary of the different solutions we consider.
``Periodic'' and ``QP'' refer to periodic ($\alpha\in\Ep\cap\EO$) and
quasiperiodic ($\alpha\in\Eii$) respectively. We give examples of the six~$z_j$
amplitudes as well as restrictions on the values of~$\chi$. The term
``super-hexagon'' refers to the superposition of two hexagonal patterns, which
can be equal or unequal amplitude. The last column gives references to relevant
earlier results or indicates whether the solutions are new.}}
 \label{tab:RoadMap} 
\end{table}

\red{We find several new types of solution, in the quasiperiodic and in the periodic
cases, and in the $\chi\neq0$ and $|\chi|\ll1$ cases. These are summarized in
\cref{tab:RoadMap}. The most significant new class of solutions is the
superposition of hexagons and roll patterns (\emph{hexa-rolls}), with the rolls 
arranged at almost any
orientation with respect to the hexagons ($\alpha\in(\Ep\cup\Eii)\cap\EO$) and
translated with respect to each other by arbitrary amounts. These bifurcate
directly from the featureless pattern even when~$\chi$ is not small (provided 
$\chi$ is not too large, see~\S\labelcref{sec:hexarolls_rolls_dominant}), in both
the periodic and the quasiperiodic cases. In the quasiperiodic case, the phason
symmetry~\cite{Echebarria2001} characteristic of quasipatterns leads to the
freedom to have arbitrary relative translations of the hexagons and rolls;
finding this same freedom in the periodic case was a surprise.}

\red{We also show that the particular example of periodic triangular
superlattice patterns reported experimentally in~\cite{Kudrolli1998} (see \cref{fig:fw_examples}a) and
explored theoretically in~\cite{Silber1998} can also be found in a much wider
class of periodic lattices. Moreover, for nearby angles~$\alpha\in\Eii$, we
find that the quasiperiodic super-hexagons can be thought of as long-range
modulations between the periodic super-hexagons and two types of periodic
superlattice triangles (see \cref{fig:example_periodic_patterns}).}

\red{Our work extends the periodic results of~\cite{Dionne1997} to the quasiperiodic
case, including quasiperiodic versions of the anti-hexagon, super-triangle and
anti-triangle patterns that occur with $\chi=0$. We also extend the previous
quasiperiodic work of~\cite{Fauve2019,Iooss2019}, which took $\chi=0$: we find
small-amplitude bifurcating solutions in~\cref{eq:amplitude} for any
$\chi\neq0$, including new quasiperiodic superposed hexagon patterns with
unequal amplitudes for $0<|\chi|\ll1$, and show that there are corresponding
quasiperiodic (and periodic) solutions of the Swift--Hohenberg equation.}

Amongst the solutions we find in the quasiperiodic
case are combinations of two hexagonal patterns, as well as the \red{hexa-roll
patterns mentioned above}.
In both the periodic and the quasiperiodic cases, the superposed hexagon and
roll patterns are new, and would not be found using the equivariant bifurcation
lemma as they have no symmetries (beyond periodic in that case). Also in both cases,
we consider the possibility that $\chi$ is small, and use the method
of~\cite{Iooss2019} on power series in two small parameters to find new
superposed hexagon patterns with unequal amplitudes, again out of range of the
equivariant bifurcation lemma.

We open the paper with a statement of the problem in \cref{sec:statement} and
develop the formal power series for the amplitude equations in
\cref{sec:formal}. We solve these equations in \cref{sec:solutions}, focusing
on the new solutions, and conclude in \cref{sec:conclusion}. Some details of
the definitions, examples and proofs are in the six appendices.

\section{Statement of the problem}\label{sec:statement}
We begin by explaining how we describe functions on lattices and quasilattices, and how the symmetries of the problem act on these functions.

\subsection{Lattices and quasilattices}

In the Fourier plane, we have two sets of six basic wave vectors
as illustrated in \cref{fig:wavevectors}:
 $\{\bk_{j},-\bk_{j}:j=1,2,3\}$ and
 $\{\bk_{j},-\bk_{j}:j=4,5,6\}$,
 both equally spaced on the unit circle, with angle $\frac{2\pi}{3}$
 between $\bk_{1}$, $\bk_{2}$ and~$\bk_{3}$ and
 between $\bk_{4}$, $\bk_{5}$ and~$\bk_{6}$,
 such that $\bk_{1}+\bk_{2}+\bk_{3}=0$ and
           $\bk_{4}+\bk_{5}+\bk_{6}=0$.
The two sets of six vectors are rotated by an angle~$\alpha$ ($0<\alpha<\frac{\pi}{3}$)
with respect to each other, so that
 $\bk_{1}$~makes an angle $-\alpha/2$ with the $x$~axis,
 while $\bk_{4}$ makes an angle $\alpha/2$ with the $x$~axis.
 The case $\alpha=\frac{\pi}{6}$ corresponds to the situation 12-fold
quasipattern treated in~\cite{Braaksma2017}, though with $\chi=0$.

The lattice (in the periodic case) or quasilattice $\Gamma$ 
are made up of integer sums of the six basic wave vectors: 
 \begin{equation}
 \Gamma =\left\{\bk\in \mathbb{R}^{2}: \bk = \sum_{j=1}^{6}m_{j}\bk_{j},
   \quad\text{with}\quad m_{j}\in \mathbb{Z}\right\}.
 \label{eq:basicvectors}
 \end{equation}
Notice that if $\bk\in\Gamma$ then $-\bk\in\Gamma$.
In the periodic case, the lattice is not dense, as in \cref{fig:wavevectors}(a),
while in the quasiperiodic case, the points in~$\Gamma$ are dense in the plane.

The periodic case occurs whenever the two sets of six wave vectors are not
rationally independent, meaning that, for example, $\bk_4$, $\bk_5$ and $\bk_6$
can all be written as rational sums of $\bk_1$ and~$\bk_2$. This happens
whenever $\cos\alpha$ and $\cos(\alpha+\frac{\pi}{3})$ are both rational, and
in this case, patterns defined by~\cref{eq:leading_order} are periodic in
space. We define the set~$\Ep$ to be these angles.
 \begin{definition}\label{def:Ep}
 Periodic case: the set $\Ep$ of angles is defined as
 \begin{equation*}
 \Ep:=\left\{\alpha \in \left(0,\frac{\pi}{3}\right): 
        \cos \alpha \in \mathbb{Q}
        \quad\text{\rm and}\quad
        \cos \left(\alpha + \frac{\pi}{3}\right)\in \mathbb{Q}\right\}.
 \end{equation*}
\end{definition}
In this case, $\Gamma$~is a lattice \red{with hexagonal symmetry}. We can replace
$\cos(\alpha+\frac{\pi}{3})$ in this definition by $\sqrt{3}\sin\alpha$.
 The set~$\Ep$ has the following properties:
 \begin{lemma} \label{lem:Ep_ab}
 (i) The set $\Ep$ is dense and has zero measure in~$(0,\frac{\pi}{3})$.

 (ii) If the wave vectors $\bk_{1}$, $\bk_{2}$, $\bk_{4}$ and
$\bk_{5}$ are not independent on $\mathbb{Q}$, then $\alpha\in\Ep$.

 (iii) If $\alpha \in \Ep$ then there exist co-prime integers $a,b$ such that 
 \begin{align}
 &a>b>\frac{a}{2}>0, \quad a\geq 3, \quad 
 \text{$a+b$ not a multiple of 3},\notag\\
 &\cos \alpha = \frac{a^{2}+2ab-2b^{2}}{2(a^{2}-ab+b^{2})},
 \quad
 \sqrt{3}\sin \alpha = \frac{3a(2b-a)}{2(a^{2}-ab+b^{2})}.  \label{eq:cosalpha}
 \end{align}
\red{Then the} wave
vectors $\bk_{j}$ are integer combinations of two smaller vectors 
$\bs_{1}$ and~$\bs_{2}$, of equal length~$\lambda=(a^{2}-ab+b^{2})^{-1/2}$,
making an angle of $\frac{2\pi}{3}$, with
 \begin{align}
  \bk_{1} &=  a    \bs_{1} + b    \bs_{2}, \quad
 &\bk_{2} &=  (b-a)\bs_{1} - a    \bs_{2}, \quad
 &\bk_{3} &= -b    \bs_{1} + (a-b)\bs_{2}, \label{eq:periodic_kj} \\
  \bk_{4} &=  a    \bs_{1} + (a-b)\bs_{2}, \quad
 &\bk_{5} &= -b    \bs_{1} - a    \bs_{2}, \quad
 &\bk_{6} &= (b-a) \bs_{1} + b    \bs_{2}. \nonumber
 \end{align}
\end{lemma}
Part (ii) of the Lemma is proved in~\cite{Iooss2019}, and parts (i) and (iii)
are proved in \cref{app:proof_of_Lemma_2p2}. The vectors $\bs_{1}$ and~$\bs_{2}$ are
illustrated in \cref{fig:wavevectors} in the case $(a,b)=(3,2)$ with
$\lambda=1/\sqrt{7}$. Requiring $a+b$ not to be a multiple of~3 means that we
need to allow $0<\alpha<\frac{\pi}{3}$ in the periodic case. In the
quasiperiodic case ($\alpha\in\Eqp$), we can always take $\alpha$ to be the smallest of the
angles between the vectors, which is why we define the set~$\Eqp$ to be the complement
of~$\Ep$ within the interval~$(0,\frac{\pi}{6}]$.

In~\cref{eq:basicvectors}, vectors $\bk\in\Gamma$ are indexed by six integers
$\bmm=(m_1,\dots,m_6)\in\mathbb{Z}^{6}$. However, using the fact that $\bk_1+\bk_2+\bk_3=0$ and $\bk_4+\bk_5+\bk_6=0$, the set~$\Gamma$ can be indexed
by fewer than six integers, and any $\bk\in\Gamma$ may be written, in both the
periodic and the quasiperiodic cases, as
 \begin{equation}
 \bk(\bmm) = m_{1}\bk_{1} + m_{2}\bk_{2} + m_{4}\bk_{4} + m_{5}\bk_{5},
 \quad
 (m_{1},m_{2},m_{4},m_{5})\in\mathbb{Z}^{4},
 \label{eq:fourvectors}
 \end{equation}
though in fact $\Gamma$~is indexed by two integers in the periodic
case~$\alpha\in\Ep$. 

\subsection{Functions on the (quasi)lattice}

We are now in a position to specify more precisely the form of the sum
in~\cref{eq:leading_order}. The function $u(\bx)$ is a real function that we
write in the form of a Fourier expansion with Fourier coefficients $u^{(\bk)}$:
 \begin{equation}
 u(\bx) = \sum_{\bk\in\Gamma} u^{(\bk)} e^{i\bk\cdot\bx},
 \quad 
 u^{(\bk)} = {\bar{u}}^{(-\bk)} \in \mathbb{C}.
 \label{eq:Fourier_u}
 \end{equation}
With $\bk\in \Gamma $ written as in~\cref{eq:fourvectors},
in the quasiperiodic case ($\alpha\in\Eqp$)
four indices are needed in the sum since the four vectors in~\cref{eq:fourvectors} are rationally independent.
In the periodic case, two indices are needed.
A norm $N_{\bk}$ for $\alpha\in\Eqp$ is defined by
 \begin{equation*}
 N_{\bk(\mathbf{m})}=|m_{1}|+|m_{2}|+|m_{4}|+|m_{5}|=|\mathbf{m}|,
 \end{equation*}
\red{where the coefficients~$m_{j}$ are uniquely defined for a given vector~$\bk\in\Gamma$}.
To give a meaning to the above Fourier expansion we need to introduce
Hilbert spaces $\mathcal{H}_{s}$, $s\geq 0:$
\begin{equation*}
 \mathcal{H}_{s} = 
   \left\{ u=\sum_{\bk\in\Gamma }u^{(\bk)}e^{i\bk\cdot\bx};
   \quad
   u^{(\bk)}=\overline{u}^{(-\bk)}\in\mathbb{C},
  \quad 
  \sum\limits_{\bk\in\Gamma}|u^{(\bk)}|^{2}(1+N_{\bk}^{2})^{s}<\infty \right\} ,
\end{equation*}
It is known that $\mathcal{H}_{s}$ is a Hilbert space with the scalar
product
 \begin{equation*}
 \langle u,v\rangle_{s} = 
 \sum_{\bk\in \Gamma }(1+N_{\bk}^{2})^{s}u^{(\bk)}\overline{v}^{(\bk)},
 \end{equation*}
and that $\mathcal{H}_{s}$ is an algebra for $s>2$ (see~\cite{Braaksma2017}), and possesses properties of Sobolev spaces $H_{s}$ in dimension~4, \red{for example $u$ is of class $C^{l}$ for $s>l+2$}. For $\alpha\in\Eqp$, a function
in~$\mathcal{H}_{s}$, defined by a convergent Fourier series as
in~\cref{eq:Fourier_u}, represents in general a quasipattern, i.e., a function that is
quasiperiodic in all directions. It is possible of course for such functions
still to be periodic (e.g., rolls or hexagons) if subsets of the Fourier
amplitudes are zero. With this definition of the scalar product, the twelve 
basic modes are orthogonal in~$\mathcal{H}_s$ and orthonormal in~$\mathcal{H}_0$:
 \begin{equation*}
   \left\langle e^{i\bk_{j}\cdot \bx},e^{i\bk_{l}\cdot\bx}\right\rangle_{0} 
 = \left\langle e^{-i\bk_{j}\cdot \bx},e^{-i\bk_{l}\cdot\bx}\right\rangle _{0}
 = \delta_{j,l}
 \quad\text{and}\quad
 \left\langle e^{\pm{i}\bk_{j}\cdot \bx},e^{\mp i\bk_{l}\cdot \bx}\right\rangle _{0}=0,
 \end{equation*}
where $\delta_{j,l}$ is the Kronecker delta.

The following useful Lemma is proven in~\cite{Iooss2019}:
 \begin{lemma} \label{lem:E0_mod_k_eq_1}
 For nearly all $\alpha \in (0,\frac{\pi}{6}]$, and in particular for 
 $\alpha\in\mathbb{Q}\pi \cap (0,\frac{\pi}{6}]$, 
 the only solutions of $|\bk(\bmm)|=1$ are $\pm \bk_{j}$, $j=1,\dots,6$.
 These vectors can be expressed with four integers as in~\cref{eq:fourvectors}:
 \begin{equation*}
 \bmm=(\pm 1,0,0,0),(0,\pm 1,0,0),(0,0,\pm 1,0),(0,0,0,\pm 1),
 \pm (1,1,0,0),\pm (0,0,1,1).
 \end{equation*}
 \end{lemma}
\red{For these values of~$\alpha$, the only vectors in~$\Gamma$ that are on the 
unit circle are the original twelve vectors, defining the set~$\EO$:}
 \red{\begin{definition}\label{def:E0}
 $\EO$ is the set of $\alpha$'s such that \cref{lem:E0_mod_k_eq_1} applies: the set of $\alpha \in (0,\frac{\pi}{6}]$ such that the only solutions of $|\bk(\bmm)|=1$ are $\pm \bk_{j}$, $j=1,\dots,6$.
\end{definition}}
The set~$\EO$ is dense and of full measure in $(0,\frac{\pi}{6}]$  
\red{(see~\cite{Iooss2019}, proof of Lemma~5)}, \red{and contains angles 
$\alpha\in\Ep$ and $\alpha\in\Eqp$}.
\red{Not every $\alpha\in\Ep$ is also in~$\EO$; for example, if $(a,b)=(8,5)$, we have $3\bk_1+\bk_2-2\bk_4+\bk_5=(5b-4a)\bs_2=(0,1)$, which is a vector on the unit circle but not in the original twelve.
For $\alpha\in\Eqp$,} it is possible to
show, for example, that $\alpha\approx25.66^\circ$
($\cos\alpha=\frac{1}{4}\sqrt{13}$) is in~$\EO$, while
$\alpha\approx26.44^\circ$ ($\cos\alpha=\frac{1}{12}(5+\sqrt{33})$) is not
(neither of these examples is a rational multiple of~$\pi$). See \cref{app:two_examples} for details \red{of these two examples}.

\subsection{Symmetries and actions}

Our problem possesses important symmetries. First, the system~\cref{eq:SHeq} is
invariant under the Euclidean group~$E(2)$ of rotations, reflections and
translations of the plane. We denote by~$\mathbf{R}_{\theta}u$ the pattern~$u$
rotated by an angle~$\theta$ centered at the origin, so
$(\mathbf{R}_{\theta}u)(\bx)=u(\mathbf{R}_{-\theta}\bx)$, \red{where $\mathbf{R}_{-\theta}\bx$ is $\bx$ rotated by an angle~$-\theta$}. We define similarly
the reflection~$\tau$ in the $x$~axis, and the translation
$\mathbf{T}_{\bdelta}$ by an amount~$\bdelta$, so $(\tau
u)(x,y)=u(x,-y)$ and $(\mathbf{T}_{\bdelta}u)(\bx)=u(\bx-\bdelta)$. Finally, in the case $\chi=0$,
equation~\cref{eq:SHeq} is odd in~$u$ and so commutes with the
symmetry~$\mathbf{S}$ defined by $\mathbf{S}u=-u$. If $\chi\neq0$, then in
addition to the change $u\rightarrow-u$, we need to change
$\chi\rightarrow-\chi$.

The leading order part~$v_1(\bx)$ of our solution will be as in~\cref{eq:leading_order}:
 \begin{equation}
 v_{1}(\bx) = \sum_{j=1}^{6} z_{j} e^{i\bk_{j}\cdot\bx} 
                           + {\bar{z}}_{j} e^{-i\bk_{j}\cdot \bx},
 \quad\text{with}\quad
 z_{j}\in\mathbb{C}.
 \label{eq:ui}
 \end{equation}
With Fourier modes restricted to those with wavevectors in~$\Gamma$, not all symmetries 
in~$E(2)$ are possible, \red{in particular, only rotations that preserve the (quasi)lattice~$\Gamma$ are permitted}. Those that are allowed
act on the basic Fourier functions as follows:
 \begin{align*}
 \mathbf{T}_{\bdelta}(e^{i\bk_{j}\cdot \bx}) &= 
   e^{i\bk_{j}\cdot (\bx-\bdelta)}, \\
 \mathbf{R}_{\frac{\pi}{3}}(e^{i\bk_{1}\cdot \bx},
 \dots,
 e^{i\bk_{6}\cdot \bx}) &=
   (e^{-i\bk_{3}\cdot \bx},e^{-i\bk_{1}\cdot \bx},e^{-i\bk_{2}\cdot \bx},e^{-i\bk_{6}\cdot \bx},e^{-i\bk_{4}\cdot \bx},e^{-i\bk_{5}\cdot \bx}), \\
 \mathbf{\tau }(e^{i\bk_{1}\cdot \bx},\dots,
 e^{i\bk_{6}\cdot\bx}) &=
   (e^{i\bk_{4}\cdot \bx},e^{i\bk_{6}\cdot\bx},e^{i\bk_{5}\cdot \bx},e^{i\bk_{1}\cdot\bx},e^{i\bk_{3}\cdot \bx},e^{i\bk_{2}\cdot\bx}).
 \end{align*}
This leads to a representation of the symmetries acting on the six
complex amplitudes~$z_j$ as
 \begin{align}
 \mathbf{T}_{\bdelta}:
 (z_{1},\dots,z_{6}) &\mapsto 
 \left(z_{1}e^{-i\bk_{1}\cdot \bdelta},
       z_{2}e^{-i\bk_{2}\cdot \bdelta},
       z_{3}e^{-i\bk_{3}\cdot \bdelta},
       z_{4}e^{-i\bk_{4}\cdot \bdelta},
       z_{5}e^{-i\bk_{5}\cdot \bdelta},
       z_{6}e^{-i\bk_{6}\cdot \bdelta}\right),  \nonumber \\
 \mathbf{R}_{\frac{\pi}{3}}:
 (z_{1},\dots,z_{6}) &\mapsto 
 \left({\bar{z}}_{2}, {\bar{z}}_{3}, {\bar{z}}_{1},
       {\bar{z}}_{5}, {\bar{z}}_{6}, {\bar{z}}_{4}\right), \label{eq:actsym} \\
 \mathbf{\tau}:
 (z_{1},\dots,z_{6}) &\mapsto 
 (z_{4},z_{6},z_{5},z_{1},z_{3},z_{2}).  \nonumber
 \end{align}
We will use these symmetries, as well as the ``hidden symmetries'' in $E(2)$~\cite{Crawford1994,Dawes2003,Dionne1997}, to restrict the form of the formal power series
for the amplitudes~$z_j$.


\section{Formal power series for solutions}\label{sec:formal}
In this section, we look for amplitude equations for solutions of~\cref{eq:SHeq}, expressed in the form of
a formal power series of the following type
 \begin{equation}
 u(\bx) = \sum_{n\geq 1}v_{n}(\bx),
 \quad 
 \mu =\sum_{n\geq 1}\mu_{n},
 \label{eq:formal}
 \end{equation}
where $v_{n}$ and $\mu_{n}$ are real. As in~\cite{Iooss2019}, the leading order
part~$v_{1}$ of a solution~$u$ satisfies
 \begin{equation*}
 \mathbf{L}_{0}v_{1}=0,
 \end{equation*}
where the linear operator $\mathbf{L}_{0}$ is defined by
 \begin{equation*}
 \mathbf{L}_{0}=(1+\Delta)^{2},
 \end{equation*}
so that $v_{1}$ lies in the kernel of~$\mathbf{L}_{0}$. Our twelve chosen
wavevectors~$\pm\bk_j$ all have length~1, so $\mathbf{L}_{0}e^{\pm
i\bk_j\cdot\bx}=0$, and we can write~$v_1$ as a linear combination of these
waves as in~\cref{eq:ui}.

Higher order terms are written concisely using multi-index notation: let 
$\bp=(p_1,\dots,p_6)$ and 
$\bpp=(p'_1,\dots,p'_6)$, where $p_j$ and $p'_j$ are non-negative integers,
and define
 \begin{equation*}
 \bz^\bp = z_1^{p_1} z_2^{p_2} z_3^{p_3} z_4^{p_4} z_5^{p_5} z_6^{p_6}
 \quad\text{and}\quad
 {\bar\bz}^\bpp = {\bar z}_1^{p'_1} {\bar z}_2^{p'_2} {\bar z}_3^{p'_3} {\bar z}_4^{p'_4} {\bar z}_5^{p'_5} {\bar z}_6^{p'_6}.
 \end{equation*}
We also take $|\bp|=p_1+\dots+p_6$ and $|\bpp|=p'_1+\dots+p'_6$.
\red{Each order~$n$ means a corresponding degree in monomials $\mathbf{z}^{p}\mathbf{\bar z}^{p'}$ with  $n=|p|+|p'|$}, so we look for $v_{n}$ and $\mu_{n}$ of the form 
 \begin{equation}
   v_{n}(\bx) = \sum_{|\bp|+|\bpp|=n} \bz^\bp {\bar\bz}^\bpp v_{\bp,\bpp}(\bx)
 \quad\text{and}\quad
 \mu_{n}      = \sum_{|\bp|+|\bpp|=n} \bz^\bp {\bar\bz}^\bpp \mu_{\bp,\bpp}.
 \label{eq:umuppp}
 \end{equation}
Here, $\mu_{\bp,\bpp}$ are constants and $v_{\bp,\bpp}(\bx)$ are functions
made up of sums of modes of order $n=|\bp|+|\bpp|$, such that 
 \begin{equation*}
 \left\langle v_{\bp,\bpp},e^{\pm{i}\bk_{j}\cdot \bx}\right\rangle_{0} = 0,
 \quad\text{for $n>1$ and $j=1,\dots,6$.}
 \end{equation*}

Writing~\cref{eq:SHeq} as
 \begin{equation}
 \mathbf{L}_{0}u=\mu u-\chi u^{2}-u^{3}  \label{eq:SH}
 \end{equation}
and replacing $u$ and $\mu$ by their expansions~\cref{eq:formal} and~\cref{eq:umuppp}, we project the PDE~\cref{eq:SHeq} onto the kernel and the range of~$\mathbf{L}_{0}$. \red{Solving~\cref{eq:SH} is equivalent to solving the projection of~\cref{eq:SH} onto the kernel 
together with the projection of~\cref{eq:SH} onto the orthogonal complement of the kernel.
 Notice that for the quasipattern case the range is not closed, so that the projection on the range is in fact a projection onto the orthogonal complement of the kernel.}
The operator~$\mathbf{L}_{0}$ is self adjoint, so the left hand side of~\cref{eq:SH} is
orthogonal to the kernel of~$\mathbf{L}_{0}$:
 $\langle\mathbf{L}_{0}u,e^{\pm{i}\bk_{j}\cdot\bx}\rangle_{0}=
  \langle{u},\mathbf{L}_{0}e^{\pm{i}\bk_{j}\cdot\bx}\rangle_{0}=0$ for any~$u$. 
\red{In fact, for any given degree $n>1$, the right hand side of~\cref{eq:SH}
is a finite Fourier series, and eliminating the part lying in the kernel gives
a remaining series with Fourier modes $e^{i\bk\cdot\bx}$, with $\bk\in\Gamma$
apart from $\{\pm \bk_{j},j=1,\dots,6\}$. For these modes we have $|\bk|\neq 1$ since
$\alpha\in\EO$. Then, the operator $\mathbf{L}_{0}$ has a formal pseudo-inverse
on its range that is orthogonal to the kernel of~$\mathbf{L}_{0}$. This
pseudo-inverse is a bounded operator in
any~$\mathcal{H}_{s}$ when $\alpha\in\EO\cap\Ep$, since in the periodic case, 
nonlinear modes are on a lattice~$\Gamma$ and are
bounded away from the unit circle. However, the pseudo-inverse is unbounded when
$\alpha\in\Eqp$ as a result of the presence of small divisors
(see~\cite{Iooss2019}). But, for a formal computation of the power
series~\cref{eq:umuppp}, we only need at each order to pseudo-invert
a \emph{finite} Fourier series, which is always possible
provided that~$\alpha\in\EO$. Solving the range equation allows us to get
$\mathbf{Q}_0 u$, which is the part of $u$ orthogonal to the kernel, as functions
of $(v_1,\mu)$, with $v_1$ given by~\cref{eq:ui}. Taking the series obtained by
solving the range equation (formally in the quasipattern case),  and replacing
them in the kernel equation (6 complex components), leads to}
 \begin{equation}
 0 = \mu z_{j} - P_{j}(\chi,\mu,z_{1},\dots,z_{6},{\bar{z}}_{1},\dots,{\bar{z}}_{6}),
 \label{eq:bifurcation}
 \end{equation}
where $j=1$, \dots, $6$ and
 \begin{equation*}
 P_{j}(\chi,\mu,z_{1},\dots,z_{6},{\bar{z}}_{1},\dots,{\bar{z}}_{6})
 = 
 \left\langle \chi u^{2} + u^{3}, e^{i\bk_{j}\cdot\bx}\right\rangle_{0},
 \end{equation*}
where~$u$ here is thought of as a function of $\bz$ and $\bar\bz$ through the 
formal power series~\cref{eq:formal} and the expansion~\cref{eq:umuppp}. The dependency in $\mu$ of $P_{j}$ occurs at orders at least $\mu |z_{j}|^3$.

Expanding $P_{j}$ in powers of
$(\mu,z_{1},\dots,z_{6},{\bar{z}}_{1},\dots,{\bar{z}}_{6})$ results in a convergent
power series in the periodic case (the $P_j$ functions are analytic in some
ball around the origin), but in general these power series are not convergent
in the quasiperiodic case. Nonetheless, the formal power series are useful in
the proof of existence of the corresponding quasipatterns.

We can now use the symmetries of the problem to investigate the structure of
the bifurcation equation~\cref{eq:bifurcation}. The equivariance
of~\cref{eq:SH} under the translations~$\mathbf{T}_{\bdelta}$ \red{and its propagation onto the bifurcation equation, using~\cref{eq:actsym},} leads to
 \begin{equation}
 e^{i\bk_{1}\cdot\bdelta}
 P_{1}(\chi,\mu,z_{1}e^{-i\bk_{1}\cdot\bdelta},\dots,
     {\bar{z}}_{6}e^{i\bk_{6}\cdot\bdelta})
 = P_{1}(\chi,\mu,z_{1},\dots,{\bar{z}}_{6}).
 \label{eq:equivariance_translation}
 \end{equation}
A typical monomial in~$P_{1}$ has the form~$\bz^\bp{\bar\bz}^\bpp$, so 
let us define 
 \begin{align*}
  n_{1} &= p_{1}-p_{1}^{\prime} - 1, \quad
 &n_{2} &= p_{2}-p_{2}^{\prime}, \quad
 &n_{3} &= p_{3}-p_{3}^{\prime}, \\
  n_{4} &= p_{4}-p_{4}^{\prime}, \quad
 &n_{5} &= p_{5}-p_{5}^{\prime}, \quad
 &n_{6} &= p_{6}-p_{6}^{\prime}.
 \end{align*}
Then, a monomial appearing in~$P_{1}$ should 
satisfy~\cref{eq:equivariance_translation}, which leads to
 \begin{equation*}
 n_{1}\bk_{1} + n_{2}\bk_{2} + n_{3}\bk_{3} + 
 n_{4}\bk_{4} + n_{5}\bk_{5} + n_{6}\bk_{6} = 0,
 \end{equation*}
and, since $\bk_{1}+\bk_{2}+\bk_{3}=0$ and 
           $\bk_{4}+\bk_{5}+\bk_{6}=0$, we obtain
 \begin{equation}
 (n_{1}-n_{3})\bk_{1} + 
 (n_{2}-n_{3})\bk_{2} + 
 (n_{4}-n_{6})\bk_{4} + 
 (n_{5}-n_{6})\bk_{5} = 0,  \label{eq:ident_kj}
 \end{equation}
which is valid in all cases (periodic or not).

In the quasilattice case, the wave vectors $\bk_{1}$, $\bk_{2}$, $\bk_{4}$ and
$\bk_{5}$ are rationally independent, so \cref{eq:ident_kj} implies $n_{1}=n_{2}=n_{3}$ and
$n_{4}=n_{5}=n_{6}$, which leads to monomials of the form
 \begin{align*}
 & z_{1}
 u_{1}^{p_{1}^{\prime }}u_{2}^{p_{2}^{\prime }}u_{3}^{p_{3}^{\prime}}
 u_{4}^{p_{4}^{\prime }}u_{5}^{p_{5}^{\prime }}u_{6}^{p_{6}^{\prime}}
 {{q}}_{1}^{n_{1}} {{q}}_{4}^{n_{4}} 
 &\text{for $n_{1}\geq0$ and $n_{4}\geq0$},
 \\
 & z_{1}
 u_{1}^{p_{1}^{\prime }}u_{2}^{p_{2}^{\prime }}u_{3}^{p_{3}^{\prime}}
 u_{4}^{p_{4}}u_{5}^{p_{5}}u_{6}^{p_{6}}
 {{q}}_{1}^{n_{1}} {\bar{q}}_{4}^{|n_{4}|} 
 &\text{for $n_{1}\geq0$ and $n_{4}<0$},
 \\
 & {\bar{z}}_{2}{\bar{z}}_{3}
 u_{1}^{p_{1}}u_{2}^{p_{2}}u_{3}^{p_{3}}
 u_{4}^{p_{4}^{\prime}}u_{5}^{p_{5}^{\prime}}u_{6}^{p_{6}^{\prime}}
 {\bar{q}}_{1}^{|n_{1}|-1} {{q}}_{4}^{n_{4}} 
 &\text{for $n_{1}<0$ and $n_{4}\geq0$},
 \\
 & {\bar{z}}_{2}{\bar{z}}_{3}
 u_{1}^{p_{1}}u_{2}^{p_{2}}u_{3}^{p_{3}}
 u_{4}^{p_{4}}u_{5}^{p_{5}}u_{6}^{p_{6}}
 {\bar{q}}_{1}^{|n_{1}|-1} {\bar{q}}_{4}^{|n_{4}|}
 &\text{for $n_{1}<0$ and $n_{4}<0$},
 \end{align*}
where we define
 \begin{equation*}
 u_{j}=z_{j}{\bar{z}}_{j},
 \quad
 q_{1}=z_{1}z_{2}z_{3}
 \quad\text{and}\quad
 q_{4}=z_{4}z_{5}z_{6}.
 \end{equation*}
Then, the quasilattice case gives the following structure for $P_{1}$:
 \begin{equation}
 P_{1}(\chi,\mu,z_{1},\dots,{\bar{z}}_{6})
 = 
 z_{1}f_{1}(\chi,\mu,u_{1},\dots,u_{6},q_{1},q_{4},{\bar{q}}_{4})
 + {\bar{z}}_{2}{\bar{z}}_{3}
    f_{2}(\chi,\mu, u_{1},\dots,u_{6},{\bar{q}}_{1},q_{4},{\bar{q}}_{4}),
 \label{eq:QP_case}
 \end{equation}
where $f_{1}$ and $f_{2}$ are power series in their arguments. We deduce the
five other components of the bifurcation equation by using the equivariance
under symmetries $\mathbf{R}_{\frac{\pi}{3}}$, $\mathbf{\tau }$, and $\mathbf{S}$
(changing $\chi$ to~$-\chi$), observing that
 \begin{align*}
 \mathbf{R}_{\frac{\pi}{3}}:
 (u_{1},u_{2},u_{3},u_{4},u_{5},u_{6},q_{1},q_{4})
 &\mapsto
 (u_{2},u_{3},u_{1},u_{5},u_{6},u_{4},{\bar{q}}_{1},{\bar{q}}_{4}),\\
 \mathbf{\tau}:
 (u_{1},u_{2},u_{3},u_{4},u_{5},u_{6},q_{1},q_{4})
 &\mapsto
 (u_{4},u_{6},u_{5},u_{1},u_{3},u_{2},q_{4},q_{1}),\\
 \mathbf{S}:
 ( \chi,u_{1},u_{2},u_{3},u_{4},u_{5},u_{6},q_{1},q_{4})
 &\mapsto
 (-\chi,u_{1},u_{2},u_{3},u_{4},u_{5},u_{6},-q_{1},-q_{4}).
\end{align*}
Equivariance under symmetry~$\mathbf{R}_{\pi}$, which changes $z_{j}$ into $\bar{z_{j}}$, gives the following property of functions $f_{j}$ in~\cref{eq:QP_case}
\begin{eqnarray*}
    f_{1}(\chi,\mu,u_{1},\dots,u_{6},\bar{q_{1}},\bar{q_{4}},q_{4})
 &=& \bar{f_{1}}(\chi,\mu, u_{1},\dots,u_{6},q_{1},q_{4},{\bar{q}}_{4}),\\
 f_{2}(\chi,\mu,u_{1},\dots,u_{6},q_{1},\bar{q_{4}},q_{4})
 &=& \bar{f_{2}}(\chi,\mu, u_{1},\dots,u_{6},\bar{q_{1}},q_{4},{\bar{q}}_{4}).
\end{eqnarray*}
It follows that the coefficients in $f_1$ and in $f_2$ are \emph{real}.
Equivariance under symmetry $\mathbf{S}$ leads to the property that in~\cref{eq:QP_case} $f_1$ and $f_2$ are respectively even and odd in $(\chi,q_1,q_4)$.

In the periodic case, when $\alpha\in\Ep$, we deduce from \cref{app:proof_of_periodic_case}
that $P_{1}(\chi,z_{1},\dots,{\bar{z}}_{6})$ may be
written as
 \begin{align}
 & z_{1}f_{3}(\chi,\mu,u_{1},\dots,u_{6},
              q_{1},q_{4},{\bar{q}}_{4},
              q_{l,k},{\bar q}_{l,k}) + {} 
 \nonumber \\
 & \quad {} + {\bar{z}}_{2}{\bar{z}}_{3}
        f_{4}(\chi,\mu, u_{1},\dots,u_{6},
              {\bar{q}}_{1},q_{4},{\bar{q}}_{4},
              q_{l,k},{\bar q}_{l,k}) +{}
 \label{eq:struct_P1_periodic_case} \\
 & \quad {} + \sum_{s,t} q_{s,t}^{\prime}
                         f_{s,t}(\chi,\mu,u_{1},\dots,u_{6},
                                 q_{1},{\bar{q}}_{1},q_{4},{\bar{q}}_{4},
                                 q_{l,k},{\bar q}_{l,k}),
 \nonumber
 \end{align}
where the monomials $q_{l,k}$, $l=I,II,III,IV,V,VI,VII,VIII,IX$, and 
$k=1,2,3$, are defined in \cref{app:proof_of_periodic_case}, the functions $f_{j}$ depend
on all arguments $q_{l,k}$ and ${\bar q}_{l,k}$, and the monomials 
$q_{s,t}^{\prime }$, $s=IV,V,VI,VII,VIII,IX$, $t=1,2,3$, are defined by
 \begin{equation*}
 q_{s,t}^{\prime }=\frac{{\bar q}_{s,t}}{{\bar{z}}_{1}}.
 \end{equation*}
We observe that the ``exotic'' terms with lowest degree 
in~\cref{eq:struct_P1_periodic_case} 
have degree $2a-1$, which is at least of 5th order, since $a\geq3$. 
Moreover, the symmetries act as indicated in \cref{app:proof_of_periodic_case}.

\section{Solutions of the bifurcation equations}\label{sec:solutions}
The strategy for proving existence of solutions of the PDE~\cref{eq:SHeq} is 
first to find solutions of the amplitude equations $P_j(\chi,z_1,\dots,{\bar
z}_6)=\mu z_{j}$ truncated at some order, and then to use an appropriate implicit
function theorem to show that there is a corresponding solution to the PDE,
using the results of~\cite{Iooss2019} in the quasiperiodic case.
\red{We refer the reader to \cref{tab:RoadMap} for a summary of the solutions
we find. The main ones are periodic and quasiperiodic versions of equal
amplitude superpositions of hexagons (\emph{super-hexagons}, for any~$\chi$),
unequal amplitude superpositions of hexagons (unequal super-hexagons, $|\chi|\ll1$ only), and
superpositions of hexagons and rolls (\emph{hexa-rolls}, $\chi$~not too large).}

\subsection{Truncation to cubic order}\label{sec:truncationCubic}

Let us first consider the terms up to cubic order for $P_{1}$. In the 
periodic case, where we notice that $a\geq 3$, and in the quasiperiodic case, 
we find the same equation:
 \begin{equation*}
 P_{1}^{(3)}=\alpha _{0}{\bar{z}}_{2}{\bar{z}}_{3}+z_{1}
 \sum_{j=1}^{6} \alpha _{j}u_{j}.
 \end{equation*}
We compute coefficients $\alpha_{j}$, $j=0,\dots,6$ from (see \cref{app:form_of_the_cubic_part})
 \begin{equation*}
 \mu z_{1} = P_{1}^{(3)} = 
   \chi \langle v_{1}^{2},e^{i\bk_{1}\cdot \bx}\rangle + 
        \langle v_{1}^{3},e^{i\bk_{1}\cdot \bx}\rangle -
   2\chi^{2}\langle v_{1}\widetilde{\mathbf{L}_{0}}^{-1}\mathbf{Q}_{0}v_{1}^{2},
                    e^{i\bk_{1}\cdot\bx}\rangle 
\end{equation*}
where $u=v_1$~\cref{eq:ui} at leading order, the scalar product is the one of $\mathcal{H}_0$, $\mathbf{Q}_{0}$ is the orthogonal projection on the range of $\mathbf{L}_{0}$, $\widetilde{\mathbf{L}_{0}}$ being the restriction of $\mathbf{L}_{0}$
on its range, the inverse of which is the pseudo-inverse of $\mathbf{L}_{0}$, \red{as explained above in \cref{sec:formal}, as $\mathbf{Q}_{0}v_{1}^{2}$ has a finite Fourier series.} The higher orders \red{(at increasing orders)} are uniquely
determined from the infinite dimensional part of the problem, provided that
$\alpha\in\EO$, they start from order at least $|v_1|^{4}$.

It is straightforward to check that
\begin{align*} 
\alpha_{0}&=2\chi,\\
\alpha_{1}&=3-\chi^2 c_1,\\
\alpha_{2}&=\alpha_{3}=6-\chi^2 c_2,\\
\alpha_{4}&=6-\chi^2 c_{\alpha},\\
\alpha_{5}&=6-\chi^2 c_{\alpha+},\\
\alpha_{6}&=6-\chi^2 c_{\alpha-},
\end{align*}
where $c_1$, $c_2$ are constants and $c_{\alpha}$, $c_{\alpha+}$ and $c_{\alpha-}$ are \emph{real functions} of~$\alpha$ (\red{real because of the equivariance under $\mathbf{R}_{\pi}$}, see the detailed computation in \cref{app:form_of_the_cubic_part}).
Hence we have the bifurcation system, written up to cubic order in~$z_{j}$
\begin{align}
2\chi \overline{z_{2}}\overline{z_{3}} &=z_{1}[\mu
-\alpha_{1}u_{1}-\alpha_{2}(u_{2}+u_{3})-\alpha_{4 }u_{4}-\alpha_{5}u_{5}-\alpha_{6
}u_{6}]  \nonumber \\
2\chi \overline{z_{1}}\overline{z_{3}} &=z_{2}[\mu
-\alpha_{1}u_{2}-\alpha_{2}(u_{1}+u_{3})-\alpha_{4 }u_{5}-\alpha_{5}u_{6}-\alpha_{6
}u_{4}]  \nonumber \\
2\chi \overline{z_{1}}\overline{z_{2}} &=z_{3}[\mu
-\alpha_{1}u_{3}-\alpha_{2}(u_{1}+u_{2})-\alpha_{4 }u_{6}-\alpha_{5}u_{4}-\alpha_{6
}u_{5}]  \label{eq:bifurcequcub} \\
2\chi \overline{z_{5}}\overline{z_{6}} &=z_{4}[\mu
-\alpha_{1}u_{4}-\alpha_{2}(u_{5}+u_{6})-\alpha_{4 }u_{1}-\alpha_{5}u_{3}-\alpha_{6
}u_{2}]  \nonumber \\
2\chi \overline{z_{4}}\overline{z_{6}} &=z_{5}[\mu
-\alpha_{1}u_{5}-\alpha_{2}(u_{4}+u_{6})-\alpha_{4 }u_{2}-\alpha_{5}u_{1}-\alpha_{6
}u_{3}]  \nonumber \\
2\chi \overline{z_{4}}\overline{z_{5}} &=z_{6}[\mu
-\alpha_{1}u_{6}-\alpha_{2}(u_{4}+u_{5})-\alpha_{4 }u_{3}-\alpha_{5}u_{2}-\alpha_{6
}u_{1}].  \nonumber
\end{align}
It remains to find all small solutions of these six equations and check whether 
they are affected by including further higher order terms.

Before proceeding, we note that in the periodic case ($\alpha\in\Ep\cap\EO$), the equivariant branching lemma can be used to find some bifurcating branches of patterns~\cite{Dionne1997}. 
In the case $\chi\neq0$, where there is no $\mathbf{S}$ symmetry, these 
branches are called:
 \begin{align*}
 \text{Super-hexagons:}\quad & 
          z_1 = z_2 = z_3 = z_4 = z_5 = z_6 \in \mathbb{R},\\
 \text{Simple hexagons:}\quad &
          z_1 = z_2 = z_3\in\mathbb{R}, \quad z_4 = z_5 = z_6 = 0,\\
 \text{Rolls (stripes):}\quad & 
          z_1\in\mathbb{R}, \quad z_2 = z_3 = z_4 = z_5 = z_6 = 0,\\
 \text{Rhombs$_{1,4}$:}\quad &
          z_1 = z_4\in\mathbb{R}, \quad z_2 = z_3 = z_5 = z_6 = 0,\\
 \text{Rhombs$_{1,5}$:}\quad &
          z_1 = z_5\in\mathbb{R}, \quad z_2 = z_3 = z_4 = z_6 = 0,\\
 \text{Rhombs$_{1,6}$:}\quad &
          z_1 = z_6\in\mathbb{R}, \quad z_2 = z_3 = z_4 = z_5 = 0,
 \end{align*}
where the conditions on the $z_j$'s give examples of each type of solution.
When $\chi=0$ and there is $\mathbf{S}$~symmetry, there are additional 
branches:
 \begin{align}
 \text{Anti-hexagons:}\quad & 
          z_1 = z_2 = z_3 = -z_4 = -z_5 = -z_6 \in \mathbb{R},\nonumber\\
 \text{Super-triangles:}\quad & 
          z_1 = z_2 = z_3 = z_4 = z_5 = z_6 \in \mathbb{R}i,\nonumber\\
 \text{Anti-triangles:}\quad & 
          z_1 = z_2 = z_3 = -z_4 = -z_5 = -z_6 \in \mathbb{R}i,\label{eq:antihexagonsetc}\\
 \text{Simple triangles:}\quad &
          z_1 = z_2 = z_3\in\mathbb{R}i, \quad z_4 = z_5 = z_6 = 0,\nonumber\\
 \text{Rhombs$_{1,2}$:}\quad & 
          z_1 = z_2\in\mathbb{R}, \quad z_3 = z_4 = z_5 = z_6 = 0.\nonumber
 \end{align}
For $(a,b)=(3,2)$, it is known that there are additional branches of the form
$|z_1|=\dots=|z_6|$, with $\arg(z_1)=\dots=\arg(z_6)\approx\pm\frac{\pi}{3}$
and $\arg(z_1)=\dots=\arg(z_6)\approx\pm\frac{2\pi}{3}$, where the amplitude
and phases of the modes are determined at fifth order~\cite{Silber1998}.
We recover all these solutions below for all $\alpha\in\Ep\cap\EO$, with the addition of a new branch, 
consisting of a superposition of hexagons and rolls, for example with 
$z_1,z_2,z_3,z_4\neq0$ and $z_5=z_6=0$. This new kind of 
solution exists in both the periodic and quasiperiodic cases, but only exists if $\alpha_1$, $\alpha_2$, $\alpha_4$, $\alpha_5$, and~$\alpha_6$ satisfy certain inequalities (true if $\chi$ is not too large). This new solution 
cannot be found using the equivariant branching 
lemma since it does not live in a one-dimensional space fixed by a symmetry subgroup (though see also~\cite{Matthews2003a}). 

We will focus below primarily on the new types of solutions: superposition of two hexagon patterns and superposition of hexagons and rolls, but even in the quasiperiodic case, there are branches of periodic patterns. These include rolls, simple hexagons, rhombs \hbox{etc.}, and can be found even with $\alpha\in \mathcal{E}_{qp}$. But, since they involve only a reduced set of wavevectors that can be accommodated in periodic domains, there is no need for the quasiperiodic techniques of~\cite{Iooss2019} in these cases.

\subsection{Super-hexagons: superposition of two hexagonal patterns} \label{sec:superp_hexa}

In the case $q_{1}q_{4}\neq 0$ (all six amplitudes are non-zero), we multiply
each equation in~\cref{eq:bifurcequcub} by the appropriate ${\bar z}_{j}$ to
obtain at cubic order
\begin{align}
2\chi \overline{q_1} &=u_{1}[\mu
-\alpha_{1}u_{1}-\alpha_{2}(u_{2}+u_{3})-\alpha_{4 }u_{4}-\alpha_{5}u_{5}-\alpha_{6
}u_{6}]  \nonumber \\
2\chi \overline{q_1} &=u_{2}[\mu
-\alpha_{1}u_{2}-\alpha_{2}(u_{1}+u_{3})-\alpha_{4 }u_{5}-\alpha_{5}u_{6}-\alpha_{6
}u_{4}]  \nonumber \\
2\chi \overline{q_1} &=u_{3}[\mu
-\alpha_{1}u_{3}-\alpha_{2}(u_{1}+u_{2})-\alpha_{4 }u_{6}-\alpha_{5}u_{4}-\alpha_{6
}u_{5}]  \label{eq:cubictwohex} \\
2\chi \overline{q_4} &=u_{4}[\mu
-\alpha_{1}u_{4}-\alpha_{2}(u_{5}+u_{6})-\alpha_{4 }u_{1}-\alpha_{5}u_{3}-\alpha_{6
}u_{2}]  \nonumber \\
2\chi \overline{q_4} &=u_{5}[\mu
-\alpha_{1}u_{5}-\alpha_{2}(u_{4}+u_{6})-\alpha_{4 }u_{2}-\alpha_{5}u_{1}-\alpha_{6
}u_{3}]  \nonumber \\
2\chi \overline{q_4} &=u_{6}[\mu
-\alpha_{1}u_{6}-\alpha_{2}(u_{4}+u_{5})-\alpha_{4 }u_{3}-\alpha_{5}u_{2}-\alpha_{6
}u_{1}].  \nonumber
\end{align}
This implies that $q_{1}$ and $q_{4}$ are real \red{since the $u_j$'s and the coefficients are real}, and shows that
 \begin{equation*}
 u_{1} = u_{2} = u_{3}
 \quad\text{and}\quad
 u_{4} = u_{5} = u_{6}
 \end{equation*}
is always a possible solution.

There are other possible solutions, particularly when $\chi$ is close to zero. Such solutions are difficult to find in general as they involve solving six coupled cubic equations. Furthermore, other solutions at cubic order might not give solutions when we consider higher order terms in the bifurcation system~\cref{eq:bifurcation}. Considering these further is beyond the scope of this paper.

To solve~\cref{eq:cubictwohex} with $u_{1}=u_{2}=u_{3}$ and
$u_{4}=u_{5}=u_{6}$, and with $q_1$ and $q_4$ real, let us set
 \begin{align}
 z_{j} &=\varepsilon e^{i\theta_{j}}\text{ for }j=1,2,3,\quad
 \varepsilon>0, \quad
 \Theta_{1}=\theta_{1}+\theta_{2}+\theta_{3}=k\pi, \nonumber \\
 z_{j} &=\delta e^{i\theta_{j}}\text{ for }j=4,5,6,\quad
 \delta>0, \quad
 \Theta _{4}=\theta_{4}+\theta_{5}+\theta_{6}=k'\pi,
 \label{eq:hexa-hexa3}
 \end{align}
where $k$ and $k'$ are integers,
so $u_1=u_2=u_3=\varepsilon^2$, $u_4=u_5=u_6=\delta^2$, $\overline{q_1}=\varepsilon^3e^{-i\Theta_1}=\varepsilon^3(-1)^k$ and $\overline{q_4}=\delta^3e^{-i\Theta_4}=\delta^3(-1)^{k'}$.
Then, for $\varepsilon\delta>0$ we have only 2 equations
 \begin{eqnarray*}
 2\chi \varepsilon (-1)^k &=&
 \mu -(\alpha _{1}+2\alpha_{2})\varepsilon ^{2}-(\alpha _{4}+\alpha _{5}+\alpha _{6})\delta ^{2}, \\
 2\chi \delta (-1)^{k'} &=&\mu -(\alpha _{1}+2\alpha _{2})\delta
^{2}-(\alpha _{4}+\alpha _{5}+\alpha _{6})\varepsilon ^{2}.
 \end{eqnarray*}
It follows that 
 \begin{equation}
 2\chi \left(\varepsilon (-1)^k - \delta (-1)^{k'}\right)
 =[(\alpha_{4}+\alpha_{5}+\alpha_{6}) - (\alpha_{1}+2\alpha_{2})]
 (\varepsilon^{2}-\delta ^{2}).
 \label{eq:id1}
 \end{equation}
Hence $\left(\varepsilon (-1)^k-\delta(-1)^{k'}\right)$ is a factor in~\cref{eq:id1}, and there are two types of solutions, depending on whether this factor is zero or not.

\paragraph{{Equal amplitude super-hexagons}}
We first consider the case where the factor is zero; it follows that
\begin{equation*}
\delta =\varepsilon>0 \quad\text{and}\quad
k=k'=0\text{ or }1\text{,}
\end{equation*}
and
\begin{equation}
\mu =2\chi \varepsilon (-1)^{k}+(\alpha _{1}+2\alpha _{2}+\alpha _{4}+\alpha
_{5}+\alpha _{6})\varepsilon ^{2}, \label{eq:solu1}
\end{equation}
or equivalently,
\begin{equation*}
\mu =2\chi \varepsilon (-1)^{k}+(33-\chi^2(c _{0}+2c _{1}+c _{\alpha}+c
_{\alpha+}+c _{\alpha-}))\varepsilon ^{2}.  
\end{equation*}
\red{We call these solutions super-hexagons in the periodic case, as in~\cite{Dionne1997},
and \emph{QP-super-hexagons} in the quasiperiodic case.}
Notice that when $|\chi|$ is not too \red{large}, the coefficient of $\varepsilon^2$ is positive, and $k$ is set by the relative signs of $\mu$ and~$\chi$. For  $|\chi|\ll\varepsilon$, the bifurcation is supercritical ($\mu>0$).

\paragraph{{Unequal amplitude super-hexagons}}
If the factor is non-zero, this implies 
 \begin{equation*}
 \varepsilon(-1)^k\neq \delta(-1)^{k'}, 
 \text{  i.e.,  } 
 \delta \neq \varepsilon,\text{  or  } (-1)^k \neq (-1)^{k'}.
 \end{equation*}
Dividing~\cref{eq:id1} by the non-zero factor leads to
 \begin{equation*}
 2\chi =C\left(\varepsilon(-1)^k+\delta(-1)^{k'}\right),
 \end{equation*}
with
 \begin{equation*}
 C\overset{def}{=}(\alpha _{4}+\alpha _{5}+\alpha _{6})-(\alpha _{1}+2\alpha_{2}).
 \end{equation*}
This leads to the non-degeneracy condition $C\neq 0$, and to the fact that this 
\red{unequal amplitude} solution is valid only for $|\chi|$ close to~0. 
The assumption on $C$ is satisfied for most values of~$\chi$ since 
 \begin{equation*}
 C=3-\chi ^{2}(c_{\alpha }+c_{\alpha +}+c_{\alpha-}-c_{1}-2c_{2}).
 \end{equation*}
Hence, for $|\chi|$ close enough to 0, we find new solutions parameterized by~$\varepsilon>0$ and~$k$:
 \begin{equation}
 \delta =\left[\frac{2\chi }{3}-\varepsilon (-1)^{k}\right](-1)^{k^{\prime}} + \mathcal{O}(\chi ^{3}).\label{eq:solu2a}
 \end{equation}
Here $k$ may be 0 or 1 and $k'$ is chosen so that $\delta>0$.
At leading order in $(\varepsilon,\chi)$, we have
 \begin{equation}
 \mu =33\varepsilon ^{2}-22\chi \varepsilon (-1)^{k}+8\chi ^{2}.
 \label{eq:solu2}
 \end{equation}
\red{The solutions are unequal ($\delta\neq\varepsilon$, with $\chi\neq0$) superpositions of hexagons,
so we call them unequal super-hexagons and unequal QP-super-hexagons in the periodic
and quasiperiodic cases.}

The next step is to show that these solutions to the cubic amplitude equations persist as solutions of the bifurcation equations~\cref{eq:bifurcation} once higher order terms are considered. This is simpler in the quasiperiodic case as there are no resonant higher order terms to consider.

\subsubsection{Quasipattern cases: Higher orders} \label{sec:quasipat_case_high_orders}

In this case wave vectors $\bk_{1}$, $\bk_{2}$, $\bk_{4}$ and $\bk_{5}$ are
rationally independent. Using the symmetries, the general form of the
six-dimensional bifurcation equation is deduced from~\cref{eq:QP_case}
and~\cref{eq:hexa-hexa3}, which gives two real bifurcation equations,
\red{where functions $f_{j}$ are formal power series in their arguments}:
 \begin{align}
 \mu  &= f_{1}(\chi,\mu, \varepsilon^{2}, \varepsilon ^{2}, \varepsilon^{2}, \delta^{2}, \delta^{2}, \delta^{2}, \varepsilon^{3}(-1)^{k}, \delta^{3}(-1)^{k'}) + {} \nonumber \\
 &\quad{}+\varepsilon(-1)^{k} f_{2}(\chi,\mu, \varepsilon^{2}, \varepsilon^{2}, \varepsilon^{2}, \delta^{2}, \delta^{2}, \delta^{2}, \varepsilon^{3}(-1)^{k}, \delta^{3}(-1)^{k'}),  \label{eq:syst_hexa-hexa} \\
 \mu  &= f_{1}(\chi,\mu, \delta^{2}, \delta^{2}, \delta^{2}, \varepsilon^{2}, \varepsilon^{2}, \varepsilon^{2}, \delta^{3}(-1)^{k'}, \varepsilon^{3}(-1)^{k}) + {} \nonumber \\
 &\quad{}+\delta(-1)^{k'} f_{2}(\chi,\mu, \delta^{2}, \delta^{2}, \delta^{2}, \varepsilon^{2}, \varepsilon^{2}, \varepsilon^{2}, \delta^{3}(-1)^{k'}, \varepsilon^{3}(-1)^{k}).  \nonumber
\end{align}

\paragraph{{Equal amplitude QP-super-hexagons}}

It is clear that we still have solutions with
 \begin{equation*}
 \varepsilon (-1)^{k}=\delta (-1)^{k'},\text{ i.e., } \varepsilon=\delta>0, k=k',
\end{equation*}
which leads to a single equation
 \begin{align}
 \mu &= f_{1}(\chi,\mu, \varepsilon^{2}, \varepsilon^{2}, \varepsilon^{2}, \varepsilon^{2}, \varepsilon^{2}, \varepsilon^{2}, \varepsilon^{3}(-1)^{k}, \varepsilon^{3}(-1)^{k}) + {} \nonumber \\
 &\quad{} + \varepsilon(-1)^{k} f_{2}(\chi,\mu, \varepsilon^{2}, \varepsilon^{2}, \varepsilon^{2}, \varepsilon^{2}, \varepsilon^{2}, \varepsilon^{2}, \varepsilon^{3}(-1)^{k}, \varepsilon ^{3}(-1)^{k}), 
 \label{eq:hexa-hexa_a}
 \end{align}
which may be solved with respect to~$\mu$ by the implicit function theorem \red{adapted for use with formal power series: we use the implicit function theorem for analytic functions, suppressing the proof of convergence for the series}.
This gives a formal power series in~$\varepsilon$, the leading order terms being~\cref{eq:solu1}.

\red{Following the process used in section~3 of~\cite{Iooss2019} for solving the range equation (the projection of~\cref{eq:SHeq} on the orthogonal complement of  $\ker \mathbf{L}_0$, with $z_{j}=\varepsilon e^{i\theta_{j}}$, $\Theta_{1}=\Theta_{4}=0$), we need typically to take $(\varepsilon,\mu,\chi)$  in a ``good set'' of parameters, where the Diophantine conditions of \cref{app:definitions_of_all_sets} are useful. Then the bifurcation equation~\cref{eq:hexa-hexa_a} may be solved by the usual implicit function theorem.  Checking that at the end the parameters lie in the ``good set'' needs a ``transversality condition,'' which is the same as in~\cite{Iooss2019}. The solution finally is proved to exist in a union of disjoint intervals for~$\varepsilon$, going to full measure as~$\varepsilon$ goes to~0.}
\begin{remark} \red{In the case of a quasiperiodic lattice, for all formal solutions found below in the form of a power series of some amplitudes, the proof of existence of a true solution follows the same lines as above. So we shall not repeat the argument.}
\end{remark}

\paragraph{{Unequal amplitude QP-super-hexagons}}

Now, assuming that $\varepsilon (-1)^{k}\neq\delta(-1)^{k'}$, and
taking the difference between the two equations in~\cref{eq:syst_hexa-hexa}, we find (simplifying the notation):
\begin{eqnarray*}
0 &=& f_{1}(\chi,\mu, \varepsilon^{2}, \delta^{2}, \varepsilon^{3}(-1)^{k}, \delta ^{3}(-1)^{k'})-f_{1}(\chi,\mu,\delta
^{2},\varepsilon ^{2},\delta ^{3}(-1)^{k'},\varepsilon
^{3}(-1)^{k})+{} \\
&&\quad{}+\varepsilon (-1)^{k}f_{2}(\chi,\mu,\varepsilon ^{2},\delta ^{2},\varepsilon
^{3}(-1)^{k},\delta ^{3}(-1)^{k'})-\delta (-1)^{k^{\prime
}}f_{2}(\chi,\mu,\delta ^{2},\varepsilon ^{2},\delta ^{3}(-1)^{k^{\prime
}},\varepsilon ^{3}(-1)^{k})
\end{eqnarray*}
where we can simplify by the factor $\varepsilon (-1)^{k}-\delta
(-1)^{k'}$. The leading terms are 
\begin{equation*}
0=2\chi -C(\varepsilon (-1)^{k}+\delta (-1)^{k'}),
\end{equation*}
as in the cubic truncation,
showing again that these solutions are only valid for $\chi$ close to~0. It
is then clear that provided that $C\neq 0$, which holds for $\chi $ close to
zero, the system formed by this last equation, with the first one of~\cref{eq:syst_hexa-hexa}, may be solved with respect to $\delta$ and~$\mu$ using the \red{formal} implicit function theorem \red{(as above, since the solution given by the principal part is not degenerate)} to obtain a formal
power series in $(\varepsilon,\chi)$,  their
leading order terms being given in~\cref{eq:solu2a}, \cref{eq:solu2}.
We notice that there are four degrees of freedom, with the
values of $\theta_{1}$, $\theta_{2}$, $\theta_{4}$ and $\theta_{5}$ being arbitrary. We also notice
that we have two possible amplitudes depending on the parity of~$k$. All
these bifurcating solutions correspond to the superposition of hexagonal
patterns of unequal amplitude, where the change in $\theta_{j}$, $j=1,2,4,5$ correspond to a shift of each
pattern in the plane.

For both types of solution, we have thus proved that there are formal power series solutions of~\cref{eq:SH},
unique up to the allowed indeterminacy on the~$\theta_{j}$, of the form~\cref{eq:hexa-hexa3}. This does not prove that all solutions take the form~\cref{eq:hexa-hexa3}.
We can state

\begin{figure}
\begin{center}
\includegraphics[width=0.6\hsize]{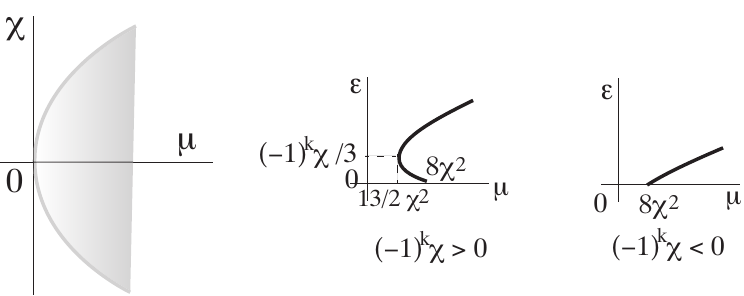}
\end{center}
\caption{Domain of existence (shaded) of bifurcating unequal amplitude QP-super-hexagons, for small~$|\chi|$. 
 These solutions only bifurcate from $\mu=0$ when $\chi=0$.}
\label{fig:bifdiag2}
\end{figure}

\begin{figure}
\begin{center}
\includegraphics[width=0.3\hsize]{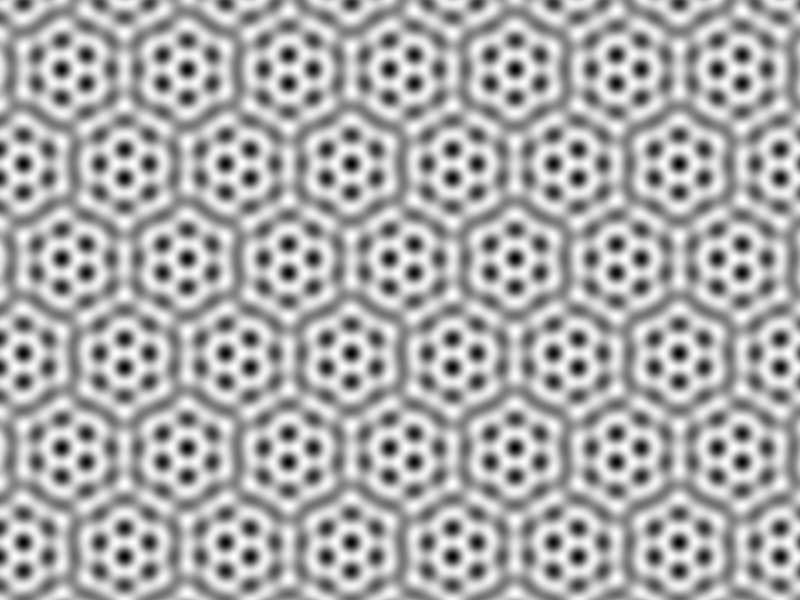}
\includegraphics[width=0.3\hsize]{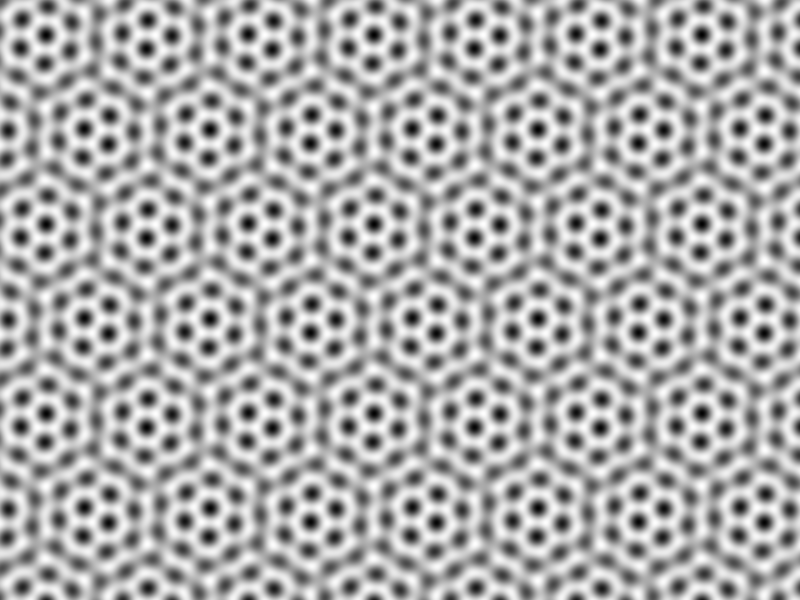}
\includegraphics[width=0.3\hsize]{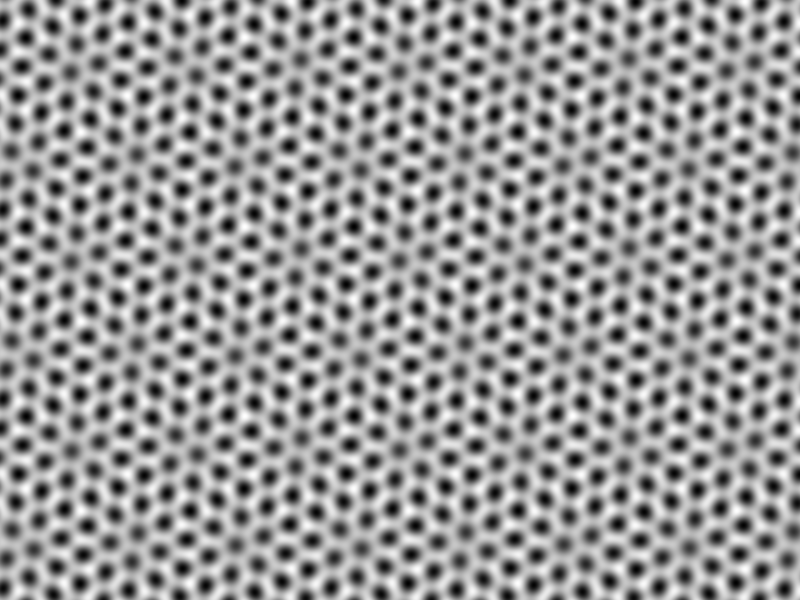}\\[10pt]
\includegraphics[width=0.3\hsize]{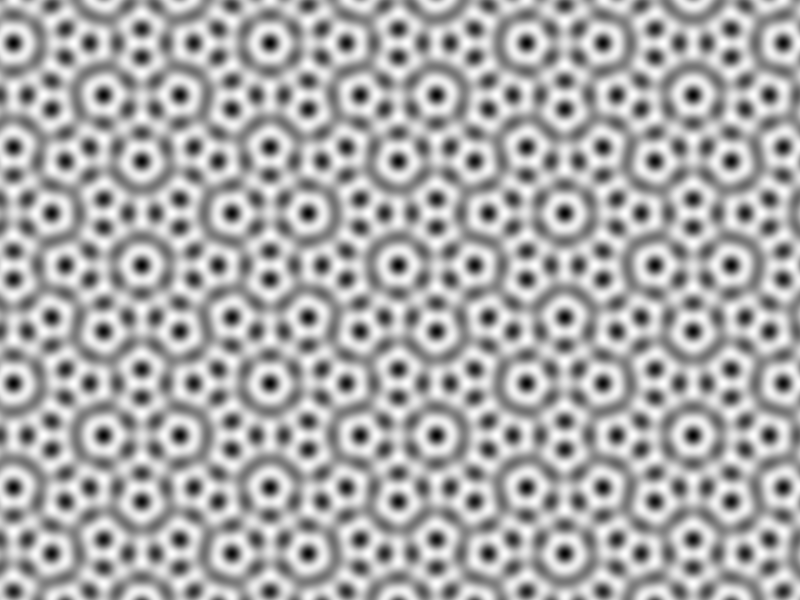}
\includegraphics[width=0.3\hsize]{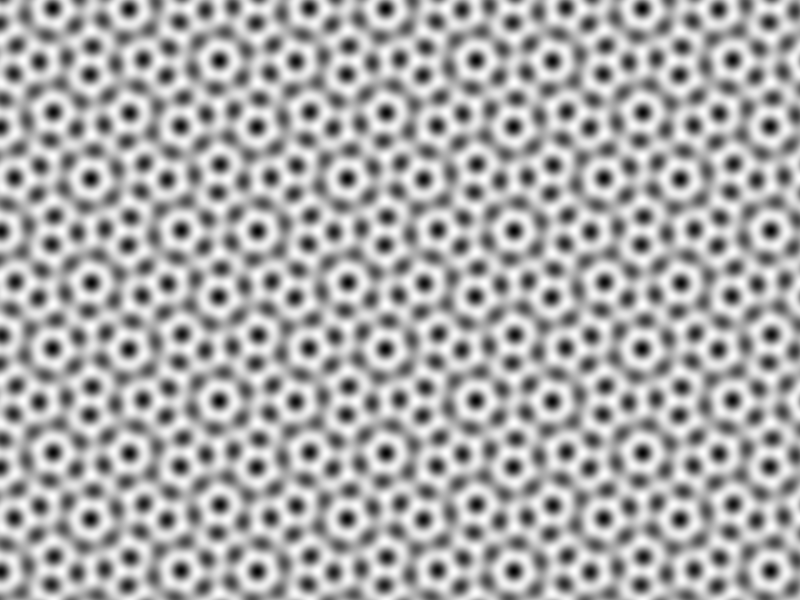}
\includegraphics[width=0.3\hsize]{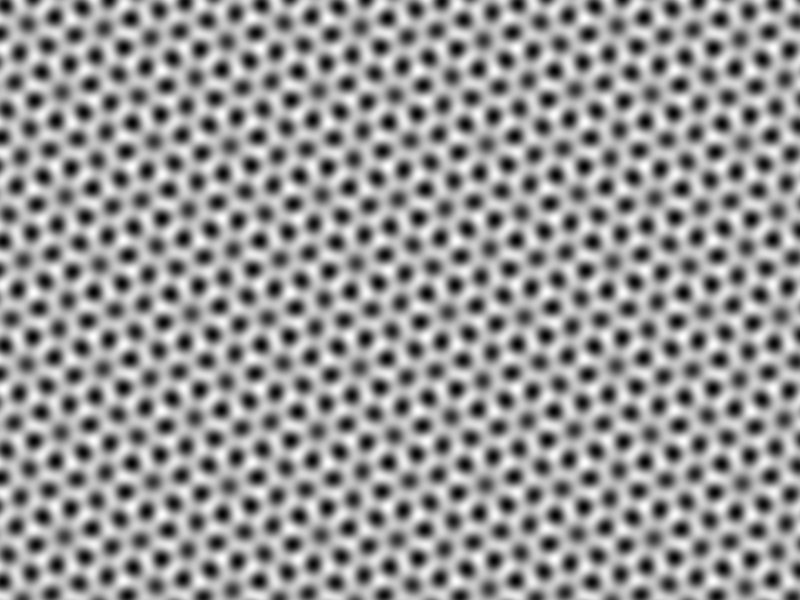}
\end{center}
 \caption{Examples of quasipatterns: superposition of hexagons.
 Top row: $\alpha=\frac{\pi}{12}=15^{\circ}$;
 bottom row: $\alpha=25.66^{\circ}$ ($\cos\alpha=\frac{1}{4}\sqrt{13}$).
 Left: \red{equal amplitude QP-super-hexagons};
 center and right: \red{unequal amplitude QP-super-hexagons},
 with $k=k'$ (center) and $k=k'+1$ (right).}
\label{fig:example_twohex_patterns}
\end{figure}

\begin{theorem}[Quasiperiodic superposed hexagons]
\label{thm:superposed_hexagons} Assume $\alpha\in\EO\cap\Eqp$, then for $\varepsilon,\chi $ fixed, we can build a four-parameter formal power series solution of~\cref{eq:SH} of the form
 \begin{align}
 u(\varepsilon,\chi,k,\Theta) &= \varepsilon u_{1} + \sum_{n\geq2} \varepsilon^{n} u_{n}(\chi, k, \Theta),\quad
 \varepsilon>0, \quad
 u_{n}\bot e^{i\bk_{j}\cdot\bx},\quad j=1,\dots,6,  \label{eq:asympexp_1} \\
 \mu (\varepsilon ,\chi ,k) &= (-1)^{k} 2\chi \varepsilon + \mu_2(\chi)\varepsilon^{2} + \sum_{n\geq 3}\varepsilon ^{n}\mu_{n}(\chi ,k),\text{ }k=0,1,  \nonumber\\
 \text{with}\quad
 u_{1} &= \sum_{j=1,\dots,6} e^{i(\bk_{j}\cdot\bx+\theta_{j})}+c.c.,\quad
 \Theta =(\theta _{1},\dots,\theta _{6}),   \nonumber \\
 \mu_2(\chi) &= 33 - \chi^2 (c_{1} + 2c_{2} + c_{\alpha}+c_{\alpha+}+c _{\alpha-})  \nonumber \\
 \theta_{1}+\theta_{2}+\theta_{3} &= k\pi, \quad
 \theta_{4}+\theta_{5}+\theta_{6} = k'\pi, \quad 
 k=k'=0,1  \nonumber \\
 u_{n}(-\chi,k,\Theta) &= (-1)^{n+1} u_{n}(\chi,k,\Theta),
 \quad
 \mu_{n}(-\chi,k) = (-1)^{n} \mu_{n}(\chi ,k). \nonumber
 \end{align}
\red{These are the equal amplitude QP-super-hexagons.}
Moreover, for a range of $(\mu,\chi)$ close to 0 (see \cref{fig:bifdiag2}),
\red{there are in addition} two \red{unequal amplitude QP-super-hexagon} solutions (for $k=0,1$),
given by
 \begin{align}
 u(\varepsilon ,\chi ,k,\Theta ) &=\varepsilon u_{10} + \delta u_{11} + 
 \sum_{m+p\geq 2} \varepsilon^{m}\chi^{p}u_{mp}(k,\Theta),
 \quad
 \varepsilon>0,\delta>0, \nonumber\\
 & \quad 
 u_{mp}\bot e^{i\bk_{j}\cdot\bx},\quad j=1,\dots,6,
 \quad
 u\text{ odd in }(\varepsilon ,\chi ),
 \nonumber\\[2pt]
 u_{10} &= \sum_{j=1,2,3}e^{i(\bk_j\cdot\bx+\theta_j)}+c.c.,
 \quad
 u_{11} = \sum_{j=4,5,6}e^{i(\bk_j\cdot\bx+\theta_{j})}+c.c.,
 \label{eq:hexa-hexa_exp_2}\\
 \theta_{1}+\theta_{2}+\theta_{3} &= k\pi, \quad
 k=0,1,\quad 
 \theta_{4}+\theta_{5}+\theta_{6} = k'\pi, \quad 
 k'=0,1\text{\ determined below},  \nonumber \\
 \delta (\varepsilon ,\chi ,k) &=
 (-1)^{k'}\left\{\frac{2\chi}{3} - (-1)^{k}\varepsilon +
 \sum_{m+p\geq 2}\varepsilon^{m}\chi^{p}\delta_{mp}(k)
 \right\},
 \quad
 (-1)^{k'}\delta \text{ odd in }((-1)^k\varepsilon ,\chi ), \nonumber\\
 \mu (\varepsilon ,\chi ,k) &=  33\varepsilon^{2}-22(-1)^{k}\varepsilon \chi + 8\chi ^{2} + 
 \sum_{m+p\geq 3}\varepsilon^{m}\chi ^{p}\mu_{mp}(k),
 \quad
 \mu \text{ even in }((-1)^k\varepsilon ,\chi ).\nonumber
 \end{align}
In the expression for~$\delta$, $k'$ is chosen so that $\delta>0$.
For either type of solution, changing $\theta_1,\theta_2,\theta_{4},\theta_{5}$
corresponds to translating each hexagonal pattern arbitrarily. \Cref{fig:example_twohex_patterns} shows examples of~$u_1$ for the two types of superposed hexagon quasipatterns, for two values of~$\alpha$.

Then, for $\alpha\in\Eii$, which is included in $\EO\cap\Eqp$, and using the same proof as in~\cite{Iooss2019}, both types of bifurcating quasipattern solutions of~\cref{eq:SHeq} are proved to
exist. The \red{equal amplitude QP-super-hexagons} have 
asymptotic expansion~\cref{eq:asympexp_1}, provided that
$\varepsilon$ is small enough, and the \red{unequal amplitude QP-super-hexagons} have
asymptotic expansion~\cref{eq:hexa-hexa_exp_2}, provided that
$\varepsilon,\chi$ are small enough.
\end{theorem}

 \begin{remark}
Symmetries of quasipatterns are hard to write down precisely~\cite{Baake2012} since the arbitrary relative position of the two hexagonal patterns may mean that there is no point of rotation symmetry or line of reflection symmetry. Nonetheless, with $\varepsilon=\delta$, the first type of solution is symmetric `on average' under rotations by $\frac{\pi}{3}$ and reflections conjugate to~$\mathbf{\tau}$. In fact the 4 parameter family of solutions is globally invariant under symmetries $\mathbf{R}_{\pi/3}$ and $\mathbf{\tau}$. 
Notice that, for the \red{unequal amplitude QP-super-hexagon solutions}, the reflection symmetry
$\mathbf{\tau}$ exchanges $(k,\varepsilon)$ with $(k',\delta)$.
 \end{remark}

\begin{figure}
\begin{center}
\includegraphics[width=0.3\hsize]{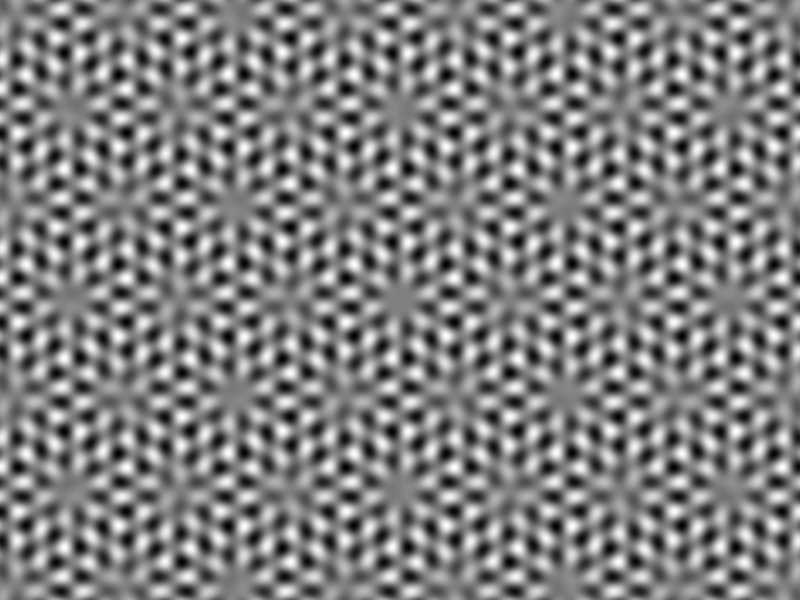}
\includegraphics[width=0.3\hsize]{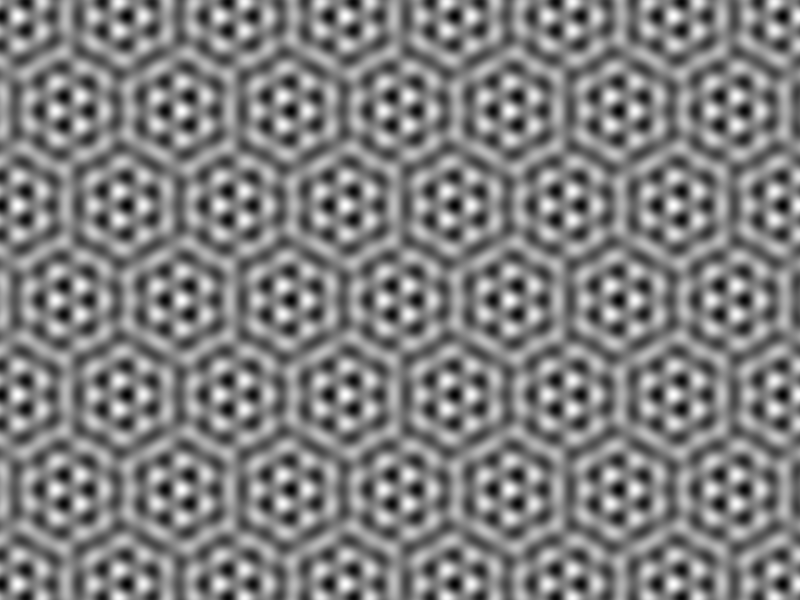}
\includegraphics[width=0.3\hsize]{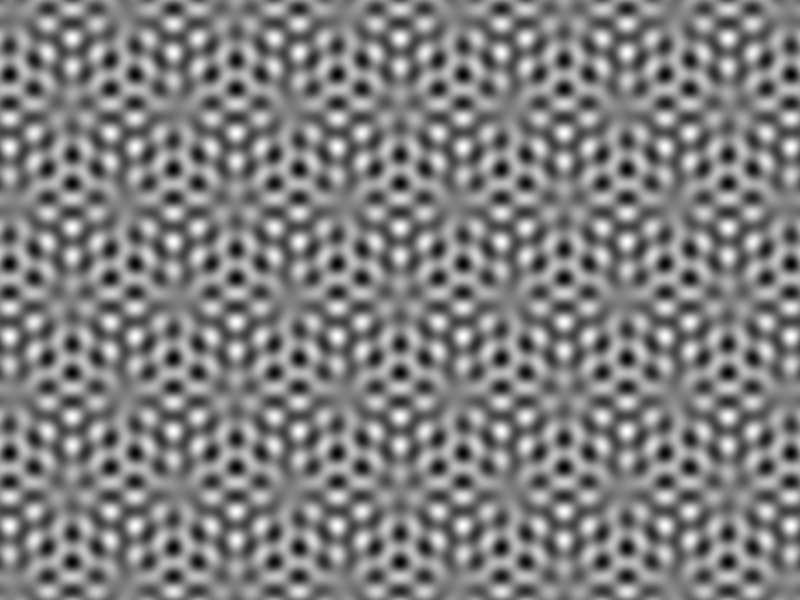}\\[10pt]
\includegraphics[width=0.3\hsize]{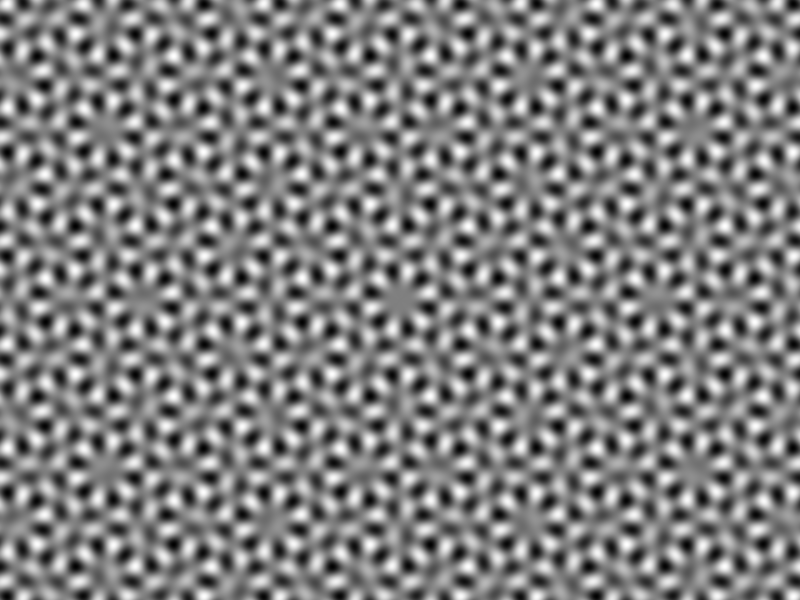}
\includegraphics[width=0.3\hsize]{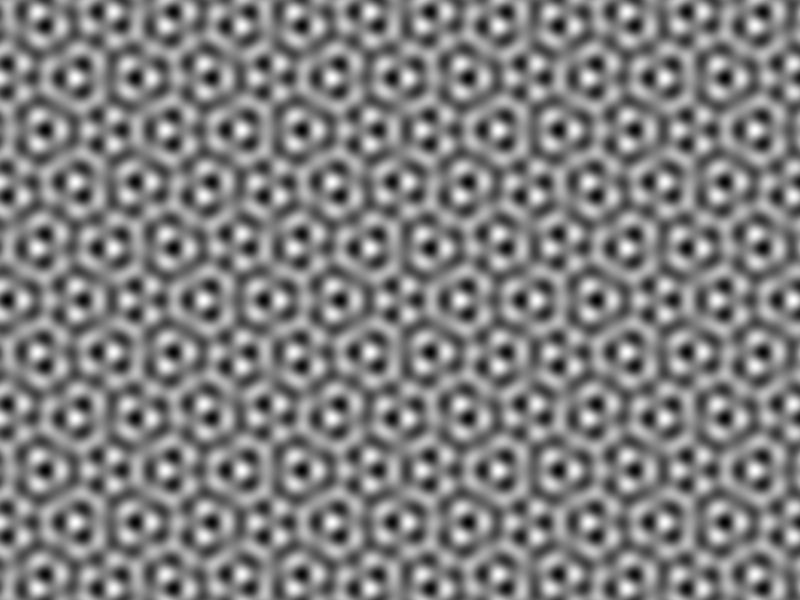}
\includegraphics[width=0.3\hsize]{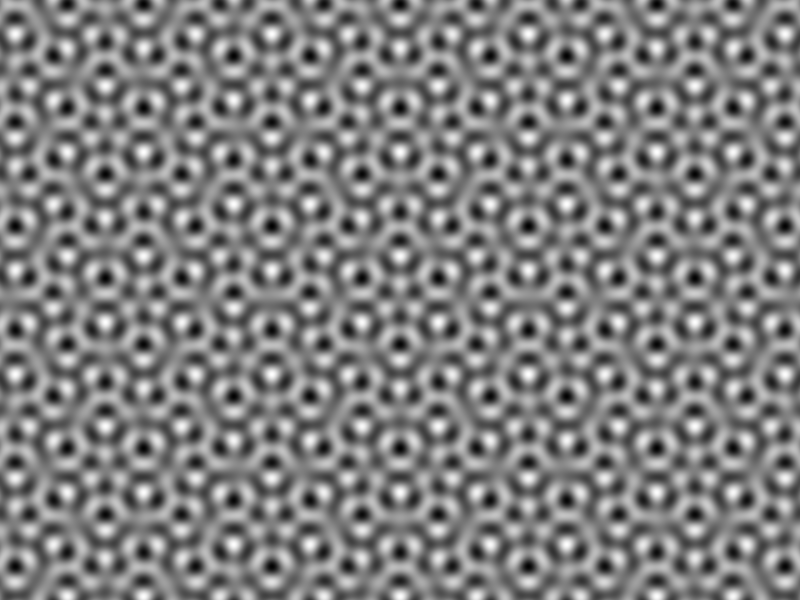}
\end{center}
\caption{Examples of quasipatterns: superposition of hexagons with $\chi=0$. 
 Top row: $\alpha=\frac{\pi}{12}=15^{\circ}$;
 bottom row: $\alpha=25.66^{\circ}$ ($\cos\alpha=\frac{1}{4}\sqrt{13}$). 
 Left: \red{QP-}anti-hexagons; center: \red{QP-}super-triangles; right: \red{QP-}anti-triangles.}
\label{fig:example_chizero_patterns}
\end{figure}

\begin{remark}
Let us observe that \red{equal amplitude QP-super-hexagons} for $\theta_{j}=0$, $j=1,\dots,6$
were already obtained for $\chi=0$ in~\cite{Iooss2019}.

In the case $\chi=0$, the \red{unequal amplitude} solutions do not exist.
The original system~\cref{eq:SHeq} is equivariant under the
symmetry $\mathbf{S}$, which
implies that in~\cref{eq:QP_case},  $f_{1}$ and $f_{2}$ are
respectively even and odd in $(q_{1},q_{4})$. For $\varepsilon =\delta $ the
bifurcation system reduces to two equations of the form
 \begin{eqnarray*}
 \mu  &=&
   f_{1}(\mu,\varepsilon ^{2},q_{1},q_{4}) + \varepsilon e^{-i\Theta_{1}}f_{2}(\mu,\varepsilon ^{2},q_{1},q_{4}) \\
 \mu  &=& 
   f_{1}(\mu,\varepsilon ^{2},q_{4},q_{1})+\varepsilon e^{-i\Theta_{4}}f_{2}(\mu,\varepsilon ^{2},q_{4},q_{1}),
\end{eqnarray*}
and we may observe new quasipattern solutions, illustrated in \cref{fig:example_chizero_patterns}.
\red{The names here are analagous to the related periodic patterns~\cite{Dionne1997}.}

\emph{QP-anti-hexagons} are obtained for (also obtained in \cite{Iooss2019}) 
\begin{eqnarray*}
\theta _{j} &=&0,\qquad j=1,2,3, \\
\theta _{j} &=&\pi ,\qquad j=4,5,6,
\end{eqnarray*}
which leads to
\begin{eqnarray*}
e^{-i\Theta _{1}} &=&1,\text{ }e^{-i\Theta _{4}}=-1, \\
q_{1} &=&\varepsilon ^{3}=-q_{4},
\end{eqnarray*}
and the parity properties of $f_{j}$ give only one bifurcation equation
\begin{equation*}
\mu =f_{1}(\mu,\varepsilon ^{2},\varepsilon ^{3},-\varepsilon ^{3})+\varepsilon
f_{2}(\mu,\varepsilon ^{2},\varepsilon ^{3},-\varepsilon ^{3}).
\end{equation*}

\emph{QP-super-triangles} are obtained for
\begin{equation*}
\theta _{j}=\pi /2,\qquad j=1,\dots,6,
\end{equation*}
which leads to
\begin{eqnarray*}
e^{-i\Theta _{1}} &=&e^{-i\Theta _{4}}=i, \\
q_{1} &=&-i\varepsilon ^{3}=q_{4},
\end{eqnarray*}
and it is clear that we have only one real bifurcation equation, with evenness (resp. oddness) with respect to the two last arguments of $f_1$ (resp. $f_2$) leading to
\begin{equation*}
\mu =f_{1}(\mu,\varepsilon ^{2},-i\varepsilon ^{3},-i\varepsilon
^{3})+i\varepsilon f_{2}(\mu,\varepsilon ^{2},-i\varepsilon ^{3},-i\varepsilon
^{3}). 
\end{equation*}

\emph{QP-anti-triangles} are obtained for
\begin{eqnarray*}
\theta _{j} &=&\pi /2\qquad j=1,2,3, \\
\theta _{j} &=&-\pi /2,\qquad j=4,5,6,
\end{eqnarray*}
which leads to
\begin{eqnarray*}
e^{-i\Theta _{1}} &=&i,\text{ }e^{-i\Theta _{4}}=-i, \\
q_{1} &=&-i\varepsilon ^{3}=-q_{4},
\end{eqnarray*}
and the parity properties of $f_{j}$ give only one real bifurcation equation
\begin{equation*}
\mu =f_{1}(\mu,\varepsilon ^{2},-i\varepsilon ^{3},i\varepsilon
^{3})+i\varepsilon f_{2}(\mu,\varepsilon ^{2},-i\varepsilon ^{3},i\varepsilon
^{3}). 
\end{equation*}
All these cases lead to series for $u$ and $\mu$, respectively odd and even in $\varepsilon$, and hence quasiperiodic anti-hexagons, super-triangles and anti-triangles in~\cref{eq:SHeq} for~$\alpha\in\Eii$ and for $\chi=0$. \red{Using the same arguments as above, we can say that these QP-anti-hexagons etc.\ are solutions of the \hbox{PDE} with $\chi=0$.}
\end{remark}

\subsubsection{Periodic case: Higher orders}
\label{sec:periodic_case_high_orders}

In this case we have more resonant terms in the bifurcation equation, as
seen in~\cref{eq:struct_P1_periodic_case}. \red{These resonant terms introduce relations
between the phases of the complex amplitudes, so the periodic superposed hexagon solutions
come in two-parameter, rather than four-parameter, families.} 
We consider here only the
\red{equal amplitude solutions}, with $\varepsilon=\delta$, but
even in this case there are two sub-types of solutions: \red{super-hexagon solutions, 
and \emph{triangular superlattice} solutions,
where the phase relationships depend on amplitude. The triangular superlattice solutions we
find are generalizations of those found by~\cite{Silber1998}; the name
comes from the triagular appearance of the $(a,b)=(3,2)$ version of this
periodic pattern (see \cref{fig:fw_examples}a and~\cite{Kudrolli1998}).}

\paragraph{Super-hexagons}
We notice that, in setting
 \begin{equation*}
 z_{j} =\varepsilon e^{i\theta _{j}},\quad
 \varepsilon>0, \quad j=1,\dots,6
 \end{equation*}
and taking
 \begin{equation}
 \theta _{1}=\theta _{2}=\theta _{3}=-\theta _{4}=-\theta _{5}=-\theta
 _{6}=k\frac{\pi}{3}  \label{eq:arguments_periodic}
 \end{equation}
we have \ $q_{1}=q_{4}=(-1)^{k}\varepsilon ^{3}$ and we can check  that the nine sets $G_{j}$ of invariant monomials satisfy (see \cref{app:proof_of_periodic_case})
\begin{align*}
G_1&=\varepsilon^{2a}, &
G_2&=G'_2=\varepsilon^{3a-b}e^{i(a+b)k\pi }, &
G_3&=G'_3=\varepsilon^{2a+b}e^{ibk\pi }, \\
G_4&=\varepsilon ^{4a-2b}, &
G_5&=G'_5=\varepsilon^{3a}e^{iak\pi }, &
G_6&=\varepsilon^{2a+2b},
\end{align*}
all these monomials being real.
In \cref{app:proof_of_periodic_case} we show that each group on the same line above is invariant under the actions of $\mathbf{R}_{\pi/3}$ and $\mathbf{\tau}$. It then follows that the system of bifurcation equations reduces to only one equation with real coefficients, as in the quasiperiodic case for the first solutions.
We have now a solution of the form
\begin{eqnarray*}
z_{1} &=&z_{2}=z_{3}=\varepsilon e^{i\theta }, \\
z_{4} &=&z_{5}=z_{6}=\varepsilon e^{-i\theta },\text{ }\theta =k\frac{\pi}{3},\text{
\ }k=0,\dots,5.
\end{eqnarray*}
The conclusion is that the power series starting as in~\cref{eq:solu1}  for $\mu $ in terms of $\varepsilon $
is still valid for the periodic case (the modifications occuring at high
order), provided we restrict the choice of arguments $\theta_{j}$ 
as~\cref{eq:arguments_periodic}. We show in~\cref{app:translations} that 
solutions with $k=0,2,4$ or with $k=1,3,5$ may be obtained from one of them, in acting a suitable translation $\mathbf{T}_{\mathbf{\delta}}$. It follows that we only find two different bifurcating patterns, corresponding to opposite signs of~$\mu$. Moreover, we notice that the solution obtained for $k=0$ is changed into the solution obtained for $k=3$ by acting the symmetry $\mathbf{S}$ on it, and changing $\chi$ into $-\chi$.
Finally, notice that
 since the Lyapunov--Schmidt method applies in this case, the series
converges, for $\varepsilon $ small enough. The above solutions have arguments $\theta_{j}=0$ or $\pi$ that do not depend on parameters~$(\mu,\chi)$; these solutions correspond to super-hexagons. 

\paragraph{Triangular superlattice solutions}
Now, in \cite{Silber1998} other solutions were found for $(a,b)=(3,2)$, just taking into account of terms of order five in the bifurcation system. Let us show that these solutions exist indeed for any $(a,b)$ and taking into account of all resonant terms.

Let us consider the particular cases with
\begin{equation*}
z_{j}=\varepsilon e^{i\theta },
\end{equation*}
then the nine sets $G_{j}$ of monomials defined in \cref{app:proof_of_periodic_case} satisfy 
\begin{eqnarray*}
G_{1} &=&\varepsilon ^{2a}e^{i(4b-2a)\theta },\text{ }\mathbf{R}_{\pi
/3}G_{1}=\overline{G_{1}},\text{ }\tau G_{1}=G_{1}, \\
G_{2} &=&G_{2}^{\prime }=\varepsilon ^{3a-b}e^{i(a+b)\theta },\text{ }
\mathbf{R}_{\pi/3}G_{2}=\overline{G_{2}},\text{ }\tau G_{2}=G_{2}, \\
G_{3} &=&\overline{G_{3}^{\prime }}=\varepsilon ^{2a+b}e^{i(2a-b)\theta },
\text{ }\mathbf{R}_{\pi/3}G_{3}=\overline{G_{3}},\text{ }\tau G_{3}=G_{3},
\\
G_{4} &=&\varepsilon ^{4a-2b}e^{i(4a-2b)\theta },\text{ }\mathbf{R}_{\pi
/3}G_{4}=\overline{G_{4}},\text{ }\tau G_{4}=G_{4}, \\
G_{5} &=&\overline{G_{5}^{\prime }}=\varepsilon ^{3a}e^{i(2b-a)\theta },
\text{ \ }\mathbf{R}_{\pi/3}G_{5}=\overline{G_{5}},\text{ }\tau G_{5}=G_{5},
\\
G_{6} &=&\varepsilon ^{2a+2b}e^{i(2a+2b)\theta },\text{ \ }\mathbf{R}_{\pi
/3}G_{6}=\overline{G_{6}},\text{ }\tau G_{6}=G_{6}.
\end{eqnarray*}
Then the first bifurcation equation becomes
\begin{equation}
\mu =f_{3}+\varepsilon e^{-3i\theta }f_{4}+\frac{G_{1}}{\varepsilon ^{2}}
f_{G_{1}}+\frac{G_{2}}{\varepsilon ^{2}}f_{G_{2}}+\frac{\overline{G_{4}}}{
\varepsilon ^{2}}f_{G_{4}}+\frac{\overline{G_{5}}}{\varepsilon ^{2}}
f_{G_{5}}+\frac{\overline{G_{6}}}{\varepsilon ^{2}}f_{G_{6}},\label{eq:bifequper2}
\end{equation}
with all $f_{j}$ functions of $(\chi,\mu,\varepsilon ^{2},\varepsilon
^{3}e^{3i\theta },\varepsilon ^{3}e^{-3i\theta },G_{1},\overline{G_{1}}
,G_{2},\overline{G_{2}},G_{3},\overline{G_{3}},G_{4},\overline{G_{4}},G_{5},
\overline{G_{5}},G_{6},\overline{G_{6}})$.
They have real coefficients, and are invariant under
symmetry $\mathbf{\tau }$, while the arguments are changed into their
complex conjugate by symmetry $\mathbf{R}_{\pi/3}$. It follows that the
bifurcation system reduces to only one complex (because of the occurrence of $\theta$)  equation, where we can express
the unknowns $(\mu,\theta)$ as functions of~$\varepsilon$. Then truncated at cubic order in $(\mu,\varepsilon)$ this equation reads
\begin{equation*}
\mu =f_{3}^{(0)}(\chi,\varepsilon ^{2},\varepsilon ^{3}e^{3i\theta
},\varepsilon ^{3}e^{-3i\theta })+\varepsilon e^{-3i\theta }f_{4}^{(0)}(\chi
,\varepsilon ^{2}), 
\end{equation*}
which is a nice perturbation at order~$\varepsilon^3$ of
the known equation 
\begin{equation*}
    \mu=(\alpha_1+2\alpha_2+\alpha_4+\alpha_5+\alpha_6)\varepsilon^2+2\chi \varepsilon e^{-3i\theta}.
\end{equation*}
This leads to the two types of solutions:
\begin{eqnarray*}
e^{3i\theta } &=&\pm 1, \\
\mu  &=&f_{3}^{(0)}(\chi,\varepsilon ^{2},\pm \varepsilon ^{3},\pm
\varepsilon ^{3})\pm \varepsilon f_{4}^{(0)}(\chi,\varepsilon ^{2}).
\end{eqnarray*}
These solutions are not degenerate, so that, if we consider
the complex equation~\cref{eq:bifequper2}, the implicit function theorem
applies for solving with respect to $(\mu,\theta)$ in \red{convergent} powers series of $\varepsilon$. This gives solutions
of the form
\begin{eqnarray*}
\theta _{l}(\varepsilon ) &=&l\frac{\pi}{3}+\mathcal{O}(\varepsilon ),\text{ }
l=0,1,2,3,4,5 \\
\mu  &=&f_{3}^{(0)}(\chi,\varepsilon ^{2},(-1)^{l}\varepsilon
^{3},(-1)^{l}\varepsilon ^{3})+(-1)^{l}\varepsilon f_{4}^{(0)}(\chi
,\varepsilon ^{2})+\mathcal{O}(\varepsilon ^{4}).
\end{eqnarray*}
Now, we observe that the cases $l=0,3$ lead to a real bifurcation equation, which fixes the argument $\theta=0$ or $\pi$. This recovers the \red{super-hexagon} solutions, already found. The remaining cases are the solutions suggested by \cite{Silber1998} (for $(a,b)=(3,2)$, not including all resonant terms).
Let us sum up the results in the following

\begin{theorem}[Periodic equal amplitude superposed hexagons]
\label{thm:superposed_hexagons_periodic} Assume $\alpha \in \EO\cap
\Ep$, then for $\varepsilon$ small enough, and $
\chi $ fixed, we can build convergent\ power series solutions
of~\cref{eq:SH}, of the form
\begin{align}
u(\varepsilon ,\chi ,k) &=\varepsilon u_{1}+\sum_{n\geq 2}\varepsilon
^{n}u_{n}(\chi ,k),\text{ \ }u_{n}\bot e^{i\bk_{j}\cdot \bx},
\text{ \ }j=1,\dots,6,\text{ \ }n\geq 2  \nonumber \\
\mu (\varepsilon ,\chi ,k) &=(-1)^{k}2\chi \varepsilon +\mu_2(\chi)\varepsilon
^{2}+\sum_{n\geq 3}\varepsilon ^{n}\mu_{n}(\chi ,k),
\label{eq:solution_periodic_case} \\
u_{n}(-\chi ,k) &=(-1)^{n}u_{n}(\chi ,k),\text{ \ }\mu_{n}(-\chi
,k)=(-1)^{n}\mu_{n}(\chi ,k);  \nonumber
\end{align}
where $\mu$ is even in $((-1)^{k} \varepsilon ,\chi)$  and $\mu_2(\chi)$ is defined at \cref{thm:superposed_hexagons}
and such that, 
for \red{super-hexagon} solutions
\begin{equation*}
    u_{1} =\sum_{j=1,\dots,6}e^{i(\bk_{j}\cdot \bx+\theta _{j})}+c.c.,
\text{ }\theta _{1}=\theta _{2}=\theta _{3}=-\theta _{4}=-\theta
_{5}=-\theta _{6}=k\pi,\text{  }k=0,\text{ or }1.
\end{equation*}
For \red{triangular superlattice} solutions, we have
\begin{equation*}
    u_{1} =\sum_{j=1,\dots,6}e^{i(\bk_{j}\cdot \bx+\theta)}+c.c.,
\text{ }\theta(\varepsilon,\chi,k)=k\frac{\pi}{3}+\sum_{n\geq1}\varepsilon^{n}\theta_{n}(\chi,k),
,k=1,2,4,5.
\end{equation*}
\end{theorem}

\begin{figure}
\begin{center}
\includegraphics[width=0.3\hsize]{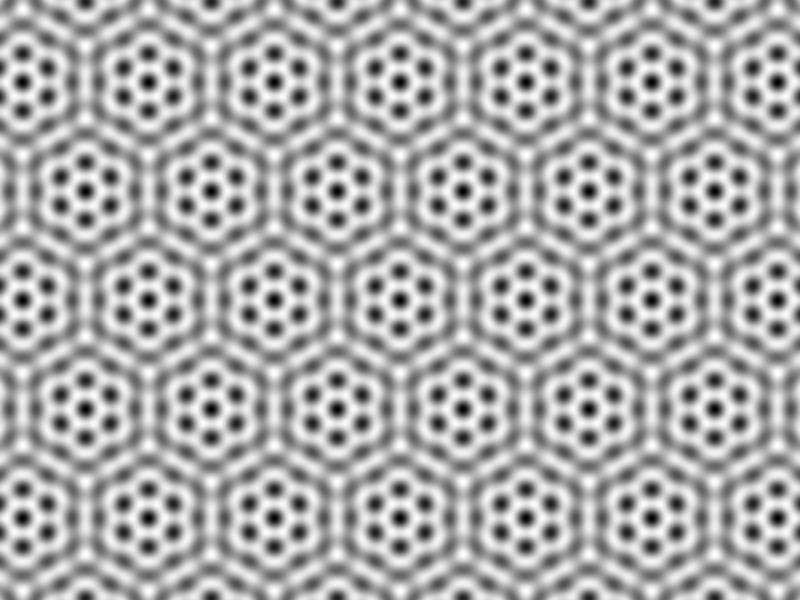}
\includegraphics[width=0.3\hsize]{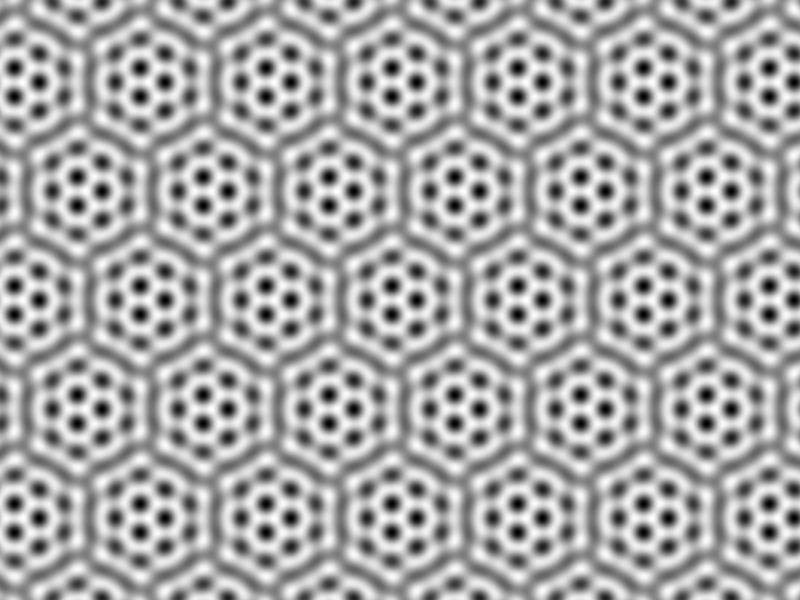}
\includegraphics[width=0.3\hsize]{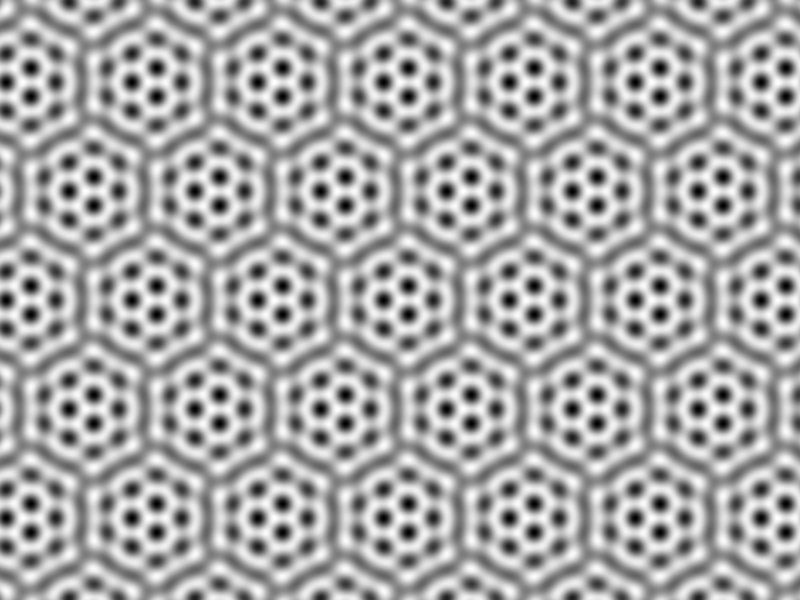}\\[10pt]
\includegraphics[width=0.3\hsize]{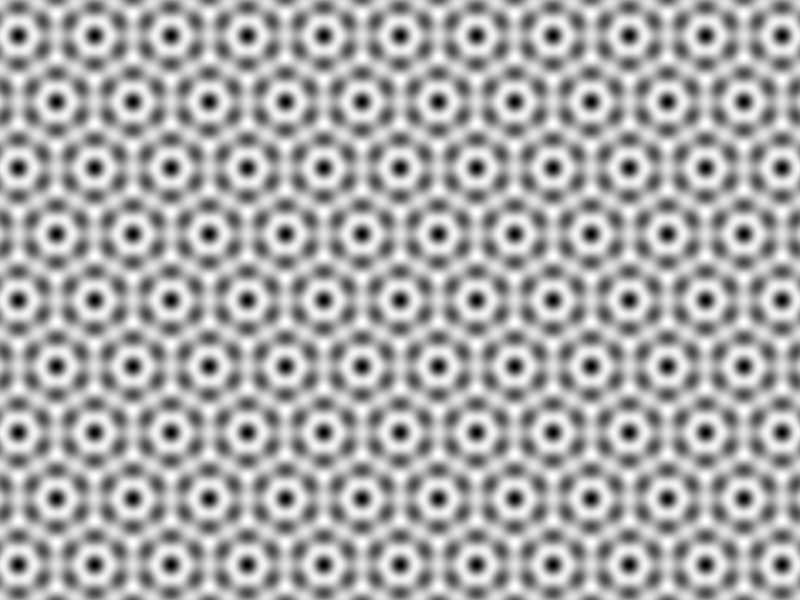}
\includegraphics[width=0.3\hsize]{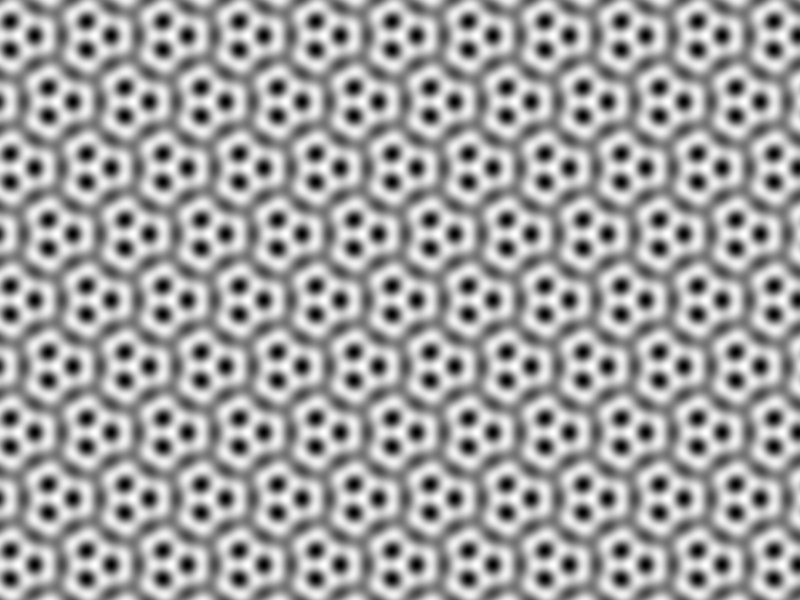}
\includegraphics[width=0.3\hsize]{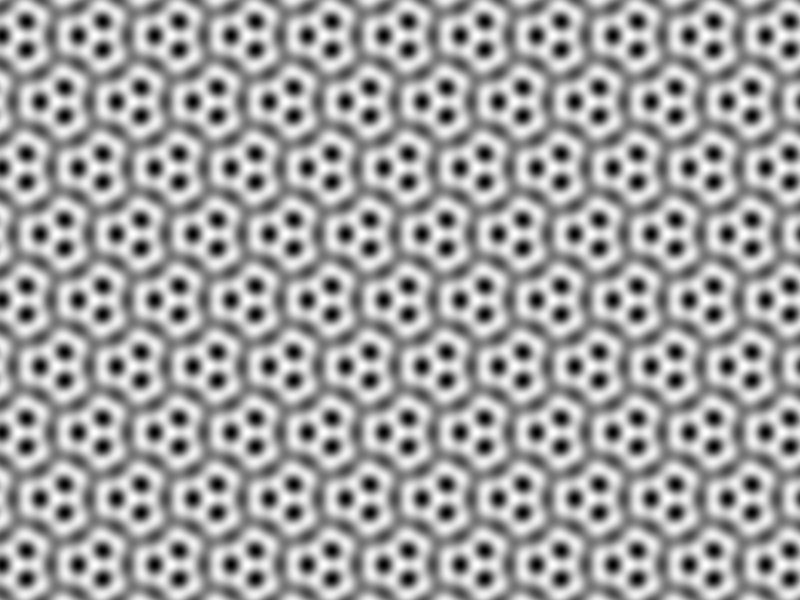}\\[10pt]
\includegraphics[width=0.9\hsize]{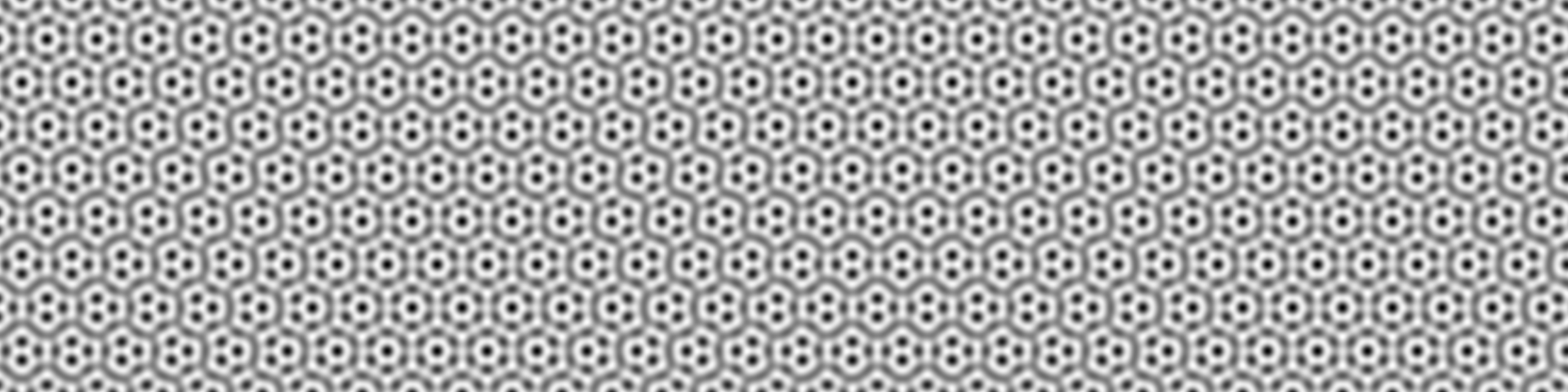}
\end{center}
\caption{Examples of periodic patterns: superposition of hexagons.
Top row: $\alpha=13.17^\circ$ ($\cos\alpha=\frac{37}{38}$, $(a,b)=(5,3)$);
middle row: $\alpha=21.79^\circ$ ($\cos\alpha=\frac{13}{14}$, $(a,b)=(3,2)$ -- 
see also \cref{fig:fw_examples}a). 
For these, the left column \red{(super-hexagons)} 
has $\theta_j=0$ for $j=1,\dots,6$. \red{The middle and right (superlattice triangles) have
$\theta_j=\frac{2\pi}{3}$ and $\theta_j=\frac{4\pi}{3}$ respectively}. The
bottom row shows a related quasiperiodic example with $\alpha=21.00^{\circ}$,
close to $21.79^\circ$, showing long-range modulation between the three
periodic patterns in the middle row.}
 \label{fig:example_periodic_patterns}
 \end{figure}

 \begin{remark}
For \red{triangular superlattice solutions}, the \red{phases of the amplitudes}
are not independent of the parameters, in contrast to the \red{super-hexagon}
solutions. These patterns are illustrated in
\cref{fig:example_periodic_patterns}. The figure includes (middle row) periodic
patterns with $\alpha=21.79^\circ$ and (bottom row) a quasiperiodic pattern with
$\alpha=21^\circ$, showing how, with a slightly different value of~$\alpha$, the
quasiperiodic pattern modulates between the three periodic solutions with
$l=0,2,4$.
 \end{remark}

 \begin{remark}
In the $\chi=0$ case, we can recover all the solutions found
by~\cite{Dionne1997} using these ideas.
 \end{remark}

\subsection{Hexa-rolls: superposition of hexagons and rolls}

\red{As in \S\labelcref{sec:superp_hexa}, we start with the cubic truncation of the 
quasiperiodic and periodic cases together, then consider the effect of higher order terms.}
Here we consider the case where $q_{1}\neq 0$ and $q_{4}=0$ 
in~\cref{eq:bifurcequcub}, so that we assume now
 \begin{equation*}
 q_{1}\neq 0,\text{ \ }z_{4}\neq 0,\text{ \ }z_{5}=z_{6}=0.
 \end{equation*}
Then the system~\cref{eq:bifurcequcub}
reduces to 4 equations
 \begin{align}
 2\chi \overline{q_{1}} &= u_{1}[\mu -\alpha_{1}u_{1}-\alpha_{2}(u_{2}+u_{3}) - \alpha_{4}u_{4}], \nonumber \\
 2\chi \overline{q_{1}} &= u_{2}[\mu -\alpha_{1}u_{2}-\alpha_{2}(u_{1}+u_{3})-\alpha_{6}u_{4}],  \nonumber \\
 2\chi \overline{q_{1}} &=u_{3}[\mu -\alpha_{1}u_{3}-\alpha_{2}(u_{1}+u_{2})-\alpha_{5}u_{4}],  \label{eq:hexa-rolls1} \\
 0 &=\mu - \alpha_{1}u_{4} - \alpha_{4}u_{1}-\alpha_{5}u_{3}-\alpha_{6}u_{2},  \nonumber
 \end{align}
where again this implies that $q_1$ is real.
Below, we study solutions of the bifurcation problem, built on a lattice
spanned by the four wave vectors $\bk_{1}$, $\bk_{2}$, $\bk_{3}$,
and~$\bk_{4}$, 
\red{and so we find solutions composed of a superposition
of hexagons and rolls. Unlike in the super-hexagon cases above, the three
amplitudes ($|z_1|$, $|z_2|$ and $|z_3|$) of the hexagonal part of the pattern
are of similar size but will not be exactly equal.} We find two different types
of solution \red{distinguished by the relative magnitudes of the hexaonal and roll parts of the
pattern}. \red{The first type occurs when}
$|\chi|$ is \red{neither too small nor too large} and is such
that rolls dominate the hexagons. \red{The second type occurs} only for
small~$|\chi|$ and is such that rolls and hexagons are more balanced.

\subsubsection{Hexa-rolls: rolls dominate hexagons}
\label{sec:hexarolls_rolls_dominant}

A consistent balance of terms in~\cref{eq:hexa-rolls1} is to have $u_1$, $u_2$ and $u_3$ be~$\mathcal{O}(\mu^2)$, so that $q_1$ is~$\mathcal{O}(\mu^3)$, while $u_4$ is~$\mathcal{O}(\mu)$. With this balance, at leading order we have the reduced system
\begin{align}
2\chi \overline{q_{1}} &=u_{1}[\mu -\alpha _{4}u_{4}],  \nonumber \\
2\chi \overline{q_{1}} &=u_{2}[\mu -\alpha _{6}u_{4}],
\label{eq:reduced_simplified} \\
2\chi \overline{q_{1}} &=u_{3}[\mu -\alpha _{5}u_{4}],  \nonumber \\
0 &=\mu -\alpha _{1}u_{4},  \nonumber
\end{align}
which leads to
\begin{eqnarray*}
z_{j} &=&\sqrt{u_{j}}e^{i\theta _{j}},
\quad j=1,2,3, \\
u_{j} &=&\mu ^{2}u_{j}^{(0)}, \quad
u_{4}=\frac{\mu }{\alpha_{1}}, \\
\Theta _{1} &=&\theta _{1}+\theta _{2}+\theta _{3}=k\pi ,
\end{eqnarray*}
with
\begin{align}
u_{1}^{(0)} &=\frac{(\alpha _{5}-\alpha _{1})(\alpha _{6}-\alpha _{1})}{
4\chi ^{2}\alpha_{1}^{2}},  \nonumber \\
u_{2}^{(0)} &=\frac{(\alpha _{5}-\alpha _{1})(\alpha _{4}-\alpha _{1})}{
4\chi ^{2}\alpha_{1}^{2}},  \label{eq:solution_I} \\
u_{3}^{(0)} &=\frac{(\alpha _{4}-\alpha _{1})(\alpha _{6}-\alpha _{1})}{
4\chi ^{2}\alpha_{1}^{2}},  \nonumber \\
(-1)^{k} &=\text{sign}[\chi (\alpha _{1}-\alpha _{4})].  \nonumber
\end{align}
The condition for the existence of this solution is that $(\alpha
_{4}-\alpha _{1})$, $(\alpha _{5}-\alpha _{1})$, $(\alpha _{6}-\alpha _{1})$ should be nonzero and
have the same sign. This condition is realized in~\cref{eq:SHeq} provided that
\begin{equation*}
3+\chi ^{2}(c_{1}-c_{\alpha }),\qquad
3+\chi ^{2}(c_{1}-c_{\alpha +}),\qquad
3+\chi ^{2}(c_{1}-c_{\alpha -}),
\end{equation*}
have the same sign, which holds at least for $|\chi|$ not too large. For
applying later the implicit function theorem, we typically need
$|\mu|\ll\min(1,|\chi|)$, \red{so $|\chi|$ should also be not too small}.
Here, for $|\chi|$ not too large, $\alpha_{1}>0$, so the bifurcation is
supercritical in this case.

Now let us consider the full bifurcation system. Setting
\begin{equation}\label{eq:uj_mu}
u_{j}=\mu ^{2}u_{j}^{(0)}(1+x_{j}),\quad
j=1,2,3,\quad 
u_{4}=\frac{\mu }{a_{1}}(1+x_{4}),
\end{equation}
we replace these expressions in~\cref{eq:hexa-rolls1} plus higher order terms
appearing in~\cref{eq:QP_case} or~\cref{eq:struct_P1_periodic_case}, and noticing that we obtain a real system of four equations \red{in all periodic and quasiperiodic cases except in the periodic case when $a-b=1$, as defined in \cref{lem:Ep_ab}}. 
 \begin{remark}
 \red{In the case $\alpha=\frac{\pi}{6}$, this
 combination of hexagons and rolls was reported 
 by~\cite{Malomed1989,Jiang2020a,Subramanian2021b}.}
 \end{remark}
 \begin{remark}
 \red{In the periodic case when $a-b=1$, a careful examination of high order
 resonant terms (as defined in \cref{app:proof_of_periodic_case}) shows that
 there remains six equations, instead of four. We might compute some new
 solution looking like the superposed hexagons and rolls (but with small
 $|z_5|$ and~$|z_6|$), however there are not
 strictly of the required form since $q_1 q_4 \neq 0$. We do not pursue these
 solutions further here.}
 \end{remark}
Then, dividing the first three
equations in~\cref{eq:reduced_simplified} (with~\cref{eq:uj_mu}) by~$\mu^{3}$, dividing the fourth one by~$\mu$, and computing the linear part in~$x_{j}$, we obtain
\begin{align}
a(x_{1}+x_{2}+x_{3})-u_{1}^{(0)}((1-\frac{\alpha _{4}}{\alpha _{1}})x_{1}-
\frac{\alpha _{4}}{\alpha _{1}}x_{4}) &=h_{1},  \nonumber \\
a(x_{1}+x_{2}+x_{3})-u_{2}^{(0)}((1-\frac{\alpha _{6}}{\alpha _{1}})x_{2}-
\frac{\alpha _{6}}{\alpha _{1}}x_{4}) &=h_{2},  \label{eq:matrix_Mprime} \\
a(x_{1}+x_{2}+x_{3})-u_{3}^{(0)}((1-\frac{\alpha _{5}}{\alpha _{1}})x_{3}-
\frac{\alpha _{5}}{\alpha _{1}}x_{4}) &=h_{3},  \nonumber \\
x_{4} &=h_{4},  \nonumber
\end{align}
with
\begin{equation*}
a=(-1)^{k}\chi \sqrt{u_{1}^{(0)}u_{2}^{(0)}u_{3}^{(0)}},
\end{equation*}
and all $h_{j}$ have $\mu $ in factor. The left hand side of the 
system~\cref{eq:matrix_Mprime} represents the differential at the origin with respect to $(x_{1},x_{2},x_{3},x_{4})$, defining a matrix~$M'$ that needs to be inverted in order to use the implicit function theorem. The determinant of matrix~$M'$ can be computed and it is
\begin{equation*}
\frac{[3(-1)^{k}\text{sign}(\chi )-2]}{128\chi ^{6}\alpha _{1}^{9}}[(\alpha
_{1}-\alpha _{4})(\alpha _{1}-\alpha _{5})(\alpha _{1}-\alpha _{6})]^{3},
\end{equation*}
which is not zero. Therefore the implicit function theorem applies,
so we can find series in powers of $\mu$ for 
$(x_{1},x_{2},x_{3},x_{4})$ solving the full bifurcation system in both the quasiperiodic case~\cref{eq:QP_case} and the periodic case~\cref{eq:struct_P1_periodic_case}. We can state the following 
\begin{theorem}[Hexa-rolls: superposed hexagons and rolls with rolls dominant]
\label{thm:hexa-rolls_(I)} Assume that $\alpha \in \EO$, \red{and in case of a periodic lattice assume $a-b>1$}. Then for fixed values of $\chi$
such that 
\begin{equation*}
(\alpha _{4}-\alpha _{1}),(\alpha _{5}-\alpha _{1}),(\alpha _{6}-\alpha _{1})
\end{equation*}
are nonzero and have the same sign, and for $\mu$ close enough to $0$, we can build a three-parameter formal 
power series in~$\varepsilon$ solution of~\cref{eq:SHeq} of the form
\begin{eqnarray*}
u(\varepsilon ,\Theta ,\chi ,j) &=&u_{1}(\varepsilon ,\Theta ,\chi
,j)+\sum\limits_{n\geq 3}\varepsilon ^{n}u_{n}(\chi ,\Theta ,j),\quad
u_{2p+1}\bot e^{i\bk_{j}\cdot \bx},
\quad j=4,\text{ or }5\text{ or }6, \\
u_{1}(\varepsilon ,\Theta ,\chi ,j) &=&\varepsilon e^{i(\bk_{j}\cdot 
\mathbf{x}+\theta _{j})}+\alpha _{1}\varepsilon ^{2}\sum_{m=1,2,3}\sqrt{
u_{m}^{(0)}}e^{i(\bk_{m}\cdot \mathbf{x}+\theta _{m})}+c.c. \\
\Theta  &=&(\theta _{1},\theta _{2},\theta _{3},\theta _{j}),\text{ \ }
\theta _{1}+\theta _{2}+\theta _{3}=k\pi ,\text{ }k=0\text{ or }1, \\
\mu (\varepsilon ,\chi ,j) &=&\alpha _{1}\varepsilon ^{2}+\sum_{n\geq 2}\mu
_{2n}(\chi ,j)\varepsilon ^{2n},\text{ even in }\varepsilon ,
\end{eqnarray*}
where $u_{m}^{(0)}$ and $k$ are determined in~\cref{eq:solution_I}. 
For $\alpha_1>0$ the bifurcation is supercritical with $\mu>0$. In the case $\alpha_1<0$, subcritical patterns can be found with $\mu<0$. In the
quasiperiodic case ($\alpha\in\Eii$), these solutions give quasipatterns using the techniques of~\cite{Iooss2019}. 
In the periodic case ($\alpha\in\Ep\cap\EO$), the classical
Lyapunov--Schmidt method give periodic pattern solutions of the PDE~\cref{eq:SHeq}. In both cases, 
the freedom left for~$\Theta$ corresponds to an arbitrary choice for translations $\mathbf{T}_{\delta}$ of the hexagons, and the arbitrary choice of $\theta_j$ ($j=4,5,6$) allows an arbitrary relative translation of the rolls. 
\Cref{fig:example_hexaroll_patterns} shows quasiperiodic examples of~$u_1$ \red{(QP-hexa-rolls)}.
\end{theorem}

\begin{figure}[tbp]
\begin{center}
\includegraphics[width=0.4\hsize]{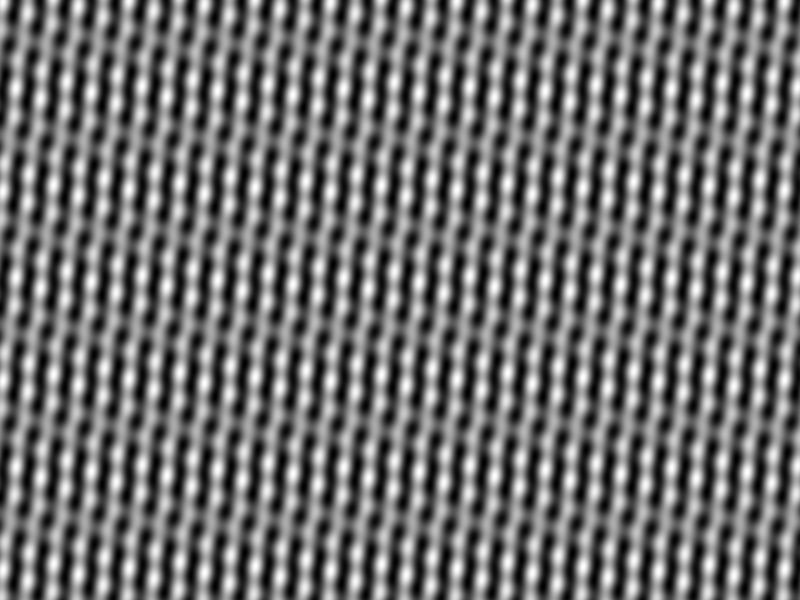}
\includegraphics[width=0.4\hsize]{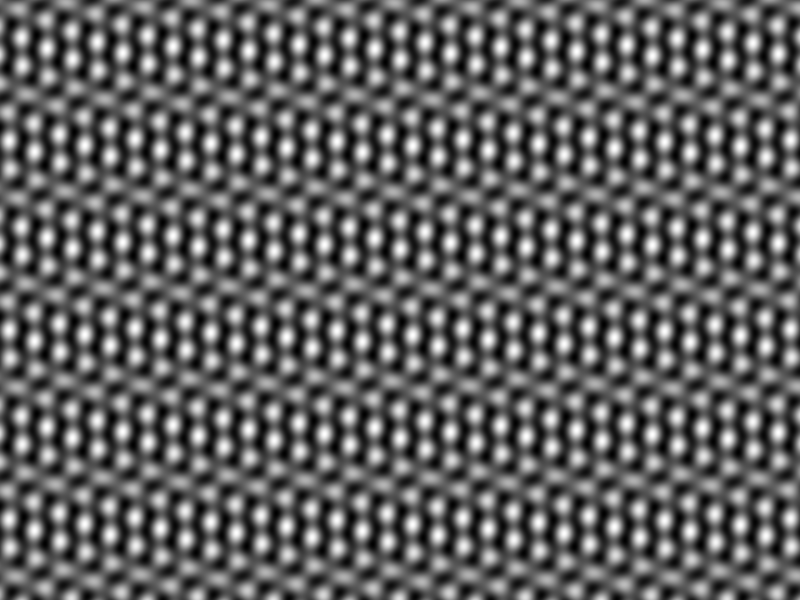}\\[10pt]
\includegraphics[width=0.4\hsize]{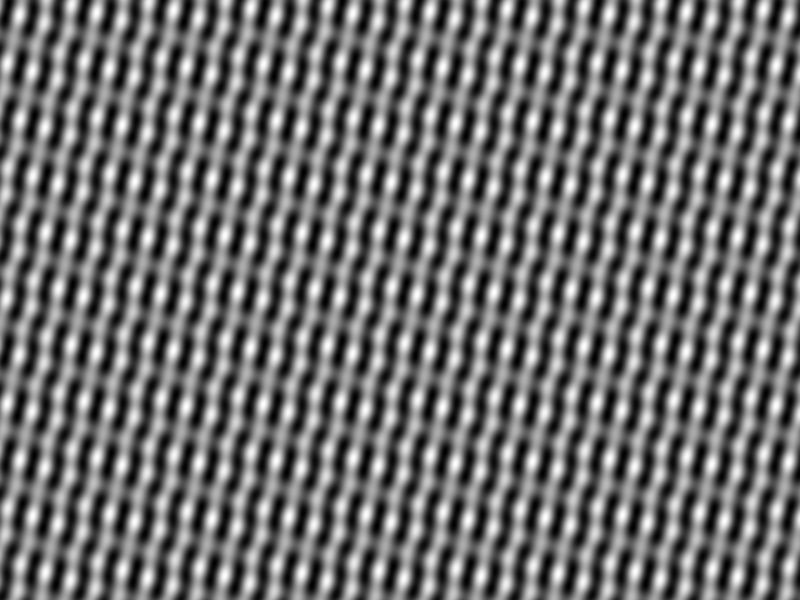}
\includegraphics[width=0.4\hsize]{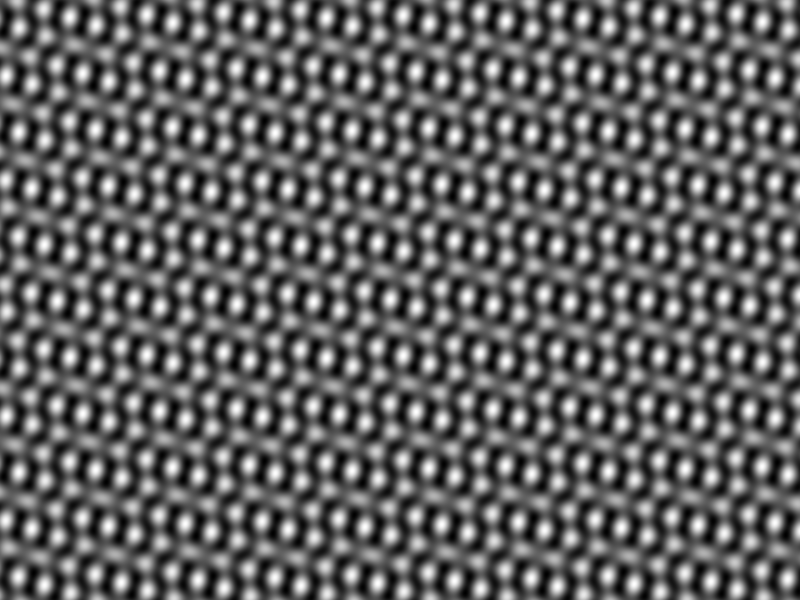}\\[0pt]
\end{center}
\caption{Examples of quasipatterns: superposition of hexagons and rolls.
 Top row: $\alpha=\frac{\pi}{12}=15^{\circ}$;
 bottom row: $\alpha=25.66^{\circ}$ ($\cos\alpha=\frac{1}{4}\sqrt{13}$).
 Left: \red{QP-hexa-rolls with rolls dominating hexagons};
 right: \red{QP-hexa-rolls with rolls and hexagons in balance}.}
\label{fig:example_hexaroll_patterns}
\end{figure}

\begin{remark}
These \red{hexa-roll} solutions are new, even in the case of a periodic lattice. They have the
\red{surprising} feature in the periodic case of allowing arbitrary relative
translations between the hexagons and rolls. Unlike the \red{super-hexagon}
solutions, these solutions require a condition on the cubic
coefficients to be satisfied in order to exist. They were not found
by~\cite{Dionne1997} since there the equivariant branching lemma was used,
which finds only solutions that are characterized by a single amplitude (these
solutions have two) and that exist for all non-degenerate values of the cubic
coefficients (here the cubic coefficients must satisfy an inequality).
\end{remark}

\subsubsection{Hexa-rolls: rolls and hexagons balance}
\label{sec:hexarolls_balanced}
\red{With small~$|\chi|$, solutions can be found where the rolls and hexagons are of similar size.}
Let us consider the system~\cref{eq:hexa-rolls1}, without the terms with $\chi ^{2}$ in
coefficients, and set
\begin{align*}
z_{1} &=\varepsilon e^{i\theta _{1}}, \qquad
z_{2}=\varepsilon e^{i\theta_{2}}, &
z_{3}&=\varepsilon \zeta _{3}e^{i\theta _{3}}, &
\theta_{1}&+\theta _{2}+\theta _{3}=k\pi , &
\varepsilon &>0, \\
u_{4} &=|z_{4}|^{2}=\varepsilon ^{2}u_{4}^{(0)}, &
z_{5}&=z_{6}=0, &
\mu &=\varepsilon ^{2}\mu ^{(0)},&
\chi &=\varepsilon \kappa ,
\end{align*}
then, after division by $\varepsilon^{4}$ the first equations, and by $\varepsilon^{2}$ the fourth one, this gives
\begin{eqnarray*}
2\kappa (-1)^{k}\zeta _{3} &=&\mu ^{(0)}-9-6\zeta _{3}^{2}-6u_{4}^{(0)}, \\
2\kappa (-1)^{k}\zeta _{3} &=&\zeta _{3}^{2}[\mu ^{(0)}-3\zeta
_{3}^{2}-12-6u_{4}^{(0)})], \\
0 &=&\mu ^{(0)}-3u_{4}^{(0)}-12-6\zeta _{3}^{2}.
\end{eqnarray*}
Eliminating $\mu ^{(0)}$ and $u_{4}^{(0)}$ leads to
\begin{equation*}
u_{4}^{(0)}=1-\frac{2\kappa }{3}\zeta _{3}(-1)^{k},
\end{equation*}
and
\begin{equation*}
(3\zeta _{3}+2\kappa (-1)^{k})(\zeta _{3}^{2}-1)=0.
\end{equation*}

\paragraph{Balanced hexa-rolls type 1}

For the solution $\zeta _{3}=1$, we obtain
\begin{equation}
z_{3}=\varepsilon e^{i\theta _{3}},\text{ }u_{4}^{(0)}=1+\frac{2\kappa }{3}
(-1)^{k+1},\text{ }\mu ^{(0)}=21+2\kappa (-1)^{k+1},  \label{eq:sol_IIa}
\end{equation}
for which we need to satisfy $u_{4}^{(0)}>0$, i.e.,
\begin{equation}
\kappa (-1)^{k}<\frac{3}{2},  \label{eq:cond_IIa}
\end{equation}
and we observe that $\mu ^{(0)}>0$ (supercritical bifurcation).
\red{These solutions have the three hexagon amplitudes equal at leading order.}

Now, we observe that the solution $\zeta _{3}=-1$ may be obtained 
from~\cref{eq:sol_IIa} in adding $\pi $ to $\theta _{3}$ and change $k$ into $k+1$. It
follows that this does not give a new solution.

\paragraph{Balanced hexa-rolls type 2}

For the solution $\zeta _{3}=\frac{2}{3}\kappa (-1)^{k+1}$, we obtain
\begin{equation}
 z_{3}=\frac{2}{3}\kappa (-1)^{k+1}\varepsilon ,
 \text{ }
 u_{4}^{(0)}=1+\frac{4}{9}\kappa ^{2},
 \text{ }
 \mu ^{(0)}=15+4\kappa ^{2},  \label{eq:sol_IIc}
\end{equation}
where there is no restriction on $\kappa$, and we observe that 
$\mu^{(0)}>0$ (supercritical bifurcation).
\red{These solutions have one of the three hexagon amplitudes different from 
the other two at leading order.}

For proving that these \red{balanced hexa-roll} solutions at leading order provide solutions for the
full system at all orders, let us define
\begin{align}
z_{1} &=\varepsilon e^{i\theta _{1}}(1+x_{1}), &
z_{2} &=\varepsilon e^{i\theta _{2}}(1+x_{2}), &
z_{3} &=\varepsilon \zeta _{3}e^{i\theta_{3}}(1+x_{3}),
\label{eq:perturb_hexa+roll} \\
u_{4} &=\varepsilon ^{2}(u_{4}^{(0)}+v_{4}), &
\mu &=\varepsilon^{2}(\mu ^{(0)}+\nu ), &
z_{5}&=z_{6}=0, \qquad
\theta _{1}+\theta _{2}+\theta _{3}=k\pi, \notag
\end{align}
where $u_{4}^{(0)}$ $\mu ^{(0)}$, and $\zeta _{3}$ are those computed above
in~\cref{eq:sol_IIa}, \cref{eq:sol_IIc}. Replacing these expressions in~\cref{eq:hexa-rolls1},
it is clear that the previously neglected terms play the role of a
perturbation of higher order. Higher orders of the bifurcation equation are
given by~\cref{eq:QP_case} or~\cref{eq:struct_P1_periodic_case}. We notice that the system is real because in
setting~\cref{eq:perturb_hexa+roll}, the monomials $q_{4}$, $q_{j,k}$, $q_{st}^{\prime }$ cancel for all $j,k,s,t$. Hence there are only four
remaining equations in the bifurcation system, with the same form in the
quasiperiodic and in the periodic cases.

Dividing by the suitable power of $\varepsilon $, the linear terms in 
$(x_{1},x_{2},x_{3},v_{4},\nu )$ are, at leading order (replacing $\mu ^{(0)}$
and $u_{4}^{(0)}$ by their values)
\begin{eqnarray*}
&&\nu -6v_{4}+2(3+2\kappa \zeta _{3}(-1)^{k})x_{1}-12(\zeta
_{3}^{2}-1)x_{3}-[2\kappa (-1)^{k}\zeta _{3}+12](x_{1}+x_{2}+x_{3}) \\
&&\nu -6v_{4}+2(3+2\kappa \zeta _{3}(-1)^{k})x_{2}-12(\zeta
_{3}^{2}-1)x_{3}-[2\kappa (-1)^{k}\zeta _{3}+12](x_{1}+x_{2}+x_{3}) \\
&&\nu -6v_{4}+2(3+2\kappa \zeta _{3}(-1)^{k})x_{3}-[2\kappa (-1)^{k}(\zeta_{3})^{-1}+12](x_{1}+x_{2}+x_{3}) \\
&&\nu -3v_{4}-12(\zeta _{3}^{2}-1)x_{3}-12(x_{1}+x_{2}+x_{3}). 
\end{eqnarray*}
The fact that we have a freedom for the choice of the scale $\varepsilon $
allows us to take $x_{1}=0$. So, if we are able to invert the matrix $M$
defined above, acting on $(x_{2},x_{3},v_{4},\nu )$, i.e., solving
\begin{equation*}
M(x_{2},x_{3},v_{4},\nu )^{t}=(h_{1},h_{2},h_{3},h_{4})^{t},
\end{equation*}
with an inverse with a norm of order 1, then this would mean that we can
invert the differential at the origin for $\varepsilon =0$, for the full
system in $(x_{2},x_{3},v_{4},\nu )$, hence we can use the implicit function
theorem to solve the full system, including all orders.

Now, we obtain
\begin{eqnarray*}
h_{2}-h_{1} &=&2x_{2}(3+2\kappa \zeta _{3}(-1)^{k}), \\
h_{3}-h_{1} &=&2x_{3}(3+2\kappa \zeta _{3}(-1)^{k})+12(\zeta
_{3}^{2}-1)x_{3}+2\kappa (-1)^{k}[\zeta _{3}-(\zeta _{3})^{-1}](x_{2}+x_{3}),
\end{eqnarray*}
which gives $x_{2}$ and $x_{3}$ provided that
\begin{equation}
(3+2\kappa \zeta _{3}(-1)^{k})\neq 0,  \label{eq:cond_inv_1}
\end{equation}
and
\begin{equation}
-6+6\kappa \zeta _{3}(-1)^{k}+12\zeta _{3}^{2}-2\kappa (\zeta
_{3})^{-1}(-1)^{k}\neq 0.  \label{eq:cond_inv_2}
\end{equation}
It appears that condition~\cref{eq:cond_inv_2} is the same as~\cref{eq:cond_inv_1} in the cases when $\zeta _{3}=\pm 1$. In the third case, when 
$\zeta _{3}=\frac{2}{3}\kappa (-1)^{k+1}$, both conditions~\cref{eq:cond_inv_1} and~\cref{eq:cond_inv_2} give
\begin{equation}
\kappa ^{2}\neq \frac{9}{4}.  \label{eq:cond_inv_IIb}
\end{equation}
Once these conditions are realized, it is clear that we can invert the
matrix $M$ (solving with respect to $(\nu,v_{4})$ is straighforward, once 
$x_{2},x_{3}$ is computed). The solution is obtained under the form of a
power series in $\varepsilon$, with coefficients depending on $\kappa$.
The series is formal in the quasiperiodic case, while it is convergent for 
$\varepsilon$ small enough in the periodic case. In all cases, the
bifurcation is supercritical $(\mu >0)$. Finally, the solutions~\cref{eq:sol_IIa} and~\cref{eq:sol_IIc} are the principal parts of superposed rolls and
hexagons. Notice that we can shift the hexagons in the plane using $\theta_{1}$ and $\theta_{2}$, and independently shift the rolls using the phase~$\theta_{4}$. Notice that a similar result holds by replacing $z_{4}$ by $z_{5}$ or $z_{6}$.

For understanding in the plane $(\mu,\chi)$ where the solutions bifurcate,
we first look at $\mu >0$ and solve at leading order the second degree equation
for $\varepsilon$. For the solution~\cref{eq:sol_IIa} this gives 
\begin{equation*}
21\varepsilon ^{2}+2\chi \varepsilon (-1)^{k+1}-\mu =0
\end{equation*}
i.e., (since $\varepsilon>0$)
\begin{equation*}
\varepsilon =\frac{(-1)^{k}\chi +\sqrt{\chi ^{2}+21\mu }}{21}.
\end{equation*}
Hence the conditions~\cref{eq:cond_IIa} and~\cref{eq:cond_inv_1} lead to
\begin{eqnarray*}
13(-1)^{k}\chi  &<&\sqrt{\chi ^{2}+21\mu }, \\
15\chi (-1)^{k+1} &\neq &\sqrt{\chi ^{2}+21\mu }.
\end{eqnarray*}
This gives the conditions (see \cref{fig:bifdiag} left side)
\begin{eqnarray*}
\mu  &>&8\chi ^{2},\text{ for }(-1)^{k}\chi >0,\text{ Parabola }(P_{1}) \\
\mu  &\neq &\frac{32}{3}\chi ^{2}\text{ for }(-1)^{k}\chi <0,
\text{ Parabola }(P_{2})
\end{eqnarray*}
For the solution~\cref{eq:sol_IIc} we have, from the expression of $\mu $ and
from~\cref{eq:cond_inv_IIb}, the conditions (see \cref{fig:bifdiag} right side)
\begin{equation*}
\mu >4\chi ^{2},\text{ \ }\mu \neq \frac{32}{3}\chi ^{2},\text{ \ Parabolas }
(P_{3})\text{ and }(P_{2}).
\end{equation*}

Finally, we state the following

\begin{figure}[tbp]
\begin{center}
\includegraphics[width=0.6\hsize]{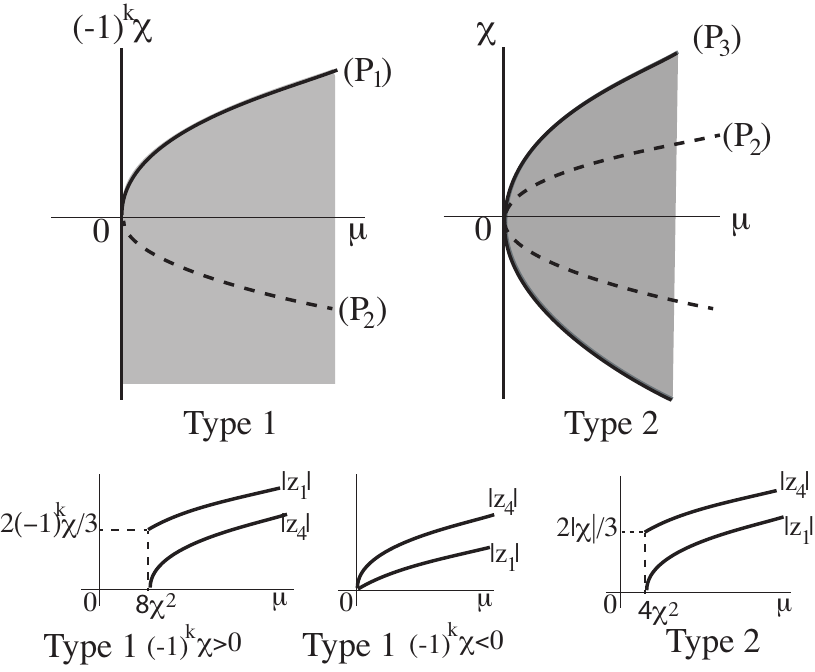}\\[0pt]
\end{center}
\caption{Domain of existence of bifurcating superposition of hexagons and
rolls \red{(balanced hexa-rolls types~1 and~2)} for small $|\chi|$. Solutions
of \red{type~1 (three hexagon amplitudes equal at leading order)} are on the left side, solutions of \red{type~2 (two of the three hexagon amplitudes equal at leading order)} are on the right side. The
parabola $(P_2)$ (dashed line) is a forbidden place.}
\label{fig:bifdiag}
\end{figure}

\begin{theorem}[Hexa-rolls: superposed hexagons and rolls in balance]
\label{thm:superp_hexa-rolls} Assume that $\alpha \in \EO$.   Then, for  $\chi=\varepsilon \kappa$, $\varepsilon >0$ close enough to
0, 
we can build a series in powers of 
$\varepsilon$, solution of~\cref{eq:SH},
of the form
\begin{eqnarray*}
u(\varepsilon ,\kappa ,\Theta ,k,j) &=&\varepsilon u_{1}(\Theta
)+\sum_{n\geq 1}\varepsilon ^{2n+1}u_{2n+1}(\kappa ,\Theta ,k,j),\text{ }
u_{2n+1}\bot e^{i\bk_{1}\cdot \bx},\text{ }n\geq 1, \\
u_{1}(\Theta ,\kappa ,k,j) &=&\sum_{m=1,2}e^{i(\bk_{m}\cdot \bx+\theta _{m})}+\zeta_3e^{i(\bk_{3}\cdot \bx+\theta_3)}
+\sqrt{u_{4}^{(0)}}e^{i(\bk_{j}\cdot \bx
+\theta _{j})}+c.c.,\\
\Theta  &=&(\theta _{1},\theta _{2},\theta _{3},\theta _{j}),\text{ }j=4
  \text{ or }5\text{ or }6, \text{  }  \theta_1+\theta_2+\theta_3=k\pi,\text{  }k=0\text{ or }1\\
\mu (\varepsilon ,\kappa ,k,j) &=&\varepsilon ^{2}\mu ^{(0)}(\kappa
,k)+\sum_{n\geq 2}\varepsilon ^{2n}\mu _{2n}(\kappa ,k,j), 
\end{eqnarray*}
\red{Balanced hexa-rolls type 1 (three hexagon amplitudes equal at leading order):}
\begin{equation*}
\zeta_3=1,\text{  }\mu ^{(0)}(\kappa,k) =(-1)^{k+1}2\kappa +21,\text{ }u_{4}^{(0)}=(-1)^{k+1}\frac{2}{3}
\kappa +1,\text{  }(-1)^{k}\kappa<3/2,\text{  }(-1)^{k}\kappa\neq -3/2.
\end{equation*}
\red{Balanced hexa-rolls type 2 (two of the three hexagon amplitudes equal at leading order):}
\begin{equation*}
    \zeta_3=\frac{2}{3} \kappa (-1)^{k+1},\text{  }\mu ^{(0)}(\kappa) = 15+4\kappa^2,\text{  }u_{4}^{(0)}=1+\frac{4}{9}\kappa^2,\text{  } \kappa \neq \pm 3/2.
\end{equation*}
The freedom left for $\Theta$ corresponds
to an arbitrary choice for translations $\mathbf{T}_{\bdelta}$, as well for
hexagons as for rolls (for $\theta_{j}$).  In the quasiperiodic case ($\alpha\in\Eii$), these solutions give quasipatterns using the methods of~\cite{Iooss2019}. See \cref{fig:bifdiag} for understanding the domain of
bifurcating solutions in the plane $(\mu,\chi)$.
\Cref{fig:example_hexaroll_patterns} shows quasiperiodic examples of~$u_1$.
\end{theorem}

\begin{remark}
\red{As for hexa-rolls with rolls dominating}, these solutions are new, even in
the periodic case. Moreover,
notice that \red{in this case also} we have the \red{surprising} freedom on
shifts for the roll part, even in the periodic case. This follows from the
reality of the 4-dimensional system.
 \end{remark}

\section{Conclusion}\label{sec:conclusion}
We have shown the existence of new quasipattern solutions of the Swift--Hohenberg equation with quadratic as well as cubic nonlinearity: superposed hexagons with unequal amplitudes (valid only for small $\mu,\chi$). The existence of superposed hexagons with equal amplitudes ($\varepsilon = \pm \delta$) had already been established in~\cite{Fauve2019,Iooss2019}. We have also found (provided the cubic coefficients satisfy an inequality) a new class of solutions, superposed hexagons and rolls: the roll amplitude dominates if the quadratic coefficient~$\chi$ is not small, but for small~$\chi=\mathcal{O}(\sqrt{|\mu|})$, the rolls and hexagons can have similar amplitudes. For small~$\chi$, we have also found superposed symmetry-broken hexagons and rolls. Our approach relies on the small-divisor techniques from~\cite{Iooss2019} for solutions of the amplitude equations to be translated into quasipattern solutions of the PDE~\cref{eq:SHeq}. The end result is that for a full measure set of angles ($\alpha\in\Eii$), two hexagonal patterns with essentially arbitrary relative orientation and position can be superposed to produce quasipattern solutions of the Swift--Hohenberg equation. Similarly, superposed hexagons and rolls, again with essentially arbitrary relative orientation and position, also give quasipattern solutions.

In the periodic case we recover the superposed hexagon solutions already known from~\cite{Dionne1997}. We have shown that the additional triangular superlattice solutions identified by~\cite{Silber1998} in the case $(a,b)=(3,2)$ also arise for general $(a,b)$. We find a new class of periodic superposed hexagon and roll solutions, provided the cubic coefficients satisfy an inequality \red{and $a>b+1$}. Surprisingly, even in the periodic case, the hexagons and rolls can be translated arbitrarily with respect to each other.

The approach we have taken differs from that familiar from equivariant bifurcation theory (which applies only in the periodic case). When the amplitude equations reduce to a single equation, the results are of course the same. The new solutions arise in cases where there is more than one equation to solve, and in some cases, these solutions have no symmetry. 
\red{Our approach indicates how a wider class of pattern solutions can be 
investigated in pattern formation problems posed on the whole plane. It is
likely that there are many other solutions still to be found: 
hexagons with superposed rhombuses dominating (see~\cite{Subramanian2021b}), 
three sets of
rolls at different angles to each other, superpositions of hexagons and
squares, or squares and rolls at different angles, \dots. In all of these cases, careful 
consideration will have to be given to the Diophantine condition and to the 
behavior of high-order nonlinear modes.}

We have not discussed stability of these quasipatterns: that is an important and difficult problem. However, the reason for including a quadratic term in the Swift--Hohenberg equation~\cref{eq:SHeq} is that three-wave interactions generated by quadratic terms, particularly in problems in which patterns on two length scales are simultaneously unstable, are known to play a key role in stabilizing quasipatterns in a variety of contexts~\cite{Edwards1994,Lifshitz1997, Rucklidge2009,Mermin1985,Newell1993,Zhang1997a, Rucklidge2012, Skeldon2015,Subramanian2016, Ratliff2019, Castelino2020,Archer2013,Argentina2012}. Despite this, we do not expect any of the new solutions to be stable in the Swift--Hohenberg equation, but they (or related solutions) may be stable in other situations.

The recently discovered ``bronze-mean hexagonal
quasicrystals'' described in~\cite{Dotera2017,Nakakura2019,Archer2021} fall into the class of superposed hexagons. These
quasicrystals are not solutions of a PDE, but rather are constructed from assemblies of three tiles: small equilateral
triangles, large equilateral triangles, and rectangles. The Fourier transform of a
six-fold aperiodic tiling made from these tiles has prominent peaks arranged as in
\cref{fig:wavevectors}(c), with $\alpha=25.66^\circ$, and the ideas presented here may be relevant to existence of this type of quasipattern in a pattern-forming~\hbox{PDE}.

Finally, we mention a potential application of this body of work to bi-layer graphene, where two layers of hexagonally connected carbon atoms are superposed with a small orientation difference~\cite{Zeller2014}: for $\alpha$ about~$1^\circ$, these bi-layer structures can be superconducting~\cite{Yankowitz2019}. Our work may be relevant for finding quasiperiodic structures in models of this system.

\section*{Acknowledgments}
We would like to acknowledge conversations with Tomonari Dotera, Ian Melbourne,
Mary Silber and Priya Subramanian, and the anonymous referees for their
constructive comments. We are grateful to Jay Fineberg and Arshad Kudrolli for
permission to reproduce the images in \cref{fig:fw_examples}. AMR is grateful
for support from the Leverhulme Trust, UK (RF-2018-449/9), and from the EPSRC,
UK (EP/P015611/1).

\appendix

\section{Definitions of all the sets of angles} \label{app:definitions_of_all_sets}
\red{Here we first recall definitions given in main text, and supplement these
with descriptions of $\Ei$ and~$\Eii$.}

\red{The set $\Ep$ (periodic case) is given in \cref{def:Ep}, and has  $\cos\alpha$
and $\sqrt{3}\sin\alpha$ both rational, with $\alpha\in(0,\frac{\pi}{3})$. The
complement of $\Ep$, restricted to $(0,\frac{\pi}{6}]$, is $\Eqp$
(quasiperiodic case). The set $\EO$, given in \cref{def:E0}, is the set of
angles~$\alpha$ such that the only solutions of $|\bk(\bmm)|=1$ are $\pm
\bk_{j}$, $j=1,\dots,6$.}

\red{The two sets~$\Ei$ and $\Eii$ are defined in detail in~\cite{Iooss2019} and
described below: these are angles $\alpha\in\Eqp$ where additional Diophantine
conditions are satisfied. The final set is~$\Eii$.}

\red{Lemma~7 of~\cite{Iooss2019} states that for nearly all $\alpha \in \Eqp\cap
(0,\frac{\pi}{6}]$, and for any $\varepsilon>0$, there exists $c>0$ such that
for all $\bmm\neq 0$ with $|\bk(\bmm)|\neq 1$,}
 \begin{equation*}
 (|\bk(\bmm)|^{2}-1)^{2}
 \geq
 \frac{c}{|N_{\bk}|^{12+\varepsilon}}
 \end{equation*}
\red{holds. The set $\Ei$ is the set of all $\alpha$'s such
that this inequality holds, and $\Ei$ is of full measure.}

\red{Let us now choose an integer $1\leq d\leq4$ and consider an expression of the form}
 \begin{equation} \label{eq:Eii_P_expression}
 P= a_{0} + \sum_{1\leq n\leq d}a_{n0}\cos ^{n}\alpha +
                                \sqrt{3}a_{n-1,1}\sin \alpha \cos^{n-1}\alpha,
 \end{equation}
\red{where the coefficients $\mathbf{a}=(a_{0},a_{n0},a_{n-1,1},n=1,\dots,d)$ are integers: $\mathbf{a}\in\mathbb{Z}^{(2d+1)}$.  
The following proposition is proved in~\cite{Iooss2019} (see Proposition~21):}

\begin{proposition}
\red{For nearly all $\alpha\in\EO\cap\Ei\cap (0,\frac{\pi}{6}],$ there exists $c>0$ such that for all $\mathbf{a}\in\mathbb{Z}^{(2d+1)}\backslash\{0\}$ and for $l=2d(2d+1)$,}
\begin{equation*}
|P|\geq \frac{c}{|\mathbf{a}|^{l}},
\end{equation*}
\red{where $\mathbf{a}=(a_{0},a_{n0},a_{n-1,1},n=1,\dots,d)$, $1\leq d\leq 4$, and}
\begin{equation*}
|\mathbf{a}|=|a_{0}|+\sum_{1\leq n\leq d}|a_{n0}|+|a_{n-1,1}|.
\end{equation*}
\end{proposition}

\red{The set $\Eii$ is the set of all $\alpha\in(0,\frac{\pi}{6}]$ such that this
inequality holds for any $d\leq 4$, \red{provided that $|\mathbf{a}|\neq0$}.
The set $\Eii$ is a subset of $\EO\cap\Ei$, and $\Eii$ is of full
measure~\cite{Iooss2019}.}

\section{Proof of the properties of two example angles} \label{app:two_examples}
\red{While the set~$\Eii$ is of full measure~\cite{Iooss2019}, in practice it
can be difficult to determine whether any particular angle is or is not in the
set. Here we take two examples and prove that $\alpha\approx25.66^\circ$
($\cos\alpha=\frac{1}{4}\sqrt{13}$) is in~$\Eii$, while
$\alpha\approx26.44^\circ$ ($\cos\alpha=\frac{1}{12}(5+\sqrt{33})$) is not.}

\subsection{First example}
Let us consider $\alpha\in\Eqp$ such that
\begin{equation*}
\cos \alpha =\frac{\sqrt{13}}{4},\text{ \ }\sqrt{3}\sin \alpha =\frac{3}{4},
\end{equation*}
with $\alpha\approx25.66^\circ$.
In order to show that $\alpha\in\Eii$, 
we must first prove that $\alpha\in\EO$, which means that the
points of the lattice $\Gamma$ on the unit circle are only the twelve basic
points $\pm{\bk_{j}}$, $j=1,\dots,6$.
For
 \[
 \bk=n_{1}\bk_{1}+n_{2}\bk_{2}+n_{4}\bk_{4}+n_{5}
 \bk_{5},\text{ \ }n_{j}\in 
 \mathbb{Z},
 \]
the condition $|\bk|^{2}=1$ becomes
\begin{eqnarray*}
1 &=&n_{1}^{2}+n_{2}^{2}+n_{4}^{2}+n_{5}^{2}-n_{1}n_{2}-n_{4}n_{5}+{} \\
&&{}+\cos \alpha (2n_{1}n_{4}+2n_{2}n_{5}-n_{1}n_{5}-n_{2}n_{4})+{} \\
&&{}+\sqrt{3}\sin \alpha (n_{2}n_{4}-n_{1}n_{5}),
\end{eqnarray*}
which, separating the rational and irrational parts, and with the given value of~$\alpha$, leads to
\begin{align}
2n_{1}n_{4}+2n_{2}n_{5}-n_{1}n_{5}-n_{2}n_{4} &=0,  \label{eq:syst_for_nj} \\
3(n_{2}n_{4}-n_{1}n_{5})+4(n_{1}^{2}+n_{2}^{2}+n_{4}^{2}+n_{5}^{2}-n_{1}n_{2}-n_{4}n_{5}) &=4.
\nonumber
\end{align}
Solving with respect to $n_{5}$ leads to
\[
n_{5}=n_{4}\frac{n_{2}-2n_{1}}{2n_{2}-n_{1}},
\]
provided that $n_{1}\neq 2n_{2}$,
\begin{eqnarray*}
0 &=&4n_{4}^{2}\left( 1+(\frac{n_{2}-2n_{1}}{2n_{2}-n_{1}})^{2}
     - \frac{n_{2}-2n_{1}}{2n_{2}-n_{1}}\right) + \\
&&{}+3n_{4}\left( n_{2}-n_{1}\frac{n_{2}-2n_{1}}{2n_{2}-n_{1}}\right)
     + 4(n_{1}^{2}+n_{2}^{2}-n_{1}n_{2}-1),
\end{eqnarray*}
i.e.,
\begin{equation*}
6n_{4}^{2}(n_{1}^{2}+n_{2}^{2}-n_{1}n_{2})+3n_{4}(n_{1}^{2}+n_{2}^{2}-n_{1}n_{2})(2n_{2}-n_{1})+2(n_{1}^{2}+n_{2}^{2}-n_{1}n_{2}-1)(2n_{2}-n_{1})^{2}=0.
\end{equation*}
The discriminant of this quadratic equation for~$n_4$ reads
\begin{eqnarray*}
\Delta 
&=&9(n_{1}^{2}+n_{2}^{2}-n_{1}n_{2})^{2}(2n_{2}-n_{1})^{2}-48(n_{1}^{2}+n_{2}^{2}-n_{1}n_{2}-1)(2n_{2}-n_{1})^{2}(n_{1}^{2}+n_{2}^{2}-n_{1}n_{2})
 \\
&=&3(n_{1}^{2}+n_{2}^{2}-n_{1}n_{2})(2n_{2}-n_{1})^{2}\left[
16-13(n_{1}^{2}+n_{2}^{2}-n_{1}n_{2})\right]. 
\end{eqnarray*}
We observe that $\Delta $ should be $\geq 0$, and since 
$(n_{1}^{2}+n_{2}^{2}-n_{1}n_{2})(2n_{2}-n_{1})^{2}\geq 0$, this implies
\[
16\geq 13(n_{1}^{2}+n_{2}^{2}-n_{1}n_{2}).
\]
This in turn implies that 
\[
n_{1}^{2}+n_{2}^{2}-n_{1}n_{2}=1\text{ or }0.
\]
The only solutions are
\[
(n_{1},n_{2})=(0,0),(0,\pm 1),(\pm 1,0),(\pm 1,\pm 1),
\]
leading to
\begin{eqnarray*}
\Delta  &=&9\text{ for }(n_{1},n_{2})=(\pm 1,0),(\pm 1,\pm 1), \\
\Delta  &=&36\text{ for }(n_{1},n_{2})=(0,\pm 1).
\end{eqnarray*}
The case $(n_{1},n_{2})=(0,0)$ in~\cref{eq:syst_for_nj}, leads to 
$n_{4}^{2}+n_{5}^{2}-n_{4}n_{5}=1$, which correspond to $\pm\bk_{4}$, $\pm \bk_{5}$ and $\pm \bk_{6}$.
The case $(n_{1},n_{2})=(\pm 1,0),(\pm 1,\pm 1)$ leads to $n_{4}=0$ or $\mp\frac{1}{2}$ (which is not acceptable). Finally the case is $(n_{1},n_{2})=(0,\pm 1)$
gives
\[
n_{4}=0\text{ or }\mp 1, 
\]
and $n_{5}=0$ or $\pm\frac{1}{2}$, and the only good possibility is $n_{4}=n_{5}=0$
and this corresponds to $\pm \bk_{1},\pm \bk_{2},\pm \bk_{3}$.
It remains to study the case $n_{1}=2n_{2}$, $n_{4}=0$. Replacing this in~\cref{eq:syst_for_nj}, we obtain
\[
6n_{2}^{2}-3n_{2}n_{5}+2n_{5}^{2}-2=0 
\]
and it is easy to conclude that there are no other solutions of~\cref{eq:syst_for_nj}.
The conclusion is that $\alpha\in\EO$.

Let us now prove that $\alpha$ satisfies the two Diophantine conditions
required in~\cite{Iooss2019} and described in \cref{app:definitions_of_all_sets}.
We observe that
\begin{eqnarray*}
4(|\bk|^{2}-1) &=&q_{0}\sqrt{13}+q_{1}, \\
q_{0} &=&2n_{1}n_{4}+2n_{2}n_{5}-n_{1}n_{5}-n_{2}n_{4}, \\
q_{1}
&=&3(n_{2}n_{4}-n_{1}n_{5})+4(n_{1}^{2}+n_{2}^{2}+n_{4}^{2}+n_{5}^{2}-n_{1}n_{2}-n_{4}n_{5})-4.
\end{eqnarray*}
Since $\sqrt{13}$ is a quadratic irrational (the solution of a quadratic equation with integer coefficients), it is known~\cite{Hardy1960} that
there exists $C>0$ such that
\[
|q_{0}\sqrt{13}+q_{1}|\geq \frac{C}{|q_{0}|+|q_{1}|},\text{ }
(q_{0},q_{1})\in \mathbb{Z}^{2}\backslash \{0\}.
\]
Since we have
\begin{eqnarray*}
|q_{0}| &\leq &\frac{3}{2}(n_{1}^{2}+n_{2}^{2}+n_{4}^{2}+n_{5}^{2}), \\
|q_{1}| &\leq &\frac{15}{2}(n_{1}^{2}+n_{2}^{2}+n_{4}^{2}+n_{5}^{2})+4 \\
|q_{0}|+|q_{1}| &\leq &11(n_{1}^{2}+n_{2}^{2}+n_{4}^{2}+n_{5}^{2}),
\end{eqnarray*}
hence
\[
  (|\bk|^{2}-1)^{2}\geq \frac{C^{\prime }}
                             {(n_{1}^{2}+n_{2}^{2}+n_{4}^{2}+n_{5}^{2})^{2}},
\]
which means that $\alpha \in \Ei$ as defined in \cite{Iooss2019} and described in \cref{app:definitions_of_all_sets}.

Now for $\Eii$,  let us follow the lines of \cref{app:definitions_of_all_sets}. \red{For this choice of~$\alpha$, and for any integer $d\leq4$, the expression~\cref{eq:Eii_P_expression} takes the form}
\[
P=\frac{b_{0}+b_{1}\sqrt{13}}{b_2},\text{  }b_0,b_1,b_2 \in \mathbb{Z},
\]
\red{where the integer denominator depends on~$\alpha$ and~$d$ but not on the integers~$\mathbf{a}$ in~\cref{eq:Eii_P_expression}.}
Then, as soon as $|b_0|+|b_1|\neq 0$ we again have a Diophantine estimate
\[
P>\frac{C^{\prime }}{|b_{0}|+|b_{1}|},
\]
where $b_2$ is absorbed into~$C'$. This is
the required property for $\alpha\in\Eii$ in
\cite{Iooss2019} (see also \cref{app:definitions_of_all_sets}), and so the proof that $\alpha\in\Eii$ is complete. More generally if $\cos\alpha$ is rational and $\sqrt{3}\sin\alpha$ is a quadratic irrational, or vice versa, $\Ei$~should be satisfied, as should the Diophantine requirement of~$\Eii$.

\subsection{Second example}
Let us consider $\alpha \in \mathcal{E}_{qp}$ such that
 \begin{equation*}
 \cos \alpha =\frac{5+\sqrt{33}}{12},\text{ \ }\sqrt{3}\sin \alpha =\frac{15-\sqrt{33}}{12}, 
 \end{equation*}
with $\alpha\approx26.44^\circ$.
We wish to prove that $\alpha\notin\Eii$. We have
 \[
 \bk=n_{1}\bk_{1}+n_{2}\bk_{2}+n_{4}\bk_{4}+n_{5}\bk_{5},\text{ \ }
     n_{j}\in \mathbb{Z},
 \]
and, again separating rational and irrational parts, the condition $|\bk|^{2}=1$ 
leads to
 \begin{equation}
 0=3(n_{1}^{2}+n_{2}^{2}+n_{4}^{2}+n_{5}^{2}-n_{1}n_{2}-n_{4}n_{5}-1)+5(n_{2}n_{4}-n_{1}n_{5})
 \label{eq:cond1}
 \end{equation}
and
 \begin{equation}
 n_{1}n_{4}+n_{2}n_{5}-n_{2}n_{4}=0.  \label{eq:cond2}
 \end{equation}
Then we observe that
 \[
 (n_{1},n_{2},n_{4},n_{5})=(2,1,-1,1)
 \]
is solution of~\cref{eq:cond1}, \cref{eq:cond2}. This means that the following
wave vectors lie on the unit circle
\begin{eqnarray*}
&&\pm (\bk_{1}-\bk_{3}-\bk_{4}+\bk_{5}) \\
&&\pm (\bk_{2}-\bk_{1}-\bk_{5}+\bk_{6}) \\
&&\pm (\bk_{3}-\bk_{2}-\bk_{6}+\bk_{4})
\end{eqnarray*}
and it is clear that $\pm \bk_{j}$, $j=1,\dots,6$ are not the only
elements of $\Gamma$ on the unit circle, so $\alpha\notin\EO$
and~$\alpha\notin\Eii$.

\section{Proof of \cref{lem:Ep_ab}} \label{app:proof_of_Lemma_2p2}
Let us show the following

\begin{lemma}
Let $\alpha \in \Ep\cap (0,\frac{\pi}{6})$, with $\cos\alpha$ and $\sqrt{3}\sin\alpha$ both rational, and define positive
integers $p,q,p^{\prime }$ such that
 \begin{equation}
 \cos \alpha =\frac{p}{q},
 \qquad
 \sqrt{3}\sin \alpha =\frac{p^{\prime }}{q},
 \qquad
 3p^{2}+p^{\prime 2}=3q^{2},  \label{eq:relat_alpha}
 \end{equation}
where $(p,q,p^{\prime })$ have no common divisor. We define~$d$ to be the greatest
common divisor of $2(p+q)$ and $(p+q+p^{\prime })$. Then, $(a,b)$
defined by
 \begin{equation}\label{eq:ab}
 a=\frac{2(p+q)}{d},
 \qquad
 b=\frac{p+q+p^{\prime }}{d}
 \end{equation}
are relatively prime integers that 
satisfy~\cref{eq:cosalpha} and $a>b>\tfrac{1}{2}a>0$.
\end{lemma}

\begin{proof}
Let us assume that~\cref{eq:relat_alpha} holds, and we seek integers $(a,b)$ such that~\cref{eq:cosalpha} holds. 
If $(a,b)$ are integers given by~\cref{eq:ab}, then 
(using $3p^{2}+p^{\prime 2}=3q^{2}$) this leads to
 \begin{eqnarray*}
 a^{2}+2ab-2b^{2} &=& p \times \frac{12(p+q)}{d^{2}}, \\
 3a(2b-a) &=& p' \times \frac{12(p+q)}{d^{2}}, \\
 2(a^{2}-ab+b^{2}) &=& q \times \frac{12(p+q)}{d^{2}}.
 \end{eqnarray*}
Dividing the first and second lines by the third leads to~\cref{eq:cosalpha}.
Now since $\alpha \in (0,\frac{\pi}{3})$ we have
\begin{equation*}
p^{\prime }<\tfrac{{3}}{2}q<3p<3q,
\end{equation*}
which leads to
\begin{equation*}
a>b>\tfrac{1}{2}a>0.
\end{equation*}
\end{proof}

It remains to check that we can assume $a+b$ not multiple of~$3$. Suppose that this is not the case, then we define
 \[
 a^{\prime }=\tfrac{1}{3}(a+b),\qquad 
 b^{\prime }=\tfrac{1}{3}(2a-b),
 \]
then it is easy to check that
\[
\cos\left(\frac{\pi}{3}-\alpha\right)
 = \frac{a^{\prime 2}+2a^{\prime }b^{\prime }-2b^{\prime 2}}
        {2(a^{\prime 2}-a^{\prime }b^{\prime }+b^{\prime 2})},
\qquad
\sqrt{3}\sin\left(\frac{\pi}{3}-\alpha\right)
   = \frac{3a^{\prime }(2b^{\prime }-a^{\prime })}
          {2(a^{\prime 2}-a^{\prime }b^{\prime }+b^{\prime 2})},
\]
hence we have for $\frac{\pi}{3}-\alpha$ the same formulas as for $\alpha $ in replacing $(a,b)$ by $(a^{\prime
},b^{\prime })$. This means that in such a case we should choose to consider
the angle $\alpha ^{\prime }=\frac{\pi}{3}-\alpha $ instead of $\alpha $, which
does not change the fact that $\alpha ^{\prime }\in (0,\frac{\pi}{3})$. If it
appears that $a^{\prime }+b^{\prime }$ is also multiple of~3, then we need to
iterate the operation. In fact this operation means that we can choose basis vectors $(s_{1}-s_{2},s_{1}+2s_{2})$ instead of 
$(s_{1},s_{2})$, for the periodic lattice: these are $\sqrt{3}$ larger.
The property~(iii) of \cref{lem:Ep_ab}  is proved. 

Now, we prove the density of $\mathcal{E}_{p}$. The continuous monotonous
function of~$x$
\begin{equation*}
\frac{x^{2}+2x-2}{2(x^{2}-x+1)}
\end{equation*}
makes a homeomorphism between$(1,2)$ and $(\frac{1}{2},1)$, it is clear that the set
of values taken by $\cos \alpha $ for $x=a/b$ rational is dense on $(\frac{1}{2},1)$.
It follows that the set of angles $\alpha \in \lbrack 0,\frac{\pi}{3})$ satisfying~\cref{eq:cosalpha} for $a/b$ rational is dense. Hence the property (i) of \cref{lem:Ep_ab} (the density of~$\mathcal{E}_{p}$) is proved.

\begin{remark}
We notice that $d$ divides $2(p+q)$, and $2p^{\prime }$ and that $d^{2}$
divides $12(p+q)$ because $p,q$ and $p^{\prime }$ have no common divisor and
$12(p+q)(q-p)=4p^{\prime 2}$
\end{remark}

\section{Proof of~\cref{eq:struct_P1_periodic_case}}\label{app:proof_of_periodic_case}
In this case the wave vectors $\bk_{j}$ are defined in \cref{eq:periodic_kj}, and \cref{eq:ident_kj} leads to
\begin{eqnarray*}
(n_{1}-n_{3})a+(n_{2}-n_{3})(b-a)+(n_{4}-n_{6})a-(n_{5}-n_{6})b &=&0, \\
(n_{1}-n_{3})b-(n_{2}-n_{3})a+(n_{4}-n_{6})(a-b)-(n_{5}-n_{6})a &=&0.
\end{eqnarray*}
Since $a$ and $b$ have no common factor, it follows that there exist 
$(j,l)\in\mathbb{Z}^{2}$ such that
 \begin{eqnarray*}
 n_{1}-n_{2}+n_{4}-n_{6} &=&jb, \\
 n_{2}-n_{3}-n_{5}+n_{6} &=&-ja, \\
 n_{2}-n_{3}-n_{4}+n_{5} &=&lb, \\
 n_{1}-n_{3}-n_{4}+n_{6} &=&la.
 \end{eqnarray*}
This system leads to
 \begin{eqnarray*}
 n_{1}-n_{3} &=&jb+\frac{l-j}{3}(a+b), \\
 n_{1}-n_{2} &=&la-\frac{l-j}{3}(a+b), \\
 n_{4}-n_{5} &=&-ja-\frac{l-j}{3}(a+b), \\
 n_{4}-n_{6} &=&jb-la+\frac{l-j}{3}(a+b).
 \end{eqnarray*}
Since $a+b$ is not a multiple of $3$, this implies that there is 
a $k\in\mathbb{Z}$ such that
 \begin{equation*}
 l-j=3k,
 \end{equation*}
and
 \begin{eqnarray*}
 n_{1}-n_{3} &=&(j+k)b+ka, \\
 n_{1}-n_{2} &=&(j+2k)a-kb, \\
 n_{4}-n_{5} &=&-(j+k)a-kb, \\
 n_{4}-n_{6} &=&(j+k)b-(j+2k)a.
 \end{eqnarray*}
We notice that the monomials invariant under $\mathbf{T}_{\bdelta}$, of
minimal degree found in \cite{Dionne1997} correspond to the following
choices: $(j,k)=(1,0),(-2,1),(1,-1)$, their complex conjugate being given by
the opposite values of $(j,k)$. The basic invariant monomials where $a$ and 
$b$ occur are found by looking for the 27 monomials independent of two of
the $z_{j}$:
 \begin{equation*}
 q_{I,1}=z_{2}^{b}{\bar{z}}_{3}^{a-b}{\bar{z}}_{5}^{a-b}z_{6}^{b},\text{
 \ }q_{I,2}={\bar{z}}_{2}^{a}{\bar{z}}_{3}^{b}z_{5}^{a}z_{6}^{a-b},
 \text{ \ }q_{I,3}=z_{2}^{a-b}z_{3}^{a}{\bar{z}}_{5}^{b}{\bar{z}}_{6}
 ^{a},
 \end{equation*}
 \begin{equation*}
 q_{II,1}=z_{2}^{b}{\bar{z}}_{3}^{a-b}z_{4}^{a-b}z_{6}^{a},\text{ \ }
 q_{II,2}=z_{2}^{a-b}z_{3}^{a}z_{4}^{b}{\bar{z}}_{6}^{a-b},\text{ \ }
 q_{II,3}=z_{2}^{a}z_{3}^{b}z_{4}^{a}z_{6}^{b},
 \end{equation*}
 \begin{equation*}
 q_{III,1}=z_{2}^{a}z_{3}^{b}z_{4}^{a-b}{\bar{z}}_{5}^{b},\text{ \ }
 q_{III,2}=z_{2}^{b}{\bar{z}}_{3}^{a-b}{\bar{z}}_{4}^{b}{\bar{z}}_{5}
 ^{a},\text{ \ }q_{III,3}=z_{2}^{a-b}z_{3}^{a}z_{4}^{a}z_{5}^{a-b},
 \end{equation*}
 \begin{equation*}
 q_{IV,1}=z_{1}^{b}z_{3}^{a}z_{5}^{a-b}{\bar{z}}_{6}^{b},\text{ \ }
 q_{IV,2}=z_{1}^{a-b}{\bar{z}}_{3}^{b}z_{5}^{b}z_{6}^{a},\text{ \ }
 q_{IV,3}=z_{1}^{a}z_{3}^{a-b}z_{5}^{a}z_{6}^{a-b},
 \end{equation*}
 \begin{equation*}
 q_{V,1}=z_{1}^{a-b}{\bar{z}}_{3}^{b}{\bar{z}}_{4}^{b}z_{6}^{a-b},\text{
 \ }q_{V,2}=z_{1}^{a}z_{3}^{a-b}{\bar{z}}_{4}^{a}{\bar{z}}_{6}^{b},
 \text{ \ }q_{V,3}=z_{1}^{b}z_{3}^{a}{\bar{z}}_{4}^{a-b}{\bar{z}_{6}}^{a},
 \end{equation*}
 \begin{equation*}
 q_{VI,1}=z_{1}^{a}z_{3}^{a-b}{\bar{z}}_{4}^{a-b}z_{5}^{b},\text{ \ }
 q_{VI,2}=z_{1}^{a-b}{\bar{z}}_{3}^{b}{\bar{z}}_{4}^{a}{\bar{z}}_{5}
 ^{a-b},\text{ \ }q_{VI,3}=z_{1}^{b}z_{3}^{a}z_{4}^{b}z_{5}^{a},
 \end{equation*}
 \begin{equation*}
 q_{VII,1}=z_{1}^{b}{\bar{z}}_{2}^{a-b}z_{5}^{a}z_{6}^{a-b},\text{ \ }
 q_{VII,2}=z_{1}^{a-b}z_{2}^{a}{\bar{z}}_{5}^{a-b}z_{6}^{b},\text{ \ }
 q_{VII,3}=z_{1}^{a}z_{2}^{b}z_{5}^{b}z_{6}^{a},
 \end{equation*}
 \begin{equation*}
 q_{VIII,1}=z_{1}^{b}{\bar{z}}_{2}^{a-b}{\bar{z}}_{4}^{a}{\bar{z}_{6}}^{b},
 \text{ \ }q_{VIII,2}=z_{1}^{a}z_{2}^{b}{\bar{z}}_{4}^{b}z_{6}^{a-b},
 \text{ \ }q_{VIII,3}=z_{1}^{a-b}z_{2}^{a}z_{4}^{a-b}z_{6}^{a},
 \end{equation*}
 \begin{equation*}
 q_{IX,1}=z_{1}^{b}{\bar{z}}_{2}^{a-b}{\bar{z}}_{4}^{a-b}z_{5}^{b},
 \text{ \ }q_{IX,2}=z_{1}^{a}z_{2}^{b}{\bar{z}}_{4}^{a}{\bar{z}}_{5}
 ^{a-b},\text{ \ }q_{IX,3}=z_{1}^{a-b}z_{2}^{a}{\bar{z}}_{4}^{b}{\bar{z}_{5}}^{a}.
 \end{equation*}
Notice that $q_{I,1}$, $q_{V,1}$, $q_{IX,1}$ are mentioned 
in~\cite{Dionne1997}. We may also notice that these invariants are not independent
since there are relationships between them and the~$u_{j}$.
We may group
these invariant monomials into nine sets of monomials 
\begin{eqnarray*}
G_{1} &=&\{q_{I,1},\overline{q_{V,1}},q_{IX,1}\}\text{ with degree }2a,
\\
G_{2} &=&\{q_{II,1},\overline{q_{VI,2}},q_{VII,1}\}\text{
with degree }3a-b, \\
G_{2}^{\prime } &=&\{q_{II,2},q_{VI,1},q_{VII,2}\}\text{with
degree }3a-b, \\
G_{3} &=&\{q_{III,1},q_{IV,1},q_{VIII,2}\}\text{ \ with
degree }2a+b, \\
G_{3}^{\prime } &=&\{q_{III,2},\overline{q_{IV,2}},q_{VIII,1}\}\text{ with degree }2a+b, \\
G_{4} &=&\{q_{III,3},q_{IV,3},\text{ }q_{VIII,3}\}\text{ with
degree }4a-2b, \\
G_{5} &=&\{\overline{q_{I,2}},q_{V,3},q_{IX,2}\}\text{ \ with
degree }3a, \\
G_{5}^{\prime } &=&\{q_{I,3},q_{V,2},q_{IX,3}\}\text{, with
degree }3a, \\
G_{6} &=&\{q_{II,3},q_{VI,3},q_{VII,3}\}\text{ with degree }
2a+2b,
\end{eqnarray*}
and their complex conjugates.

Let us control the action of various symmetries (other than~$\mathbf{T}_{\delta}$, 
which leaves them invariant), useful for obtaining the system of
6 complex bifurcation equations. We have
\begin{eqnarray*}
\mathbf{R}_{\pi/3}\{q_{I,1},\overline{q_{V,1}},q_{IX,1}\} &=&\{q_{V,1}
,\overline{q_{IX,1}},\overline{q_{I,1}}\}, \\
\mathbf{\tau }\{q_{I,1},\overline{q_{V,1}},q_{IX,1}\} &=&\{q_{I,1},q_{IX,1}
,\overline{q_{V,1}}\},   \\
\mathbf{S}\{q_{I,1},\overline{q_{V,1}},q_{IX,1}\} &=&\{q_{I,1},\overline{q_{V,1}},q_{IX,1}\}, 
\end{eqnarray*}
\begin{eqnarray*}
\mathbf{R}_{\pi/3}\{q_{II,1},\overline{q_{VI,2}},q_{VII,1}\} &=&\{q_{VI,2},
\overline{q_{VII,1}},\overline{q_{II,1}}\}, \\
\mathbf{\tau }\{q_{II,1},\overline{q_{VI,2}},q_{VII,1}\}
&=&\{q_{VII,2},q_{VI,1},q_{II,2}\},   \\
\mathbf{S}\{q_{II,1},\overline{q_{VI,2}},q_{VII,1}\} &=&(-1)^{a+b}\{q_{II,1},
\overline{q_{VI,2}},q_{VII,1}\},  
\end{eqnarray*}
\begin{eqnarray*}
\mathbf{R}_{\pi/3}\{q_{II,2},q_{VI,1},q_{VII,2}\} &=&
 \{\overline{q_{VI,1}},\overline{q_{VII,2}},\overline{q_{II,2}}\},  \\
\mathbf{\tau }\{q_{II,2},q_{VI,1},q_{VII,2}\} &=&
 \{q_{VII,1},\overline{q_{VI,2}},q_{II,1}\},   \\
\mathbf{S}\{q_{II,2},q_{VI,1},q_{VII,2}\} &=&
 (-1)^{a+b}\{q_{II,2},q_{VI,1},q_{VII,2}\},  
\end{eqnarray*}%
\begin{eqnarray*}
\mathbf{R}_{\pi/3}\{q_{III,1},q_{IV,1},q_{VIII,2}\} &=&
 \{\overline{q_{IV,1}},\overline{q_{VIII,2}},\overline{q_{III,1}}\},  \\
\mathbf{\tau }\{q_{III,1},q_{IV,1},q_{VIII,2}\} &=& 
 \{q_{IV,2},\overline{q_{III,2}},\overline{q_{VIII,1}}\},   \\
\mathbf{S}\{q_{III,1},q_{IV,1},q_{VIII,2}\} &=&
 (-1)^{b}\{q_{III,1},q_{IV,1},q_{VIII,2}\},  
\end{eqnarray*}%
\begin{eqnarray*}
\mathbf{R}_{\pi/3}\{q_{III,2},\overline{q_{IV,2}},q_{VIII,1}\} &=&
 \{q_{IV,2},\overline{q_{VIII,1}},\overline{q_{III,2}}\}, \\
\mathbf{\tau }\{q_{III,2},\overline{q_{IV,2}},q_{VIII,1}\} &=&
 \{\overline{q_{IV,1}},\overline{q_{III,1}},\overline{q_{VIII,2}}\},   \\
\mathbf{S}\{q_{III,2},\overline{q_{IV,2}},q_{VIII,1}\} &=&
 (-1)^{b}\{q_{III,2},\overline{q_{IV,2}},q_{VIII,1}\},  
\end{eqnarray*}
\begin{eqnarray*}
\mathbf{R}_{\pi/3}\{q_{III,3},q_{IV,3},q_{VIII,3}\} &=&
 \{\overline{q_{IV,3}},\overline{q_{VIII,3}},\overline{q_{III,3}}\},  \\
\mathbf{\tau }\{q_{III,3},q_{IV,3},q_{VIII,3}\} &=&
 \{q_{IV,3},q_{III,3},q_{VIII,3}\},   \\
\mathbf{S}\{q_{III,3},q_{IV,3},q_{VIII,3}\} &=&
 \{q_{III,3},q_{IV,3},q_{VIII,3}\},  
\end{eqnarray*}
\begin{eqnarray*}
\mathbf{R}_{\pi/3}\{\overline{q_{I,2}},q_{V,3},q_{IX,2}\} &=&
 \{\overline{q_{V,3}},\overline{q_{IX,2}},q_{I,2}\},  \\
\mathbf{\tau }\{\overline{q_{I,2}},q_{V,3},q_{IX,2}\} &=&
 \{\overline{q_{I,3}},\overline{q_{IX,3}},\overline{q_{V,2}}\},   \\
\mathbf{S}\{\overline{q_{I,2}},q_{V,3},q_{IX,2}\} &=&
 (-1)^{a}\{\overline{q_{I,2}},q_{V,3},q_{IX,2}\},  
\end{eqnarray*}
\begin{eqnarray*}
\mathbf{R}_{\pi/3}\{q_{I,3},q_{V,2},q_{IX,3}\} &=&
 \{\overline{q_{V,2}},\overline{q_{IX,3}},\overline{q_{I,3}}\},   \\
\mathbf{\tau }\{q_{I,3},q_{V,2},q_{IX,3}\} &=&
 \{q_{I,2},\overline{q_{IX,2}},\overline{q_{V,3}}\},   \\
\mathbf{S}\{q_{I,3},q_{V,2},q_{IX,3}\} &=&
 (-1)^{a}\{q_{I,3},q_{V,2},q_{IX,3}\},  
\end{eqnarray*}
\begin{eqnarray*}
\mathbf{R}_{\pi/3}\{q_{II,3},q_{VI,3},q_{VII,3}\} &=&
 \{\overline{q_{VI,3}},\overline{q_{VII,3}},\overline{q_{II,3}}\},  \\
\mathbf{\tau }\{q_{II,3},q_{VI,3},q_{VII,3}\} &=&
 \{q_{VII,3},q_{VI,3},q_{II,3}\},   \\
\mathbf{S}\{q_{II,3},q_{VI,3},q_{VII,3}\} &=&
 \{q_{II,3},q_{VI,3},q_{VII,3}\}.
\end{eqnarray*}
All this leads in a straightforward way to~\cref{eq:struct_P1_periodic_case}.

\section{Form of the cubic part of the bifurcation system}\label{app:form_of_the_cubic_part}
Equation~\cref{eq:SH},
projected orthogonally on the complement of
$\ker\mathbf{L}_{0}$, leads to 
\begin{equation}
\widetilde{\mathbf{L}_{0}}w=\mu w - 
   \chi \mathbf{Q}_{0}(v_{1}+w)^{2}-\mathbf{Q}_{0}(v_{1}+w)^{3},
  \label{eq:range}
\end{equation}
where we set
 \begin{equation*}
 u=v_{1}+w,\text{ }v_{1}\in \ker \mathbf{L}_{0},\text{ }
 w\in \{\ker \mathbf{L}_{0}\}^{\bot },
 \end{equation*}
and $\mathbf{Q}_{0}$ is the orthogonal projection on the complement of  $\ker\mathbf{L}_{0}$, $\widetilde{\mathbf{L}_{0}}$ being the restriction of $\mathbf{L}_{0}$
on its range, the inverse of which is the pseudo-inverse of $\mathbf{L}_{0}$
(bounded in the periodic case, unbounded in the quasiperiodic case because
of small divisors). Equation~\cref{eq:range} may be solved formally with
respect to $w$ as a power series in $v_{1}$ and~$\mu$. We have at quadratic order
\begin{equation*}
w_{2}=-\chi \widetilde{\mathbf{L}_{0}}^{-1}\mathbf{Q}_{0}v_{1}^{2},
\end{equation*}
and at cubic order in $v_{1},\mu $
 \begin{equation*}
 w_{3}=-\mu \chi \widetilde{\mathbf{L}_{0}}^{-2}\mathbf{Q}_{0}v_{1}^{2} + 
  2\chi^{2}\widetilde{\mathbf{L}_{0}}^{-1}\mathbf{Q}_{0}[v_{1}\widetilde{\mathbf{L}_{0}}^{-1}\mathbf{Q}_{0}v_{1}^{2}] - 
 \widetilde{\mathbf{L}_{0}}^{-1}\mathbf{Q}_{0}v_{1}^{3}.
 \end{equation*}
Now the bifurcation equation is
\begin{equation*}
0=\mu v_{1}-\chi \mathbf{P}_{0}(v_{1}+w)^{2}-\mathbf{P}_{0}(v_{1}+w)^{3},
\end{equation*}
where $\mathbf{P}_{0}$ is the orthogonal projection on $\ker \mathbf{L}_{0}$
and where we replace $w$ by its formal expansion in powers of $(\mu,v_{1})$.
This leads to
\begin{equation*}
\mu v_{1}=\chi \mathbf{P}_{0}v_{1}^{2}+\mathbf{P}_{0}v_{1}^{3}-2\chi ^{2}
\mathbf{P}_{0}v_{1}\widetilde{\mathbf{L}_{0}}^{-1}\mathbf{Q}_{0}v_{1}^{2}+
\mathcal{O}(v_{1}^{4}).
\end{equation*}
It follows that, up to cubic order in $(\mu,v_{1})$, the bifurcation system
reads
\begin{equation*}
\mu v_{1}=\chi \mathbf{P}_{0}v_{1}^{2}+\mathbf{P}_{0}v_{1}^{3}-2\chi ^{2}
\mathbf{P}_{0}v_{1}\widetilde{\mathbf{L}_{0}}^{-1}\mathbf{Q}_{0}v_{1}^{2}.
\end{equation*}
The scalar product with $e^{i\bk_{1}\cdot \bx}$ gives
\begin{equation}
\mu z_{1}=\chi \langle v_{1}^{2},e^{i\bk_{1}\cdot \bx}\rangle
+\langle v_{1}^{3},e^{i\bk_{1}\cdot \bx}\rangle -
 2\chi^{2}\langle v_{1}\widetilde{\mathbf{L}_{0}}^{-1}\mathbf{Q}_{0}v_{1}^{2},e^{i\bk_{1}\cdot\bx}\rangle.
\end{equation}
It is straightforward to  check that
\begin{equation*}
\langle v_{1}^{2},e^{i\bk_{1}\cdot \bx}\rangle =
  2 \overline{z_{2}}\overline{z_{3}},
\end{equation*}
\begin{eqnarray*}
 \langle v_{1}^{3},e^{i\bk_{1}\cdot \bx}\rangle  &=&
\langle 3z_{1}^{2}\overline{z_{1}}e^{i\bk_{1}\cdot\bx}+6\sum_{j=2,\dots,6}z_{1}z_{j}\overline{z_{j}}e^{i\bk_{1}\cdot\bx},
        e^{i\bk_{1}\cdot \bx}\rangle  \\
 &=& 3z_{1}u_{1}+6z_{1}(u_{2}+u_{3}+u_{4}+u_{5}+u_{6}).
\end{eqnarray*}
The next term is more complicated:
\begin{equation*}
\langle v_{1}\widetilde{\mathbf{L}_{0}}^{-1}\mathbf{Q}_{0}v_{1}^{2},
        e^{i\bk_{1}\cdot \bx}\rangle =
 \sum_{j=1,\dots,6}z_{j}
    \langle \widetilde{\mathbf{L}_{0}}^{-1}\mathbf{Q}_{0}v_{1}^{2},
            e^{i(\bk_{1}-\bk_{j})\cdot \bx}\rangle
 + \sum_{j=1,\dots,6}\overline{z_{j}}
    \langle \widetilde{\mathbf{L}_{0}}^{-1}\mathbf{Q}_{0}v_{1}^{2},
            e^{i(\bk_{1}+\bk_{j})\cdot \bx}\rangle ,
\end{equation*}
and the relevant terms in $v_{1}^{2}$ are those with an exponent  
 \begin{equation*}
 (\bk_{1}\mp \bk_{j})\cdot \bx,
 \text{ such that } \bk_{1}\mp \bk_{j}\neq \pm \bk_{l},
 \text{ } l=1,\dots,6.
\end{equation*}
the operator $\widetilde{\mathbf{L}_{0}}^{-1}$ provides a multiplication by
\begin{equation*}
(1-|\bk_{1}\mp \bk_{j}|^{2})^{-2}.
\end{equation*}
We notice that
\begin{eqnarray*}
 |\bk_{1}-\bk_{2}| &=&|\bk_{1}-\bk_{3}|,
 \text{ while }|\bk_{1}+\bk_{2}|,|\bk_{1}+\bk_{3}|\text{ do not appear,} \\
&&|\bk_{1}\pm \bk_{4}|,|\bk_{1}\pm \bk_{5}|,|\bk_{1}\pm \bk_{6}|
 \text{ all different and functions of }\alpha.
\end{eqnarray*}%
Hence 
\begin{equation*}
2\chi ^{2}\langle v_{1}\widetilde{\mathbf{L}_{0}}^{-1}\mathbf{Q}_{0}v_{1}^{2},
                  e^{i\bk_{1}\cdot \bx}\rangle
 = \chi^{2}z_{1}[c_{1}u_{1} + c_{2}(u_{2}+u_{3}) + c_{\alpha }u_{4} + 
                 c_{\alpha+}u_{5} + c_{\alpha-}u_{6}],
\end{equation*}
with
\begin{eqnarray*}
c_{1} &=&2(1+1/9),\text{ since }|2\bk_{1}|=2, \\   
c_{2} &=&2(1+1/2),\text{ since }|\bk_{1}-\bk_{2}|=\sqrt{3}, \\
c_{\alpha } &=&2[1+2(1-|\bk_{1}-\bk_{4}|^{2})^{-2} +
                   2(1-|\bk_{1}+\bk_{4}|^{2})^{-2}], \\
c_{\alpha +} &=&2[1+2(1-|\bk_{1}-\bk_{5}|^{2})^{-2} +
                    2(1-|\bk_{1}+\bk_{5}|^{2})^{-2}], \\
c_{\alpha -} &=&2[1+2(1-|\bk_{1}-\bk_{6}|^{2})^{-2} +
                    2(1-|\bk_{1}+\bk_{6}|^{2})^{-2}].
\end{eqnarray*}

\section{Looking for translations}\label{app:translations} 
Let us consider the cases with $\alpha\in\Ep$, then we can
choose the translation operator $\mathbf{T}_{\mathbf{\delta }}$ such that
\begin{align}
\mathbf{\delta}\cdot \bk_{j} &=\phantom{-}\frac{2\pi}{3}\mod 2\pi,\text{ for }j=1,2,3,
\label{eq:id_delta} \\
{}&=-\frac{2\pi}{3} \mod 2\pi,\text{ for }j=4,5,6.  \notag
\end{align}
Indeed, we set 
\begin{equation*}
\mathbf{\delta}=\frac{2\pi}{3}\lambda ^{2}m\mathbf{s}_{1},
\end{equation*}
where $\mathbf{s}_{1}$ and $\lambda $ are defined at \cref{lem:Ep_ab} and $m$ is an
integer. Then~\cref{eq:id_delta} leads to
\begin{eqnarray*}
m(2a-b) &=&2(1+3n_{1}), \\
m(2b-a) &=&2(1+3n_{2}), \\
m(a+b) &=&2(-1+3n_{4}), \\
m(a-2b) &=&2(-1+3n_{5}),
\end{eqnarray*}
where $n_{1},n_{2},n_{4},n_{5}$ are integers. It follows that
\begin{eqnarray*}
n_{2}+n_{5} &=&0, \\
am &=&2(n_{1}+n_{4}),\\
a(2n_{4}-n_{1}-1) &=& b(n_{1}+n_{4}), \\
a(n_{1}+n_{4}+3n_{2}+1) &=& 2b(n_{1}+n_{4}).
\end{eqnarray*}
The last two lines give
\begin{equation*}
n_{2}=n_{4}-n_{1}-1,
\end{equation*}
and so
\begin{eqnarray*}
n_{1}+n_{4} &=&la, \\
2n_{4}-n_{1}-1 &=&lb,
\end{eqnarray*}
where $l$ is an integer, leading to
\begin{equation*}
3n_{4}=1+l(a+b).
\end{equation*}
Since $a+b$ is not multiple of 3, we have to look at two cases: $a+b=3j+1$ \
or $a+b=3j+2$. 

For $a+b=3j+1$ we choose $l=2$, hence
\begin{equation*}
n_{4} =2j+1, \text{  }
n_{1} =2a-2j-1,  \text{  }
n_{2} =4j-2a+1,  \text{  }
n_{5} =-n_{2},  \text{  }
m =4.
\end{equation*}
For $a+b=3j+2$ we choose $l=1$, hence
\begin{equation*}
n_{4} =j+1,  \text{  }
n_{1} =a-j-1,  \text{  }
n_{2} =2j-a+1,  \text{  }
n_{5} =-n_{2},  \text{  }
m =2.
\end{equation*}
It follows that the solutions in \cref{thm:superposed_hexagons_periodic} obtained for $\theta_{1}=\theta _{2}=\theta _{3}=-\theta _{4}=-\theta _{5}=-\theta _{6}=k\frac{\pi}{3}$,
provide \emph{only two different patterns}, one corresponding to $k=0,2,4$,
the other for $k=1,3,5$.

\bibliographystyle{siamplain}
\bibliography{allrefs}

\begin{thebibliography}{10}

\bibitem{Achim2014}
{\sc C.~V. Achim, M.~Schmiedeberg, and H.~L\"owen}, {\em Growth modes of
  quasicrystals}, Phys. Rev. Lett., 112 (2014), p.~255501,
  \url{https://doi.org/10.1103/PhysRevLett.112.255501}.

\bibitem{Arbell2002}
{\sc H.~Arbell and J.~Fineberg}, {\em Pattern formation in two-frequency forced
  parametric waves}, Phys. Rev. E, 65 (2002), p.~036224,
  \url{https://doi.org/10.1103/PhysRevE.65.036224}.

\bibitem{Archer2021}
{\sc A.~J. Archer, T.~Dotera, and A.~M. Rucklidge}, {\em Rectangle--triangle
  soft-matter quasicrystals with hexagonal symmetry}, in preparation,  (2021).

\bibitem{Archer2013}
{\sc A.~J. Archer, A.~M. Rucklidge, and E.~Knobloch}, {\em Quasicrystalline
  order and a crystal-liquid state in a soft-core fluid}, Phys. Rev. Lett., 111
  (2013), p.~165501, \url{https://doi.org/10.1103/PhysRevLett.111.165501}.

\bibitem{Argentina2012}
{\sc M.~Argentina and G.~Iooss}, {\em Quasipatterns in a parametrically forced
  horizontal fluid film}, Physica D, 241 (2012), pp.~1306--1321,
  \url{https://doi.org/10.1016/j.physd.2012.04.011}.

\bibitem{Aumann2002}
{\sc A.~Aumann, T.~Ackemann, E.~G. Westhoff, and W.~Lange}, {\em Eightfold
  quasipatterns in an optical pattern-forming system}, Phys. Rev. E, 66 (2002),
  p.~046220, \url{https://doi.org/10.1103/PhysRevE.66.046220}.

\bibitem{Baake2012}
{\sc M.~Baake and U.~Grimm}, {\em Mathematical diffraction of aperiodic
  structures}, Chem. Soc. Rev., 41 (2012), pp.~6821--6843,
  \url{https://doi.org/https://doi.org/10.1039/C2CS35120J}.

\bibitem{Barkan2014}
{\sc K.~Barkan, M.~Engel, and R.~Lifshitz}, {\em Controlled self-assembly of
  periodic and aperiodic cluster crystals}, Phys. Rev. Lett., 113 (2014),
  p.~098304, \url{https://doi.org/10.1103/PhysRevLett.113.098304}.

\bibitem{Braaksma2019}
{\sc B.~Braaksma and G.~Iooss}, {\em Existence of bifurcating quasipatterns in
  steady {B}{\'e}nard--{R}ayleigh convection}, Arch. Ration. Mech. Anal., 231
  (2019), pp.~1917--1981, \url{https://doi.org/10.1007/s00205-018-1313-6}.

\bibitem{Braaksma2017}
{\sc B.~Braaksma, G.~Iooss, and L.~Stolovitch}, {\em Proof of quasipatterns for
  the {S}wift--{H}ohenberg equation}, Commun. Math. Phys., 353 (2017),
  pp.~37--67, \url{https://doi.org/10.1007/s00220-017-2878-x}.

\bibitem{Carr1981}
{\sc J.~Carr}, {\em Applications of Centre Manifold Theory}, Springer, New
  York, 1981.

\bibitem{Castelino2020}
{\sc J.~K. Castelino, D.~J. Ratliff, A.~M. Rucklidge, P.~Subramanian, and C.~M.
  Topaz}, {\em Spatiotemporal chaos and quasipatterns in coupled
  reaction--diffusion systems}, Physica D, 407 (2020), p.~132475,
  \url{https://doi.org/10.1016/j.physd.2020.132475}.

\bibitem{Crawford1994}
{\sc J.~D. Crawford}, {\em ${D}_4 \dot{+} {T}^2$ mode interactions and hidden
  rotational symmetry}, Nonlinearity, 7 (1994), pp.~697--739,
  \url{https://doi.org/10.1088/0951-7715/7/3/002}.

\bibitem{Dawes2003}
{\sc J.~H.~P. Dawes, P.~C. Matthews, and A.~M. Rucklidge}, {\em Reducible
  actions of ${D}_4\sdp{T}^2$: superlattice patterns and hidden symmetries},
  Nonlinearity, 16 (2003), pp.~615--645,
  \url{https://doi.org/10.1088/0951-7715/16/2/315}.

\bibitem{Dionne1997}
{\sc B.~Dionne, M.~Silber, and A.~C. Skeldon}, {\em Stability results for
  steady, spatially periodic planforms}, Nonlinearity, 10 (1997), pp.~321--353,
  \url{https://doi.org/10.1088/0951-7715/10/2/002}.

\bibitem{Dotera2017}
{\sc T.~Dotera, S.~Bekku, and P.~Ziherl}, {\em Bronze-mean hexagonal
  quasicrystal}, Nature Mat., 16 (2017), pp.~987--992,
  \url{https://doi.org/10.1038/nmat4963}.

\bibitem{Echebarria2001}
{\sc B.~Echebarria and H.~Riecke}, {\em Sideband instabilities and defects of
  quasipatterns}, Physica D, 158 (2001), pp.~45--68,
  \url{https://doi.org/10.1016/S0167-2789(01)00319-0}.

\bibitem{Edwards1994}
{\sc W.~S. Edwards and S.~Fauve}, {\em Patterns and quasi-patterns in the
  {F}araday experiment}, J. Fluid Mech., 278 (1994), pp.~123--148,
  \url{https://doi.org/10.1017/S0022112094003642}.

\bibitem{Fauve2019}
{\sc S.~Fauve and G.~Iooss}, {\em Quasipatterns versus superlattices resulting
  from the superposition of two hexagonal patterns}, Comptes Rendus Mecanique,
  347 (2019), pp.~294--304, \url{https://doi.org/10.1016/j.crme.2019.03.006}.

\bibitem{Gokce2020}
{\sc A.~G{\"o}k{\c{c}}e, S.~Coombes, and D.~Avitabile}, {\em Quasicrystal
  patterns in a neural field model}, Phys. Rev. Research, 2 (2020), p.~013234,
  \url{https://doi.org/10.1103/PhysRevResearch.2.013234}.

\bibitem{Golubitsky1988}
{\sc M.~Golubitsky, I.~Stewart, and D.~G. Schaeffer}, {\em Singularities and
  Groups in Bifurcation Theory. Volume~II}, Springer, New York, 1988.

\bibitem{Hardy1960}
{\sc G.~H. Hardy and E.~M. Wright}, {\em An Introduction to the Theory of
  Numbers}, Clarendon Press, Oxford, 4th~ed., 1960.

\bibitem{Hayashida2007}
{\sc K.~Hayashida, T.~Dotera, A.~Takano, and Y.~Matsushita}, {\em Polymeric
  quasicrystal: Mesoscopic quasicrystalline tiling in {$ABC$} star polymers},
  Phys. Rev. Lett., 98 (2007), p.~195502,
  \url{https://doi.org/10.1103/PhysRevLett.98.195502}.

\bibitem{Hoyle2006}
{\sc R.~B. Hoyle}, {\em Pattern Formation: an Introduction to Methods},
  Cambridge University Press, Cambridge, 2006.

\bibitem{Iooss2019}
{\sc G.~Iooss}, {\em Existence of quasipatterns in the superposition of two
  hexagonal patterns}, Nonlinearity, 32 (2019), pp.~3163--3187,
  \url{https://doi.org/10.1088/1361-6544/ab230a}.

\bibitem{Iooss2010}
{\sc G.~Iooss and A.~M. Rucklidge}, {\em On the existence of quasipattern
  solutions of the {S}wift--{H}ohenberg equation}, J. Nonlin. Sci., 20 (2010),
  pp.~361--394, \url{https://doi.org/10.1007/s00332-010-9063-0}.

\bibitem{Jiang2015}
{\sc K.~Jiang, J.~Tong, P.~Zhang, and A.-C. Shi}, {\em Stability of
  two-dimensional soft quasicrystals in systems with two length scales}, Phys.
  Rev. E, 92 (2015), p.~042159,
  \url{https://doi.org/10.1103/PhysRevE.92.042159}.

\bibitem{Jiang2020a}
{\sc Z.~Jiang, S.~Quan, N.~Xu, L.~He, and Y.~Ni}, {\em Growth modes of
  quasicrystals involving intermediate phases and a multistep behavior studied
  by phase field crystal model}, Phys. Rev. Materials, 4 (2020), p.~023403,
  \url{https://doi.org/10.1103/PhysRevMaterials.4.023403}.

\bibitem{Kudrolli1998}
{\sc A.~Kudrolli, B.~Pier, and J.~P. Gollub}, {\em Superlattice patterns in
  surface waves}, Physica D, 123 (1998), pp.~99--111,
  \url{https://doi.org/10.1103/10.1016/S0167-2789(98)00115-8}.

\bibitem{Lifshitz2007a}
{\sc R.~Lifshitz and H.~Diamant}, {\em Soft quasicrystals -- {W}hy are they
  stable?}, Philos. Mag., 87 (2007), pp.~3021--3030,
  \url{https://doi.org/10.1080/14786430701358673}.

\bibitem{Lifshitz1997}
{\sc R.~Lifshitz and D.~M. Petrich}, {\em Theoretical model for {F}araday waves
  with multiple-frequency forcing}, Phys. Rev. Lett., 79 (1997),
  pp.~1261--1264, \url{https://doi.org/10.1103/PhysRevLett.79.1261}.

\bibitem{Malomed1989}
{\sc B.~A. Malomed, A.~A. Nepomnyashchi\u\i, and M.~I. Tribelski\u\i}, {\em
  Two-dimensional quasiperiodic structures in nonequilibrium systems}, Sov.
  Phys. JETP, 69 (1989), pp.~388--396.

\bibitem{Matthews2003a}
{\sc P.~C. Matthews}, {\em Transcritical bifurcation with ${O}(3)$ symmetry},
  Nonlinearity, 16 (2003), pp.~1449--1471,
  \url{https://doi.org/10.1088/0951-7715/16/4/315}.

\bibitem{Mermin1985}
{\sc N.~D. Mermin and S.~M. Troian}, {\em Mean-field theory of quasicrystalline
  order}, Phys. Rev. Lett., 54 (1985), pp.~1524--1527,
  \url{https://doi.org/10.1103/PhysRevLett.54.1524}.

\bibitem{Muller1994}
{\sc H.~W. M\"uller}, {\em Model equations for two-dimensional quasipatterns},
  Phys. Rev. E, 49 (1994), pp.~1273--1277,
  \url{https://doi.org/10.1103/PhysRevE.49.1273}.

\bibitem{Nakakura2019}
{\sc J.~Nakakura, P.~Ziherl, J.~Matsuzawa, and T.~Dotera}, {\em Metallic-mean
  quasicrystals as aperiodic approximants of periodic crystals}, Nature Comm.,
  10 (2019), pp.~1--8, \url{https://doi.org/10.1038/s41467-019-12147-z}.

\bibitem{Newell1993}
{\sc A.~C. Newell and Y.~Pomeau}, {\em Turbulent crystals in macroscopic
  systems}, J. Phys. A, 26 (1993), pp.~L429--L434,
  \url{https://doi.org/10.1088/0305-4470/26/8/006}.

\bibitem{Porter2004}
{\sc J.~Porter, C.~M. Topaz, and M.~Silber}, {\em Pattern control via
  multifrequency parametric forcing}, Phys. Rev. Lett., 93 (2004), p.~034502,
  \url{https://doi.org/10.1103/PhysRevLett.93.034502}.

\bibitem{Ratliff2019}
{\sc D.~J. Ratliff, A.~J. Archer, P.~Subramanian, and A.~M. Rucklidge}, {\em
  Which wave numbers determine the thermodynamic stability of soft matter
  quasicrystals?}, Phys. Rev. Lett., 123 (2019), p.~148004,
  \url{https://doi.org/10.1103/PhysRevLett.123.148004}.

\bibitem{Rucklidge2003}
{\sc A.~M. Rucklidge and W.~J. Rucklidge}, {\em Convergence properties of the
  8, 10 and 12 mode representations of quasipatterns}, Physica D, 178 (2003),
  pp.~62--82, \url{https://doi.org/10.1016/S0167-2789(02)00792-3}.

\bibitem{Rucklidge2009}
{\sc A.~M. Rucklidge and M.~Silber}, {\em Design of parametrically forced
  patterns and quasipatterns}, SIAM J. Appl. Dynam. Syst., 8 (2009),
  pp.~298--347, \url{https://doi.org/10.1137/080719066}.

\bibitem{Rucklidge2012}
{\sc A.~M. Rucklidge, M.~Silber, and A.~C. Skeldon}, {\em Three-wave
  interactions and spatiotemporal chaos}, Phys. Rev. Lett., 108 (2012),
  p.~074504, \url{https://doi.org/10.1103/PhysRevLett.108.074504}.

\bibitem{Savitz2018}
{\sc S.~Savitz, M.~Babadi, and R.~Lifshitz}, {\em Multiple-scale structures:
  from {F}araday waves to soft-matter quasicrystals}, IUCrJ, 5 (2018),
  pp.~247--268, \url{https://doi.org/10.1107/S2052252518001161}.

\bibitem{Shechtman1984a}
{\sc D.~Shechtman, I.~Blech, D.~Gratias, and J.~W. Cahn}, {\em Metallic phase
  with long-range orientational order and no translational symmetry}, Phys.
  Rev. Lett., 53 (1984), pp.~1951--1953,
  \url{https://doi.org/10.1103/PhysRevLett.53.1951}.

\bibitem{Silber1998}
{\sc M.~Silber and M.~R.~E. Proctor}, {\em Nonlinear competition between small
  and large hexagonal patterns}, Phys. Rev. Lett., 81 (1998), pp.~2450--2453,
  \url{https://doi.org/10.1103/PhysRevLett.81.2450}.

\bibitem{Skeldon2007}
{\sc A.~C. Skeldon and G.~Guidoboni}, {\em Pattern selection for {F}araday
  waves in an incompressible viscous fluid}, SIAM J. Appl. Math., 67 (2007),
  pp.~1064--1100, \url{https://doi.org/10.1137/050639223}.

\bibitem{Skeldon2015}
{\sc A.~C. Skeldon and A.~M. Rucklidge}, {\em Can weakly nonlinear theory
  explain {F}araday wave patterns near onset?}, J. Fluid Mech., 777 (2015),
  pp.~604--632, \url{https://doi.org/10.1017/jfm.2015.388}.

\bibitem{Subramanian2016}
{\sc P.~Subramanian, A.~J. Archer, E.~Knobloch, and A.~M. Rucklidge}, {\em
  Three-dimensional icosahedral phase field quasicrystal}, Phys. Rev. Lett.,
  117 (2016), p.~075501, \url{https://doi.org/10.1103/PhysRevLett.117.075501}.

\bibitem{Subramanian2021b}
{\sc P.~Subramanian, I.~G. Kevrekidis, and P.~G. Kevrekidis}, {\em Exploring
  critical points of energy landscapes: From low-dimensional examples to phase
  field crystal pdes}, Commun. Nonlinear Sci. Numer. Simulat., 96 (2021),
  p.~105679, \url{https://doi.org/10.1016/j.cnsns.2020.105679}.

\bibitem{Swift1977}
{\sc J.~Swift and P.~C. Hohenberg}, {\em Hydrodynamic fluctuations at the
  convective instability}, Phys. Rev. A, 15 (1977), pp.~319--328,
  \url{https://doi.org/10.1103/PhysRevA.15.319}.

\bibitem{Vanderbauwhede1992}
{\sc A.~Vanderbauwhede and G.~Iooss}, {\em Center manifold theory in infinite
  dimensions}, in Dynamics Reported: Expositions in Dynamical Systems (New
  Series), C.~Jones, K.~U., and H.~O. Walther, eds., vol.~1, Springer, Berlin,
  1992, pp.~125--163.

\bibitem{Yankowitz2019}
{\sc M.~Yankowitz, S.~Chen, H.~Polshyn, Y.~Zhang, K.~Watanabe, T.~Taniguchi,
  D.~Graf, A.~F. Young, and C.~R. Dean}, {\em Tuning superconductivity in
  twisted bilayer graphene}, Science, 363 (2019), pp.~1059--1064,
  \url{https://doi.org/10.1126/science.aav1910}.

\bibitem{Zeller2014}
{\sc P.~Zeller and S.~G{\"u}nther}, {\em What are the possible moir{\'e}
  patterns of graphene on hexagonally packed surfaces? {U}niversal solution for
  hexagonal coincidence lattices, derived by a geometric construction}, New J.
  Phys., 16 (2014), p.~083028,
  \url{https://doi.org/10.1088/1367-2630/16/8/083028}.

\bibitem{Zeng2004}
{\sc X.~B. Zeng, G.~Ungar, Y.~S. Liu, V.~Percec, S.~E. Dulcey, and J.~K.
  Hobbs}, {\em Supramolecular dendritic liquid quasicrystals}, Nature, 428
  (2004), pp.~157--160, \url{https://doi.org/10.1038/nature02368}.

\bibitem{Zhang1996}
{\sc W.~B. Zhang and J.~Vi\~nals}, {\em Square patterns and quasipatterns in
  weakly damped {F}araday waves}, Phys. Rev. E, 53 (1996), pp.~R4283--R4286,
  \url{https://doi.org/10.1103/PhysRevE.53.R4283}.

\bibitem{Zhang1997a}
{\sc W.~B. Zhang and J.~Vi\~nals}, {\em Pattern formation in weakly damped
  parametric surface waves driven by two frequency components}, J. Fluid Mech.,
  341 (1997), pp.~225--244, \url{https://doi.org/10.1017/S0022112097005387}.

\end{thebibliography}

\end{document}